%% file: main.tex
\pdfoutput=1
\documentclass{article}

\usepackage{csquotes}
\usepackage[english]{babel}
\usepackage{url}
\usepackage{stmaryrd}
\usepackage{amsthm,amstext,amssymb}
\usepackage{epsfig}
\usepackage{slashed}
\usepackage{fancyhdr}
\usepackage{graphicx}
\usepackage{pdfpages}
\usepackage{float}
\usepackage{subfig}
\usepackage{simplewick}
\usepackage{graphics,psfrag}
\usepackage{pstricks}
\usepackage{latexsym}
\usepackage{changepage}
\usepackage{caption}
\usepackage{paralist}
\usepackage{tikzit}
\usepackage{qcircuit}
\usepackage{tabularx}

\input{style/zx.tikzstyles}

\input{style/circuits.tikzstyles}

\usepackage{titlesec}
\titleformat*{\paragraph}
  {\normalfont\normalsize\itshape}

\usepackage{graphicx}
\usepackage{array}
\usepackage{hyperref}
\usepackage{tikz-cd}
\usepackage{mathrsfs}
\usepackage{braket}
\usepackage{soul}

\usepackage{scalefnt}
\usepackage[all]{xypic}
\usepackage{slashed}
\usepackage{extarrows}
\usepackage{lettrine}
\usepackage{booktabs}
\usepackage{paralist}
\usepackage{subfig}
\usepackage{wrapfig}
\usepackage{xcolor}

\usepackage{mathtools}

\usepackage{verbatim}
\usepackage{stmaryrd}

\usepackage{multirow} 

\DeclareMathAlphabet{\mathcal}{OMS}{cmsy}{m}{n}
\newtheorem{lemma}{Lemma}[section] 
\newtheorem{proposition}[lemma]{Proposition}
\newtheorem{corollary}[lemma]{Corollary}
\newtheorem{example}[lemma]{Example}
\newtheorem{theorem}[lemma]{Theorem}

\newtheorem{definition}[lemma]{Definition}
\newtheorem{remark}[lemma]{Remark}

\newcommand{\C}{\mathbb{C}}

\newcommand{\F}{\mathbb{F}}
\newcommand{\Z}{\mathbb{Z}}
\newcommand{\N}{\mathbb{N}}

\newcommand{\Vect}{\mathtt{Vect}}

\newcommand{\del}{\partial}
\newcommand{\bul}{\bullet}
\newcommand{\tens}{\mathop{{\otimes}}}

\DeclareMathOperator{\im}{im}
\DeclareMathOperator{\supp}{supp}
\DeclareMathOperator{\Span}{span}

\newcommand{\MatF}{\mathtt{Mat}_{\F_2}}
\newcommand{\Chains}{\mathtt{Ch}(\MatF)}
\newcommand{\Cochains}{\mathtt{coCh}(\MatF)}

\newcommand{\FHilb}{\mathtt{FHilb}}

\definecolor{zx_red}{RGB}{232, 165, 165}
\definecolor{zx_green}{RGB}{216, 248, 216}

\newsavebox{\pullback}
\sbox\pullback{%
\begin{tikzpicture}%
\draw (0,0) -- (2ex,0ex);%
\draw (2ex,0ex) -- (2ex,2ex);%
\end{tikzpicture}}

\newsavebox{\pushout}
\sbox\pushout{%
\begin{tikzpicture}%
\draw (0,0) -- (0,2ex);%
\draw (0,2ex) -- (2ex,2ex);%
\end{tikzpicture}}

\usepackage[letterpaper,top=2cm,bottom=2cm,left=3cm,right=3cm,marginparwidth=1.75cm]{geometry}

\usepackage[style=alphabetic, maxnames=10]{biblatex}
\usepackage{amsmath}
\allowdisplaybreaks
\usepackage{cleveref}
\usepackage{graphicx}

\usepackage{authblk}

\title{Engineering CSS surgery: compiling any CNOT in any code}    
\author[1,2,3]{Clément Poirson}
\author[3]{Joschka Roffe}
\author[3,4,5]{Robert I. Booth}
\affil[1]{
  Alice \& Bob, 49 Bd du Général Martial Valin, 75015 Paris, France
}
\affil[2]{
  Inria Paris, France
}
\affil[3]{
  University of Edinburgh, United Kingdom
}
\affil[4]{
  University of Oxford, United Kingdom
}
\affil[5]{
  University of Bristol, United Kingdom
}

\date{}

\begin{document}

\maketitle

\vspace{-1cm}
\begin{abstract}
We introduce a framework for implementing logic in CSS quantum error correction codes, building on the surgery methods of Cowtan and Burton \cite{Cowtan_2024}. Our approach offers a systematic methodology for designing and analysing surgery protocols. At the physical level, we introduce the concept of subcodes, which encapsulate all the necessary data for performing surgery. At the logical level, leveraging homological algebra, subcodes enable us to track the logical operations induced by any surgery protocol, regardless of the choice of logical operator basis. In particular, we make no assumptions on the structure of the logical operators of the code. As a proof of concept, we develop a surgery protocol inspired by lattice surgery that implements a logical CNOT gate between any two logical qubits of any CSS code, with fault-tolerance guarantees. 
\end{abstract}

\vspace{5mm}

Although quantum computing holds great promise for a wide range of applications, current hardware is severely impacted by noise. This noise leads to the accumulation of errors throughout a computation, producing random outcomes from what should be deterministic processes. To mitigate this, quantum error correction (QEC) plays a key role. By encoding information in logical qubits, QEC aims at striking a balance between the number of extra qubits used and the system's resistance to noise. 
Among the various families of quantum error-correcting codes, quantum low-density parity-check (qLDPC) codes have emerged as particularly promising. Theoretically, some of these codes can asymptotically optimise the trade-off between noise resistance and qubit overhead \cite{panteleev2022asymptoticallygoodquantumlocally, leverrier2022quantumtannercodes}. These are known as ``good qLDPC codes''. Practically, recent simulations of finite-size qLDPC codes have shown that they can offer the same level of protection as surface codes while using around $10\times$ fewer qubits \cite{xu2023constantoverheadfaulttolerantquantumcomputation, Bravyi_2024,scruby2024highthresholdlowoverheadsingleshotdecodable,malcolm2025computingefficientlyqldpccodes}.
The recent advances in qLDPC codes have largely been enabled by the study of Calderbank-Steane-Shor (CSS) codes, a broader class of stabiliser codes that includes qLDPC codes. This research has been enriched by viewing these codes through the lens of homology. In particular, some of the most effective constructions of quantum codes based on classical ones---known as product constructions---are derived from homological algebra \cite{Tillich_2014, Breuckmann_2021, Panteleev_2022, audoux2018tensorproductscsscodes}.

While classical error correction is concerned with storing information, quantum error correction must address the additional challenge of enabling computations on encoded data. To achieve this goal, we must find a set of gates, known as `universal gate set', which can implement all possible computations on the logical data. These gates should be performed efficiently, without compromising the system's resistance to errors. Identifying quantum codes that support fault-tolerant computation has proven to be a significant and ongoing challenge \cite{Campbell_2017, Webster_2022}.

Now that good qLDPC codes have been discovered \cite{panteleev2022asymptoticallygoodquantumlocally, leverrier2022quantumtannercodes}, the core challenge is to go from good quantum memories to fault-tolerant quantum computers. The gold standard for implementing fault-tolerant gates in qLDPC codes is through transversal gates, which are inherently fault-tolerant \cite{Quintavalle_2023, Breuckmann_2024,sayginel2024faulttolerantlogicalcliffordgates, malcolm2025computingefficientlyqldpccodes}. However, the Eastin-Knill theorem prohibits the existence of a fully universal transversal gate set \cite{Eastin_2009}, necessitating the development of alternative methods for fault-tolerant gate construction. Quantum code surgery is one such method and the focus of this work. 

To date, quantum error correction experiments have focused predominantly on implementing surface code logical qubits \cite{acharya2024quantumerrorcorrectionsurface}. The surface architecture is favoured due to straightforward implementation on a 2D-local array of qubits as well as its well-established, fault-tolerant universal gate set composed of the Clifford gate set $\{CNOT, H, S\}$ plus the $T$ gate \cite{Litinski_2019}. Fault-tolerant CNOT gates are realised via \textit{lattice surgery} operations between surface code patches, as first proposed in \cite{Horsman_2012}.
Initially developed to implement a CNOT between two patches of surface codes, lattice surgery has since been shown to be more versatile, capable of performing any CSS operation\footnote{\label{footnote:css_operation}CSS operations correspond to phase-free ZX-diagram i.e., composition of gates within $\{\bra{0}, \bra{+}, CNOT, \ket{0}, \ket{+}\}$.} \cite{de_Beaudrap_2020}. The lattice surgery-based protocol for implementing a CNOT consists of six steps, illustrated in \cref{fig:lattice_surgery}.  It relies on the decomposition of the CNOT \cref{eq:cnot_ls} into two joint logical operator measurements mediated by an auxiliary qubit. 

The main goal of this work is to generalise lattice surgery of \cite{Horsman_2012, de_Beaudrap_2020} to any CSS code. Such a generalisation has been an active topic of research, and a number of approaches have been proposed \cite{cohen2022lowoverheadfaulttolerantquantumcomputing, Cowtan_2024, cross2024improvedqldpcsurgerylogical,zhang2024timeefficientlogicaloperationsquantum, ide2024faulttolerantlogicalmeasurementshomological, williamson2024lowoverheadfaulttolerantquantumcomputation,swaroop2024universaladaptersquantumldpc, he2025extractorsqldpcarchitecturesefficient, cowtan2025parallellogicalmeasurementsquantum}.

\subsection*{Prior work} 
The standard model of surgery, initiated in \cite{Horsman_2012} and adopted in works such as \cite{cohen2022lowoverheadfaulttolerantquantumcomputing}, implements logical measurements by introducing auxiliary physical qubits and performing joint Pauli measurements. This paradigm faces two central challenges:
\begin{enumerate}
    \item \textbf{Resource efficiency:} The auxiliary qubits and measured operators must be kept minimal and low-weight to ensure fault tolerance and reduce hardware overhead.
    \item \textbf{Measurement precision:} We aim to measure exactly one logical operator, which becomes difficult when the desired operator is not irreducible, i.e., when the physical Pauli that describes the logical operator contains another logical operator in its support.
\end{enumerate}

In \cite{cohen2022lowoverheadfaulttolerantquantumcomputing}, a protocol was proposed capable of measuring arbitrary logical operators (including reducible ones), at the cost of a substantial $\mathcal{O}(d^2)$ overhead in auxiliary qubits. This method was refined in \cite{cross2024improvedqldpcsurgerylogical}, which reduced the overhead for codes with good expansion at the cost of requiring an irreducible logical operator basis. Further improvements to generalised surgery have been outlined in \cite{ide2024faulttolerantlogicalmeasurementshomological} and \cite{williamson2024lowoverheadfaulttolerantquantumcomputation}, both of which relax the irreducibility condition and do not require the code to have expansion properties. The scheme proposed in \cite{williamson2024lowoverheadfaulttolerantquantumcomputation} is particularly flexible, applying to any stabiliser code and relying only on low-weight measurements with an auxiliary qubit overhead of $\mathcal{O}(d \log(d)^3)$.

Subsequent innovations in generalised surgery have focused on 1) methods for measuring multiple logical operators simultaneously, and 2) the modular integration of different surgery schemes in order to makes it possible to seamlessly combine them when compiling logical measurements fault-tolerantly. In particular, it is preferable to avoid the exponential blow-up in the number of schemes otherwise required to implement every possible Pauli measurement. Addressing the first problem, \cite{zhang2024timeefficientlogicaloperationsquantum} proposed a scheme for time-efficient simultaneous measurements, which was later optimised in terms of qubit overhead in \cite{cowtan2025parallellogicalmeasurementsquantum}. For the second, \cite{swaroop2024universaladaptersquantumldpc} and \cite{he2025extractorsqldpcarchitecturesefficient} developed universal adapter and extractor protocols that extend to arbitrary stabiliser codes. A more detailed summary of prior works following this approach to surgery can be found in section 3.2 of \cite{he2025extractorsqldpcarchitecturesefficient}.

In parallel to these Pauli-measurement-based approaches, an alternative line of work---initiated by \cite{Cowtan_2024} and developed further in \cite{cowtan2024ssipautomatedsurgeryquantum}---proposed a fundamentally different formulation. This approach formulates surgery without the use of joint operator measurements, by representing CSS codes as chain complexes \emph{instead of Tanner graphs}, and transformations of CSS codes as chain maps\footnote{The problem with representing surgery of CSS codes in terms of gluing and cutting Tanner graphs, is that Tanner graphs forget the linear structure of the stabiliser group since they prioritise tracking the parity checks of the code.}. 
Notably, this approach allows for richer physical-level operations, including CNOT gates between physical qubits in addition to Pauli measurements. However, earlier versions of this formalism still required logical operators to be irreducible in order to make the resulting logical action explicit.

\subsection*{Our contributions}

In this work, we build on the CSS surgery formulation of Cowtan and Burton \cite{Cowtan_2024}. We address a key limitation of \cite{Cowtan_2024}, namely the difficulty of describing the induced logical operation in full generality. This enables us to perform physical operations going beyond Pauli measurements, and, at the logical level, we are no longer restricted to implementing logical gates via joint logical operator measurements of Pauli operators. The chain complex picture can describe operations between CSS codes that correspond to arbitrary phase-free ZX-diagrams \cite{kissinger2022phasefree} (or CSS operations, see footnote~\ref{footnote:css_operation}), at both the logical and physical levels. In this paper, our main technical contributions are two-fold:
\begin{enumerate}
    \item \textbf{Subcodes}: We reformulate the surgery procedure, originally given via a categorical construction known as a \emph{pushout} in \cite{Cowtan_2024}, as a \emph{quotient of a CSS code by a subcode}. This is a very natural approach from the perspective of homological algebra and allows us to easily determine what merges of two codes are possible. Many constructions that appear in the literature of surgery can be pictured as subcodes, such as the logical operator subcomplexes of \cite{Cowtan_2024}, and also the classical codes derived from the support of a logical operator of \cite{cohen2022lowoverheadfaulttolerantquantumcomputing,cross2024improvedqldpcsurgerylogical} (that inspired \cite{Cowtan_2024}).
    \item \textbf{Logical operations and irreducibility}: By combining our subcode reformulation of merges with standard tools of homological algebra, we can compute the logical operators of the merged code. We show that this approach is in principle completely agnostic to the problem of irreducibility, providing a means to compile logic without first requiring a special basis. More precisely, we obtain a \emph{long exact sequence} that simultaneously characterises the logical operators of the merged code and the logical operation induced by the merge. 
\end{enumerate}

The main application we focus on in this work is the compilation of CNOT gates in arbitrary CSS codes. By carefully picking the sequence of merges and splits, we show that our framework can be used to compile CNOTs between any two logical qubits whilst ensuring fault-tolerance guarantees at each step.
We emphasise that our formalism also allows for more sophisticated logical operations tailored for specific uses. For instance, the CNOT protocol can be adapted to perform an array of CNOTs. In addition to designing logical CNOTs, we can add and measure logical qubits in given states. 
Our framework can also be used to design code-switching operations between CSS codes as instances of CSS surgery. For example, in appendix~\ref{appendix:code_switching}, we replicate the code-switching protocol presented in \cite{heußen2024efficientfaulttolerantcodeswitching} between the 3D colour code and the Steane code \cite{Anderson_2014}.

However, it cannot describe logical or physical operations that live outside the set of CSS operations, such as Hadamard, \(S\), or \(T\) gates. Furthermore, while our formalism can yield fault-tolerant operations, it does not inherently ensure that the physical implementation of a surgery protocol will be fault-tolerant. Rather, it provides a flexible toolkit for realising any desired logical phase-free ZX operation. Identifying concrete constraints on the construction that guarantee fault-tolerance remains an open question.

\Cref{sec:lattice_surgery} contains a refresher on lattice surgery as originally formulated by Horsman et al. \cite{Horsman_2012}. In \cref{sec:math_prelem}, we review the symplectic representation of the Pauli group, before describing the representation of CSS codes as (co)chain complexes and CSS operations as (co)chain maps. In \Cref{sec:css_surgery}, we give the definition of subcodes, merges and splits---a reformulation of the pushout construction of \cite{Cowtan_2024}---and we provide the long exact sequence that enables computation of the logical operation induced by any surgery procedure. Finally, \cref{sec:ft_cnot} describes our protocol for designing a CNOT gate on arbitrary CSS codes, and presents its fault-tolerance guarantees.

\tableofcontents

\section{A refresher on lattice surgery}
\label{sec:lattice_surgery}

\begin{figure}[p]
    \centering
    \includegraphics[width=0.85\linewidth]{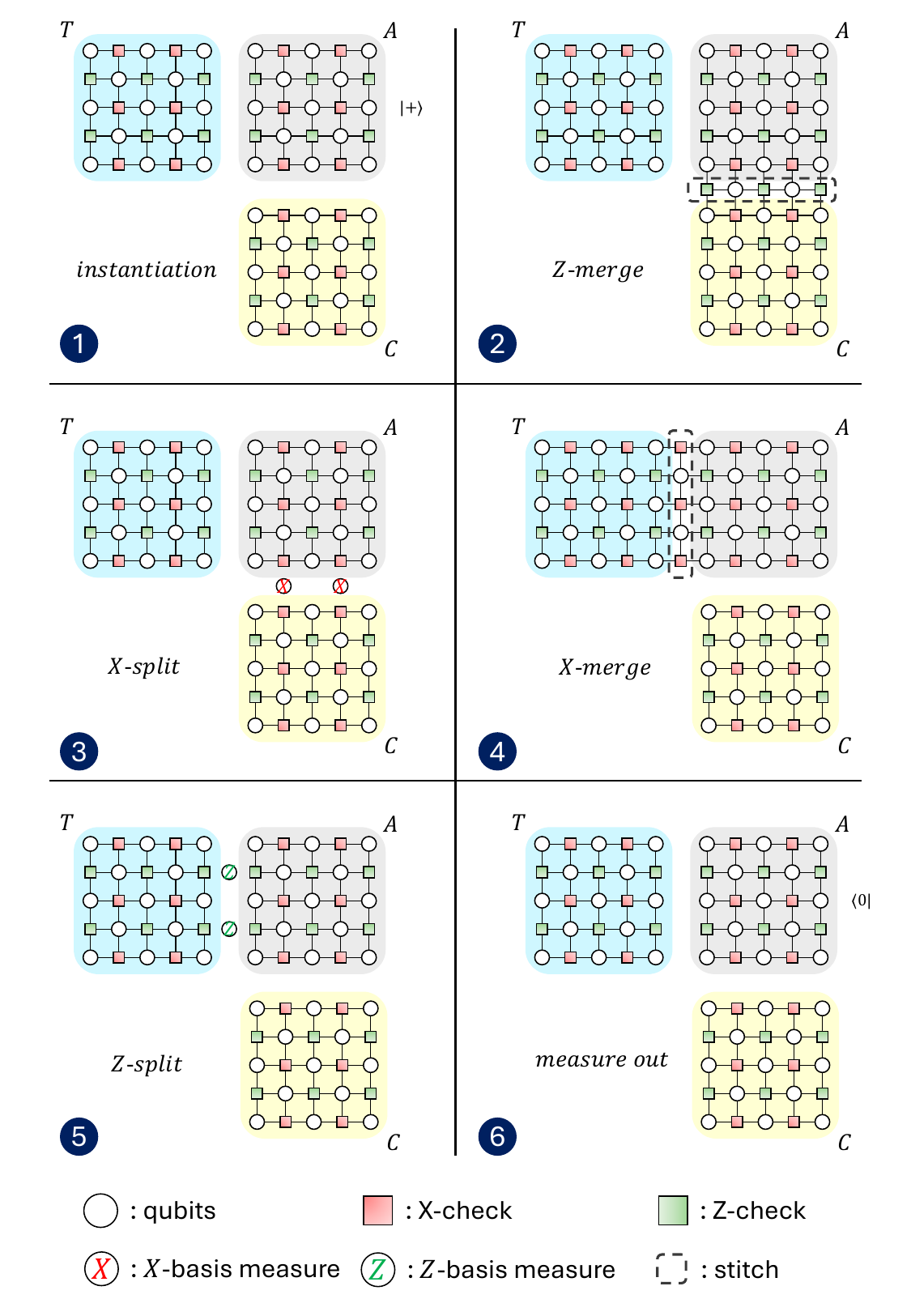}
    \caption{\textbf{Implementation of the $CNOT$ gate using lattice surgery.} 
    The process involves merging and splitting logical qubits $C$ (control) and $T$ (target) with a auxiliary surface $A$. For simplicity, measurement outcomes will be assumed to be positive. In step 1, an auxiliary logical qubit encoded in $A$ is instantiated in the $\ket{+}$ state. In step 2, two auxiliary physical qubits are introduced in the $\ket{+}$, while three Z-Pauli measurements introduce three new Z-stabilisers. The two auxiliary qubits are then measured out in the X-basis in step 3. Step 2 and 3 together implement an $M_{ZZ}$ measurement between the logical qubits $C$ and $A$. Similarly, step 4 and 5 follow the same scheme but swapping the roles of X and Z. It implements an $M_{XX}$ measurement. Finally, the logical qubit in $A$ is measured out in the Z-basis. These measurements outcomes are non-deterministic, but corrections can be applied to ensure that the logical effect corresponds to a $CNOT$ gate.
    }
    \label{fig:lattice_surgery}
\end{figure}

In order to construct a fault-tolerant quantum computer, designing codes with
good memory properties is only half the battle. In order to actually
run a quantum computation, we also need to be able to implement logical
gates fault-tolerantly. This is a non-trivial task. An \textbf{encoder}
for an \([[n,k,d]]\) code is an isometry \(E : \C^{2^k} \to \C^{2^n}\),
and a \textbf{logical operation} is a map \(M : \C^{2^k} \to \C^{2^\ell}\)
that acts on the \(k\) logical qubits encoded by \(E\).  A \textbf{physical
implementation} of \(M\) for the encoder \(E\) is an operation \(\overline{M}
: \C^{2^m} \to \C^{2^n}\) which acts on the physical qubits and such that
\begin{equation}
  \label{eq:encoder_equation}
  \input{encoder_equation.tikz}
\end{equation}
In many cases, \(E = E'\), but this does not have to be the case. As we are
about to describe, it is \emph{impossible} for lattice surgery. 

In the case of the surface code, lattice surgery provides techniques for the
fault-tolerant implementation of CNOT gates \cite{Horsman_2012}. The basic idea of
lattice surgery is to decompose the CNOT gate, controlled on a qubit \(C\)
and targetting a qubit \(T\), as the following sequence of operations:
\begin{equation} \label{eq:cnot_ls}
  \input{cnot.tikz}
\end{equation}
where we make use of an additional auxiliary qubit \(A\) which is initialised
in the \(\ket{+}\) state, and where we have furthermore used the ZX-calculus
notation for the linear maps given by
\begin{equation}
  \label{eq:logical_surgery}
  \begin{aligned}
    &\input{logical_merge_green.tikz} \coloneqq \ket{0}\bra{00} + \ket{1}\bra{11},
    \quad &\input{logical_merge_red.tikz} \coloneqq \ket{+}\bra{++} + \ket{-}\bra{--},
    \\
    &\input{logical_split_green.tikz} \coloneqq \ket{00}\bra{0} + \ket{11}\bra{1},\text{ and}
    \quad &\input{logical_split_red.tikz} \coloneqq \ket{++}\bra{+} + \ket{--}\bra{-}.
  \end{aligned}
\end{equation}
We emphasise two important points. First, there are many different
decompositions of a CNOT gate using this family of linear maps, and
these can be used to perform different kinds of lattice surgery on the
surface code. Second, the attentive reader will have noticed that the
operations described in equation~\eqref{eq:logical_surgery} are \emph{not
isometries}. This means that they cannot in principle be implemented
deterministically on a quantum computer. Nevertheless, they can be implemented
up to a \emph{probablistic} Pauli error, and this Pauli error is relatively
easy to account for in simple circuits like a CNOT gate \cite{de_Beaudrap_2020}.

Specialising \cref{eq:encoder_equation} to the case of lattice
surgery, we now need to find physical implementations of each of the operations of \cref{eq:logical_surgery} for surface code encoders. These
physical implementations are called \emph{merges} and \emph{splits},
and are depicted in terms of the Tanner graphs of the surface codes in \cref{fig:lattice_surgery}:
\begin{align}
  \label{eq:physical_surgery}
  \input{physical_Z-merge.tikz} \quad &\text{is implemented by step 2;} \\
  \input{physical_X-split.tikz} \quad &\text{is implemented by step 3;} \\
  \input{physical_X-merge.tikz} \quad &\text{is implemented by step 4; and,} \label{eq:physical_surgery_X} \\
  \input{physical_Z-split.tikz} \quad &\text{is implemented by step 5.}
\end{align}
In each step, \(E\) and \(E'\) are encoders for surface codes of varying
shapes and sizes. Finally, step 6 implements the measurement at the end of
the circuit. An overview of the proof of the logical operation realised by merges and splits can be found in \cref{app:lattice_surgery}. Otherwise, the complete proof is in \cite{de_Beaudrap_2020}.

The goal of this article is to show how these operations can be generalised to the case where the encoders \(E,E'\) are encoders for arbitrary CSS codes.
\begin{remark}
    We make the distinction between ``lattice surgery'',
    thought of as the calculus of merges and splits on surface codes, and
    the CNOT protocol implemented using lattice surgery. In fact, lattice
    surgery doesn't reduce to only performing CNOT, even though this is the
    main application.
\end{remark}

\section{Mathematical preliminaries} \label{sec:math_prelem}

The Pauli group \(\mathcal{P}_n\) acting on \(n\) qubits is the
group of unitary operators on \(\C^{2^n}\) generated by the Pauli \(X\), \(Y\)
and \(Z\) operators acting on each individual qubit. Using the two relations \(Y =
iXZ\) and \(XZ = -ZX\), any Pauli operator takes the form \(i^k \prod_{j = 1}^n
X_j^{x_j} Z_j^{z_j}\) for some \(k \in \Z/4\Z\) and \(x_j,z_j \in \F_2\), where
the indices \(j\) indicate on which qubit the Pauli \(X\) and \(Z\) gates act.

A \textbf{stabiliser group} on \(n\) qubits is an abelian subgroup \(\mathcal{S}\)
of \(\mathcal{P}_n\) such that \(-\mathbb{I} \notin \mathcal{S}\). It is often convenient to present a stabiliser group \(\mathcal{S}\) by giving its generators: a collection \(G \subseteq \mathcal{P}_n\) of
Pauli operators such that the group generated by \(G\) is \(\langle G \rangle =
\mathcal{S}\). A \textbf{stabiliser code} is a sub-Hilbert space of \(\C^{2^n}\) defined by a stabiliser group \(\mathcal{S}\) as
\begin{equation}
  \left\{\ket{\psi} \in \C^{2^n} : \forall g \in \mathcal{S}, \ g\ket{\psi}=\ket{\psi} \right\}.
\end{equation}
We can alternatively define it as the image of the projector $P_\mathcal{S} \coloneqq \prod_{S \in G} \frac{\mathbb{I} + S}{2}$. This subspace is also called the ``codespace''.
A \textbf{Calderbank-Shor-Steane (CSS) code} is a stabiliser code for which we can find generators of the stabiliser group that are either a product of only Pauli \(X\) gates or of
only Pauli \(Z\) gates.

The Hilbert-space formulation of the stabiliser theory is convenient since it is operationally close to its implementation in terms of operators on
actual quantum hardware. On the other hand, it offers relatively few tools
for analysing computation in stabiliser codes. The remainder
of this section is dedicated to presenting mathematical reformulations of
the stabiliser theory which yield such analytical tools: the symplectic
perspective on the Pauli group \cite{gross_finite_2005,Rengaswamy_2018}, the chain complex formulation of
CSS codes, and the homological framework that derives from it \cite{Breuckmann_2021,audoux2018tensorproductscsscodes}.

\subsection{The symplectic representation of the Pauli group}

As discussed above, any Pauli operator \(i^k \prod_{j = 1}^n X_j^{x_j}
Z_j^{z_j}\) can be parametrised by the elements \(k \in \Z/4\Z\) and
\(x_j,z_j \in \F_2\). Within this parametrisation lies a group isomorphism
\cite{Rengaswamy_2018, gross_finite_2005}:
\begin{equation}
  \label{gamma}
  \begin{aligned}
      \gamma: \F_2^n \oplus \F_2^n &\longrightarrow \mathcal{P}_n / \langle i\mathbb{I} \rangle^{\otimes n} \\
       \left(x|z\right) &\longmapsto \prod_{j = 1}^n X_j^{x_j} Z_j^{z_j}
   \end{aligned}
\end{equation}
This provides a simple way to tell if two Pauli operators commute
or anticommute: given \((x_1|z_1), (x_2|z_2) \in \F_2^n \oplus \F_2^n\),
the Paulis \(\gamma(x_1|z_1)\) and \(\gamma(x_2|z_2)\) commute if and only if
\begin{equation}
  \omega((x_1|z_1), (x_2|z_2)) = (x_1|z_1) \Omega (x_2|z_2)^T = x_1z_2^T + z_1x_2^T = 0
  \quad \text{where} \quad
  \Omega = \begin{pmatrix}
      0 & \mathbb{I}_n \\
      \mathbb{I}_n & 0
    \end{pmatrix}.
\end{equation}
$\omega(\cdot,\cdot)$ satisfies the definition of a symplectic product
\cite{gross_finite_2005}, so $\left( \F_2^n \oplus \F_2^n, \omega \right)$
is called a \textbf{symplectic space}.

Using the symplectic representation, we can represent a stabiliser group using
a \textbf{parity check matrix} $H = (H_X | H_Z)$. Given a generating set, each
row of \(H\) corresponds to the symplectic representation of a generator.
The parity check matrix needs to represent an abelian group of Paulis to be
valid. This is equivalent to verifying that $H \Omega H^T = 0$.

Let us clarify why we quotient by \(\langle i\mathbb{I} \rangle^{\otimes n}\)
in \cref{gamma}. In the context of quantum error correction, our main
concern is the manipulation of stabiliser codes and logical operators. Since
any non trivial code cannot contain \(-\mathbb{I}_n\) as a stabiliser, it
should also exclude any element of \(\sqrt{\langle -\mathbb{I}_n \rangle}
= \langle i\mathbb{I}_n \rangle\). Hence, no stabiliser can have a \(i\)
or \(-i\) phase for the code to be non-trivial. Choosing this representation
of stabiliser groups in terms of parity-check matrices also requires us to
disregard the sign of any stabiliser: a stabiliser group containing a Pauli
\(s\) is effectively equivalent to one where \(s\) is replaced by \(-s\). From
an error correction standpoint, this is reasonable because \(\mathcal{S}\)
cannot contain both \(s\) and \(-s\) (otherwise \(s \cdot -s = -\mathbb{I}
\in \mathcal{S}\)), and the \(+1\) eigenspace of \(s\) is isomorphic to the
\(+1\) eigenspace of \(-s\). Since \(s\) and \(-s\) have the same non-identity
entries, the code distance remains unchanged whether \(s\) or \(-s\) is part
of the group. We conclude that the two codes have the same properties. Any non
trivial stabiliser code—up to stabiliser sign—is therefore in one-to-one
correspondence with certain subgroups of \(\mathcal{P}_n / \langle i\mathbb{I}
\rangle^{\otimes n} \cong \F_2^n \oplus \F_2^n\).

\subsection{CSS codes are chain complexes}

\subsubsection{Stabilisers and boundary maps}

By definition of CSS codes, we can find X-type stabilisers and Z-type
stabilisers that, together, generate the whole stabiliser group. This
means there exists a generating set of the stabiliser group such that the
parity-check matrix takes the form
\begin{equation}
    \label{eq:CSS_parity_check}
    H= \left(
    \begin{array}{c|c}
    P_X & 0 \\
    0 & P_Z
    \end{array}
    \right).
\end{equation}
The commutativity condition of the corresponding stabiliser group, \(H \Omega
H^T = 0\), then boils down to verifying that $P_Z^TP_X=0$. We can recast this commutation
relation as a sequence of maps
\begin{equation} \label{chain_complex}
    \begin{tikzcd}
        \F_2^{m_Z} \arrow[d, phantom, "\rotatebox{90}{$\cong$}"] & \F_2^n \arrow[d, phantom, "\rotatebox{90}{$\cong$}"] & \F_2^{m_X} \arrow[d, phantom, "\rotatebox{90}{$\cong$}"] \\
        C_2 \arrow[r, "P_Z^T"] \arrow[d, phantom, "\rotatebox{90}{$\leadsto$}"] & C_1 \arrow[r, "P_X"] \arrow[d, phantom, "\rotatebox{90}{$\leadsto$}"] & C_0 \arrow[d, phantom, "\rotatebox{90}{$\leadsto$}"]\\
        \text{Z-stabilisers} & \text{Z-operators} & \text{X-syndromes}
    \end{tikzcd}
\end{equation}
which composes to \(0\), and furthermore contains all of the data of the
CSS code.

Such a sequence is called a \textbf{chain complex}, and every parity check
matrix of a CSS code is therefore in one-to-one correspondence with a chain
complex $C_\bul$ of length 2, i.e., pairs of matrices $(\del_1,\del_2)$—called
\textbf{boundary maps}—which compose to zero, $\del_1\del_2=0$. The commutation
relation condition $P_Z^TP_X=0$ guarantees that the chain complex in
\cref{chain_complex} is valid. $m_Z$ is the number of generators of the group
of Z-stabilisers—which is not necessarily of minimum size—and similarly
for $m_X$ and X-stabilisers.

The three \textbf{complexes} \(C_2,C_1,C_0\) involved in the chain complex \eqref{chain_complex} are vector spaces equipped with a basis. To give an interpretation to the complexes \(C_2,C_1,C_0\), we must consider how the vectors in each complex behaves when composed with the boundary maps \(P_Z^T\) and \(P_X\). As we consider the rows of $P_Z$ to be a generating set of the Z-stabilisers, every Z-stabiliser can be written $(0|P_Z^Ts)$ for some $s \in \F_2^{m_Z}$. A vector $z \in \F_2^n$ represents a Z-operator, i.e. a Pauli with symplectic representation $(0|z)$. In fact, $P_Xz \in \F_2^{m_X}$ gives the corresponding X-syndrome. In particular, if this Z-operator is a Z-stabiliser, meaning that we can find $s \in \F_2^{m_Z}$ such that $z=P_Z^Ts$, we naturally get $P_Xz=0$. But some Z-operators may be in the kernel of $P_X$ without being stabilisers. These are the Z-logical operators.

\begin{remark}
    \label{rem:chain_complex_operational}
    There is a distinction to be made between a stabiliser code and the parity check matrix of this code. While a stabiliser code is uniquely defined with a stabiliser group, it is operationally more relevant to define it according to a generating set of the stabiliser group, i.e., a parity check matrix\footnote{The generating set is not necessarily of minimum size. Thus two parity check matrices with different dimensions might generate the same code.}.
    From a practical standpoint, the rows of the parity check matrix gives the measurements to extract the syndrome. Hence, even if two parity check matrices define the same stabiliser group, they might not be ``equally convenient'' for the syndrome extraction. For instance, one might give higher weight measurements to extract the syndrome. This is why the chain complexes picture focuses on the parity check matrix of CSS codes and not on the stabiliser group.
\end{remark}

\subsubsection{Z-X duality}

There is a tight relationship between the \(X\) and \(Z\) operators which define a CSS code. Exchanging \(X\) and \(Z\) operators at the level of stabilisers, physical operators \emph{and} syndromes yields an equivalent description of the \emph{same} code. A \textbf{cochain complex} of length 2 is a pair of matrices $(\del^1,\del^2)$ composing ``backwards'' to zero, $\del^2\del^1=0$. Hence, the same CSS code can be viewed as a cochain complex
$C^\bul$:
\begin{equation}
    \begin{tikzcd}
        \F_2^{m_Z} \arrow[d, phantom, "\rotatebox{90}{$\cong$}"] & \F_2^n \arrow[d, phantom, "\rotatebox{90}{$\cong$}"] & \F_2^{m_X} \arrow[d, phantom, "\rotatebox{90}{$\cong$}"] \\
        C^2 \arrow[d, phantom, "\rotatebox{90}{$\leadsto$}"]  & \arrow[l, "P_Z"'] C^1 \arrow[d, phantom, "\rotatebox{90}{$\leadsto$}"] & \arrow[l, "P_X^T"'] C^0 \arrow[d, phantom, "\rotatebox{90}{$\leadsto$}"] \\
        \text{Z-syndromes} & \text{X-operators} & \text{X-stabilisers}
    \end{tikzcd}
\end{equation}
The transpose of the commutation relation condition ${(P_Z^TP_X)}^T=P_X^TP_Z=0$ guarantees the validity of the cochain complex.
The interpretation is analogous to the case of the chain complex. Every X-stabiliser can be written $\left( P_X^Ts|0 \right)$ for some $s \in \F_2^{m_X}$. A vector $x \in \F_2^n$ is an X-operator, i.e., a Pauli with symplectic representation $(x|0)$. Finally, $P_Zx \in \F_2^{m_Z}$ corresponds to the Z-syndrome of the X-operator $x$.

Although $C^\bul$ is technically redundant, it is useful to track both the chain and cochain complex to get the full picture of the code. In particular, even though \(C_1\) and \(C^1\) are both isomorphic to $\F_2^n$, they carry distinct meanings. Vectors in $C_1$ correspond to physical \(Z\) operators whereas \(C^1\) corresponds to physical \(X\) operators. We will therefore think of a CSS code as a pair of chain and cochain complexes $(C_\bul, C^\bul)$.

\subsubsection{Logical operators from homology and cohomology spaces} \label{seq:logop}

For any stabiliser code, we can obtain the space of \textbf{logical operators} as the quotient space
\begin{equation}
    \label{eq:stabiliser_logical_operators}
    L := \ker(H \Omega) / \im(H^T)
\end{equation} 
Here, \(\ker(H \Omega)=\{ y \in \F_2^{2n} : H \Omega y = 0 \}\) is the subspace corresponding
exactly to the physical operators that commute with all stabilisers of the
code. The image \(\im(H^T)\) describes all of the linear combinations of the generators, i.e., the space of all stabilisers. The quotient
in \cref{eq:stabiliser_logical_operators} gives the equivalence
classes of physical Pauli operators that cannot be detected by the code, up
to stabilisers. The equivalence classes represent logical Pauli operators,
and the elements of these equivalence classes are the physical implementations
of the corresponding logical Pauli operator.

In the particular case of CSS codes, 
\begin{equation}
  H \Omega = \left(
  \begin{array}{c|c}
    0 & P_X \\
    P_Z & 0
  \end{array}
  \right)
  \quad \text{so that} \quad
  L = \ker(H \Omega)/\im(H^T) \cong \ker(P_X)/\im(P_Z^T) \oplus \ker(P_Z)/\im(P_X^T),
\end{equation} 
where we split the Z-logical operators \(\ker(P_X)/\im(P_Z^T)\) from the
X-logical operators \(\ker(P_Z)/\im(P_X^T)\).  The quotients \(H_1(C_\bul)
= \ker(P_X)/\im(P_Z^T)\) and \(H^1(C_\bul) = \ker(P_Z)/\im(P_X^T)\) are
respectively called the (degree 1) \textbf{homology} and \textbf{cohomology}
spaces of the chain complex $C_\bul$ representing the code. Hence,
\begin{equation}
    L \cong H_1(C_\bul) \oplus H^1(C_\bul).
\end{equation}
So far, we have described how to obtain the \emph{space} of all logical 
operators of a CSS code via its (co)homology.  We have \emph{not} described how
these logical operators act on the codespace, and
this is not an inherent property of the CSS code under consideration. There
can be different encoders for the \emph{same} code, and then the same state
encoded by different encoders will yield two code states that differ by a
unitary. Therefore, in order to be able to reason about logical operators on
encoded states, we need to track some additional data on top of the (co)chain
complexes that eliminates this ambiguity. As the following proposition shows, this additional data takes the form of a choice of basis for the logical operators.

\begin{proposition}\label{prop:encoder-basis}
    Picking an encoder for a CSS code amounts to choosing an individual \(X\) and \(Z\) logical operator for each encoded qubit, i.e. an isomorphism \(\varepsilon : \F_2^k \oplus \F_2^k \xrightarrow{\varepsilon} H^1(C_\bul) \oplus H_1(C_\bul)\) that has to respect the commutation relation of Pauli operators, that is, it has to preserve the symplectic structure.
\end{proposition}

\begin{proof}
    Given an encoder $E:{(\C^{2})}^{\tens k}\to {(\C^{2})}^{\tens n}$, the isomorphism $\varepsilon$ is specified by the commutation relations between $E$ and the Pauli group $\mathcal{P}_k$, i.e. for each $x_l,z_l \in \F_2^k$, we get $EX^{x_l}Z^{z_l}=X^xZ^zE$ for some $x,z \in \F_2^n$. Then, $\varepsilon$ is defined as $\varepsilon(x_l|z_l)=([x],[z])$ and it is straightforward to verify that if two Pauli operators $(x_1|z_1)$ and $(x_2|z_2)$ commute, so will each representatives of their associated logical operators $\varepsilon(x_1|z_1)$ and $\varepsilon(x_2|z_2)$. Vice versa, given a isomorphism $\varepsilon$, using the Choi-Jamiołkowski isomorphism, there is a unique isometry $E:{(\C^{2})}^{\tens k}\to {(\C^{2})}^{\tens n}$ such that for each $x_l,z_l \in \F_2^k$ and for each representative $\varepsilon(x_l|z_l)=([x],[z]) \in H_1(C_\bul) \times H^1(C_\bul)$, we have $X^xZ^zE=EX^{x_l}Z^{z_l}$. Indeed, consider the state
    \begin{equation}
        \ket{E}=\input{ketE.tikz}
    \end{equation}
    over $n+k$ qubits. By definition of an encoder, for each generator $g$ of the stabiliser group, $I_{2^k}\tens g$ stabilises $\ket{E}$. As the CSS code is of dimension $k$, there are $n-k$ such generators. Moreover, we can get $2k$ more stabilisers by exploiting the commutation relation the encoder must satisfy with $\varepsilon$. Indeed, $\left\{X^{\delta_i}\tens X^{x_i}Z^{z_i}, Z^{\delta_i}\tens X^{x'_i}Z^{z'_i} : i\in \llbracket 1, k\rrbracket\right\}$---where $\varepsilon(\delta_i|0_k)=([x_i],[z_i])$ and $\varepsilon(0_k|\delta_i)=([x'_i],[z'_i])$---is a set of $2k$ stabilisers. Hence, $\ket{E}$ is a quantum state over $n+k$ qubits stabilised by $n+k$ Pauli operators. This implies that $\ket{E}$ is unique.
\end{proof}

\begin{remark} \label{rem:type_preserving}
    (Type-preserving encoders) For any CSS code, there exists an encoder \( E \) that preserves the type of Pauli logical operators. Specifically, for any logical operator of the form \( X^x \) or \( Z^z \), there exist vectors \( x_l, z_l \in \mathbb{F}_2^k \) such that \( X^x E = E X^{x_l} \) and \( Z^z E = E Z^{z_l} \). This property follows from the fact that CSS codes admit an encoder that can be represented by a phase-free ZX-diagram~\cite{kissinger2022phasefree}.
\end{remark}

Two bases of X- and Z-logical operators respecting the commutation relation of Pauli operators will be referred to as \textbf{dual bases}. Importantly, given a basis of X-logical operators, we can always find a unique dual basis of Z-logical operators and vice versa.

\begin{proposition}\label{prop:dual_bases}
    Consider a CSS code $(C_\bul, C^\bul)$, and a basis $\mathcal{B}=\{[z_1],
    \ldots, [z_k]\}$ of $H_1(C_\bul)$. The dual basis $\mathcal{C}=\{[x_1],
    \ldots, [x_k]\}$ of $H^1(C_\bul)$ associated to $\mathcal{B}$ is the
    unique basis such that for all $i,j \in \llbracket 1, k \rrbracket$,
    \begin{equation} \label{canonical_condition}
        [x_i] \cdot [z_j] := x_i \cdot z_j = \delta_{i,j}
    \end{equation}
    Similarly, given a basis of $H^1(C_\bul)$, there is a unique basis of
    $H_1(C_\bul)$ satisfying \cref{canonical_condition}.
\end{proposition}

\begin{proof}
   We refer to prop.~\ref{canonical_bases}.
\end{proof}

Given dual bases $\{[z_1], \ldots, [z_k]\}$ and $\{[x_1], \ldots, [x_k]\}$,
we will define \textbf{logical qubits} with respect to a pair of logical operators
$([x_i], [z_i])$ such that $[x_i] \cdot [z_i] = 1$.

\subsection{Transformations of CSS codes are chain maps}

\subsubsection{Preserving code maps}

In order to reason about the types of operations we can implement using
surgery in CSS codes, it is first important to understand the transformations
of codes that we can implement on the physical level that preserve the CSS structure of the code.  We will call such operations \textbf{CSS code maps}. It
is hopefully clear that all CSS code maps must be stabiliser operations\footnote{Stabiliser operations are generated by Clifford unitaries, initialisations of \emph{stabiliser} states \(\ket{S}\) and their duals \(\bra{S}\). This is exactly the set of linear maps that take stabiliser groups to stabiliser groups.}, since CSS codes are stabiliser codes, but in order for the resulting
code after the transformation to be CSS, a CSS code map must also preserve
the separation of stabilisers into a \(Z\) block \(P_Z\) and an \(X\) block
\(P_X\). This latter property is reflected at the level of chain complexes,
and we can identify CSS code maps with certain families of linear maps between
the chain and cochain complex representing a CSS code. In fact, we can identify
two building blocks out of which any CSS code map can be constructed: the Z-
and X-preserving code maps \cite{Cowtan_2024}.

\paragraph{Z-preserving code maps} 
Consider two CSS codes represented by chain complexes \(C_\bul\) and \(D_\bul\)
as in \cref{chain_complex}.  A \textbf{chain map} $f_\bul : C_\bul
\rightarrow D_\bul$ (of length two) is a collection of three matrices $\{f_2,f_1,f_0\}$
\begin{equation}\label{eq:chain_map}
    \begin{tikzcd}
        C_2 \arrow[r, "\del_2^C"]\arrow[d, "f_2"] & C_1 \arrow[r, "\del_1^C"]\arrow[d, "f_1"] & C_0 \arrow[d, "f_0"]\\
        D_2 \arrow[r, "\del_2^D"] & D_1 \arrow[r, "\del_1^D"] & D_0
    \end{tikzcd}
    \quad \text{such that} \quad
    \begin{cases}
        f_1\del_2^C = \del_2^Df_2 \\
        f_0\del_1^C = \del_1^Df_1
    \end{cases},
\end{equation}
i.e., every square in the diagram commutes. In terms of quantum error
correction, a chain map can be interpreted as a physical transformation \(f_1\)
from a code \(C_\bul\) to a code \(D_\bul\), with two further conditions imposed
by existence of \(f_2\) and \(f_0\):
\begin{itemize}
    \item The existence of \(f_2\) imposes that Z-operators that are
    stabilisers remain stabilisers. Concretely, considering that each
    Z-stabiliser in $C_1$ can be written $\del_2^Cx \in \im(\del_2^C)$, $f_1$
    maps it to a stabiliser in $D_1$ since $f_1 \del_2^Cx = \del_2^D f_2 x \in
    \im(\del_1^D)$.
    \item The second contingency is on syndromes. Similarly, an X-syndrome of
    a Z-operator $\del_1^Cy \in \im(\del_1^C)$ is mapped to an X-syndrome
    in $D_0$ via $f_0$ as $f_0\del_1^C y = \del_1^D f_1 y \in
    \im(\del_1^D)$.
\end{itemize}

\begin{definition}
    Let $(C_\bul, C^\bul)$ and $(D_\bul, D^\bul)$ be two CSS codes. A
    \textbf{Z-preserving code map} from $(C_\bul, C^\bul)$ to $(D_\bul,
    D^\bul)$ is a chain map $f_\bul : C_\bul \rightarrow D_\bul$ (of length two).
\end{definition}

\paragraph{X-preserving code maps}
The same concept applies when we choose to represent $(C_\bul, C^\bul)$
and $(D_\bul, D^\bul)$ with their cochain complexes.  A \textbf{cochain
map} $g^\bul : C^\bul \rightarrow D^\bul$ (of length two) is a collection of three matrices
$\{g^2,g^1,g^0\}$
\begin{equation}
    \begin{tikzcd}
        C^2 \arrow[d, "g^2"] & \arrow[l, "\del^2_C"'] C^1 \arrow[d, "g^1"] & \arrow[l, "\del^1_C"'] C^0 \arrow[d, "g^0"]\\
        D^2 & \arrow[l, "\del^2_D"'] D^1 & \arrow[l, "\del^1_D"'] D^0
    \end{tikzcd}
\end{equation}
such that every square in the diagram commutes.  The interpretation is
the same up to exchanging the roles of X and Z Paulis. We call these maps \textit{X-preserving code maps}.

\begin{definition}
    An \textbf{X-preserving code map} from $(C_\bul, C^\bul)$ to $(D_\bul,
    D^\bul)$ is a cochain map $g^\bul : C^\bul \rightarrow D^\bul$ (of length two).
\end{definition}

We postpone the question of how to implement preserving code maps physically to focus on another meaningful particularity of preserving code maps: they lift to maps between the logical spaces. A natural challenge encountered when we aim at computing with encoded data is to determine the effect an operation performed on the physical qubits has on the logical qubits. In the case where the physical operation is a preserving code map, the resulting logical operation is easy to deduce. This comes as a consequence of homological algebra. In fact, every (co)chain map lifts to a map between the (co)homology spaces. In other words, every Z- (or X-) preserving code map gives a map between the Z- (or X-) logical operators. Knowing how the logical operators are mapped makes it possible to determine the logical operation.

\paragraph{Maps between Z-logical operators} Consider a Z-preserving code map $f_\bul : C_\bul \rightarrow D_\bul$. It lifts to a map between the homology spaces $H_1(C_\bul)$ and $H_1(D_\bul)$ defined for all $[z] \in H_1(C_\bul)$ as
\begin{equation}
    f_{1*}([z]):=[f_1(z)].
\end{equation}
In fact, $H_1$ is a functor. We refer the interested reader to \cref{hom&cat}.

\paragraph{Maps between X-logical operators} Similarly, for a X-preserving code map $g^\bul : C^\bul \rightarrow D^\bul$, the corresponding map between the cohomology spaces $H^1(C_\bul)$ and $H^1(D_\bul)$ is defined for all $[x] \in H^1(C_\bul)$ as
\begin{equation}
    g^1_*([x]):=[g^1(x)].
\end{equation}

Armed with the transformation of the physical and logical operators described by the preserving code maps (more precisely the middle map) and the lift on the (co)homology spaces, we can now decipher the operation implemented at the physical and logical level, namely on the physical and logical Hilbert space.

\subsubsection{Interpretation of preserving code maps on the physical and logical Hilbert space}

This section introduces the physical implementations of preserving code maps as well as the operation they turn out to implement on the logical Hilbert space. These operations are most naturally represented using the ZX-calculus, as we have chosen to do here. Being able to describe the operational interpretation of the surgery procedure at the level of Hilbert spaces (i.e. quantum systems) is obviously essential, but an in-depth understanding of this section is not required to understand the surgery operations of \cref{sec:css_surgery}.

So far, preserving code maps have been presented through the lens of chain complexes. However, they have a concrete interpretation at the physical level.
Returning to the physical Hilbert space, suppose we have physical qubits encoded
in a CSS code $(C_\bul, C^\bul)$. We wish to describe what kind of quantum
channel the preserving code maps represent. They take the form of a parity
map \(\ket{x} \mapsto \ket{Ax}\) for some binary matrix A, and these parity
maps have a particularly simple interpretation using the \emph{SZX-calculus}
\cite{https://doi.org/10.4230/lipics.mfcs.2019.55}. For a reader not familiar
with neither the ZX- nor the SZX-calculus, we refer to sections 3 and 4 of
\cite{vandewetering2020zxcalculusworkingquantumcomputer}. Meanwhile, for
readers already familiar with ZX-calculus, we informally introduce the two
kinds of generators added to make ZX-calculus ``scalable''. Informally,
these are the gatherers and dividers
\begin{equation}\input{gatherer.tikz} \ \ \text{and} \ \ \input{gdiv0.tikz}\end{equation}
where thick wires indexed by an integer $n$ represent $n$ single wires parallel to each others, and, up to global phases, the matrix arrows
\begin{equation}\label{eq:right_and_left_arrows}
    \input{rep0.tikz} = \input{rep1.tikz} \ \ \text{and} \ \ \input{rep0b.tikz} = \input{rep1b.tikz}
\end{equation}
where $A$ and $B$ are seen as adjacency matrices of bipartite graphs. 

We can also define these generators by their action on the Hilbert space. Suppose $A \in \F_2^{n\times m}$ and $B \in \F_2^{p\times q}$. For a qubit state in the computational basis $\ket{x}$ with $x \in \F_2^m$, the right arrow in \cref{eq:right_and_left_arrows} corresponds to the linear map $R_A : \ket{x} \mapsto \ket{Ax}$. Meanwhile, the left arrow corresponds to the linear map $H^{\tens p} R_{B^T}H^{\tens q}=H^{\tens p}\left(\ket{x} \mapsto \ket{B^Tx}\right)H^{\tens q}$. We are now ready to give the linear map implementing preserving code maps.

\begin{remark}
   In the following, to streamline the notations, we will use the convention that for any map indexed by an integer $\epsilon_n$, we define $\epsilon^n = \epsilon_n^T$.
\end{remark}

\paragraph{Interpretation of preserving code maps on the physical Hilbert space} The next two propositions translate preserving code maps in terms of SZX-diagrams applied on the physical qubits of the code. Both interpretations essentially consist in a left (resp. right) matrix arrow followed by a series of Z- (resp. X-) projections which add the missing stabilisers of the new code~\footnote{We are grateful to Alexander Cowtan for giving us the correct interpretations through personal communications.}.

\begin{proposition} \label{hilb_representation_chain} 
    Every Z-preserving code map $f_\bul : C_\bul \rightarrow D_\bul$ is interpreted on the physical Hilbert space as an SZX diagram
    \begin{equation}\label{eq:hilb_representation_chain}
        \input{SZX_f_0.tikz}
    \end{equation}
    where $\del_\Omega:=\del_2^D\restriction_\Omega$ with $\Omega$ being any supplementary space of $\im(f_2)$, i.e. $D_2=\im(f_2)\oplus \Omega$.
\end{proposition}

\begin{proof}
    We refer to appendix~\ref{ap:interpretation}.
\end{proof}

A simpler, though less precise, physical interpretation is obtained by replacing $\del_\Omega$ with $\del_2^D$. Since any input state lies in the codespace of $(C_\bul,C^\bul)$, some of the resulting projections act as the identity---this occurs when certain Z-stabilisers are preserved and carried over through the SZX left arrow $f_1^T$. As we point it out in the proof (appendix~\ref{ap:interpretation}), these correspond precisely to elements of $\im f_2$.

We highlight two important points regarding the projection in \cref{eq:hilb_representation_chain}:
\begin{itemize}
    \item $\Omega$ corresponds to the Z-stabiliser subset of $(D_\bul,D^\bul)$ that cannot be obtained as Z-stabilisers of $(C_\bul,C^\bul)$ mapped via $f_2$. In \cite{Cowtan_2024}, the projection onto these Z-stabilisers of the code $(D_\bul,D^\bul)$ was omitted. Consequently, to ensure that the codespace of $(C_\bul,C^\bul)$ is correctly mapped to that of $(D_\bul,D^\bul)$, it was necessary to require the surjectivity of $f_2$ in the Z-preserving code map $f_\bul$. Without this additional condition, we demonstrate in appendix~\ref{ap:counterexample} that, starting from the codespace of $(C_\bul,C^\bul)$ and applying the physical interpretation \cref{eq:hilb_representation_chain}, certain expected Z-stabilisers given by the interpretation of $D_\bul$ may fail to stabilise the resulting space.

    \item In the particular case where $f_2$ is surjective, we get $\Omega=0$ and $\del_\Omega=0$. So, the interpretation simplifies to
    \begin{equation}
        \input{SZX_f_0_surjective.tikz}.
    \end{equation}
    We recover the original interpretation presented in \cite{Cowtan_2024} and in earlier versions of this work.
\end{itemize}

The same reasoning applies dually for cochain maps with the same considerations regarding the projection onto some of the X-stabilisers of $(D_\bul,D^\bul)$.

\begin{proposition} \label{hilb_representation_cochain}
    Every X-preserving code map $g^\bul : C^\bul \rightarrow D^\bul$ is interpreted on the physical Hilbert space as an SZX diagram.
    \begin{equation}\label{eq:hilb_representation_cochain}
        \input{SZX_g0.tikz}
    \end{equation}
    where $\del^\Omega:=\del^1_D\restriction_\Omega$ with $\Omega$ being any supplementary space of $\im(g^0)$, i.e. $D^0=\im(g^0)\oplus \Omega$.
\end{proposition}

The same remarks made for prop.~\ref{hilb_representation_chain} apply here as well dually.

\paragraph{Interpretation of preserving code maps on the logical Hilbert space}
The same way we have used SZX-calculus to define the physical implementation
of preserving code maps, we can describe the quantum channel on the logical
Hilbert space with left and right SZX arrows. There is however one additional
challenge: for a Z-preserving code map $f_\bul : C_\bul \rightarrow D_\bul$,
the lift $f_{1*}: H_1(C_\bul) \rightarrow H_1(D_\bul)$ to the logical space
is not a matrix, but simply a linear map between two $\F_2$-vector spaces
with no preferential choice of basis. This echoes \cref{seq:logop} where we
highlighted the importance of choosing a basis of logical operators to be
able to compute with encoded data. Thus, $f_{1*}$ can be represented as a
matrix if bases of $H_1(C_\bul)$ and $H_1(D_\bul)$ are already chosen.

\begin{definition} \label{def:log_hilb_rep}
    Let $f_\bul : C_\bul \rightarrow D_\bul$ be a Z-preserving code map. Consider two bases $\mathcal{B}$ and $\mathcal{C}$ of $H_1(C_\bul)$ and $H_1(D_\bul)$, respectively. The interpretation of $f_{1*}$ on the logical Hilbert space is given by the following SZX diagram
    \begin{equation}
        \input{F_0star.tikz}
    \end{equation}
    where $F_{1*}$ is the matrix representation $F_{1*} = Mat_{(\mathcal{B},\mathcal{C})}(f_{1*})$ of $f_{1*}$ in the given bases.
\end{definition}

The interpretation of the lift is chosen such that it corresponds to the logical operation induced by the interpretation of the chain map.

\begin{theorem} \label{theorem:soundness}
    Let $f_\bul : C_\bul \rightarrow D_\bul$ be a Z-preserving code map. Consider two type-preserving encoders $E_C$ and $E_D$ for $(C_\bul,C^\bul)$ and $(D_\bul,D^\bul)$, respectively (cf. remark~\ref{rem:type_preserving}). According to prop.~\ref{prop:encoder-basis}, they define a unique choice of basis $\varepsilon_C$ and $\varepsilon_D$. Denoting $\varepsilon_C(0|\delta_i)=\left([0], \left[z^{(C)}_i\right]\right)$ for all $i\in \llbracket1,k_C\rrbracket$ and $\varepsilon_D(0|\delta_j)=\left([0], \left[z^{(D)}_j\right]\right)$ for all $j \in \llbracket1,k_D\rrbracket$, we get 
    \begin{equation}
        \input{interpretation_proof.tikz}=\input{F_0star.tikz}
    \end{equation}
    where $F_{1*} = Mat_{\left(\left\{\left[z^{(C)}_i\right]\right\}_i,\left\{\left[z^{(D)}_j\right]\right\}_j\right)}(f_{1*})$ and $\del_\Omega$ is defined as in prop~\ref{hilb_representation_chain}.
\end{theorem}

\begin{proof}
    We refer to appendix~\ref{ap:interpretation}.
\end{proof}

\begin{definition} \label{def:log_hilb_rep_cochain}
    Let $g^\bul : C^\bul \rightarrow D^\bul$ be an X-preserving code map. Consider two bases $\mathcal{B}$ and $\mathcal{C}$ of $H^1(C_\bul)$ and $H^1(D_\bul)$, respectively. The interpretation of $g^1_*$ on the logical Hilbert space is given by the following SZX diagram
    \begin{equation}
        \input{G0star.tikz}
    \end{equation}
    where $G^1_*$ is the matrix representation $G^1_* = Mat_{(\mathcal{B},\mathcal{C})}(g^1_*)$ of $g^1_*$ in the given bases.
\end{definition}

Finally, we can apply the same reasoning dually for lifts of X-preserving code maps, and get the dual of theorem~\ref{theorem:soundness}.

Theorem~\ref{theorem:soundness} sets the limits of the physical operations captured by our theoretical framework. Preserving code maps can only describe CSS operations on the physical qubits.
Not surprisingly, this limitation yields another one at the logical scale: only CSS operations can be described on the logical qubits. The logical operation realised by a preserving code map does not ``mix'' the types of the logical operators: Z-logicals remain Z-logicals and X-logicals remain X-logicals.
Thus, even though (co)chain complexes describe all implementable CSS codes, (co)chain maps only capture CSS operations. Among the generators $\{CNOT, H, S\}$ of the Clifford group, only the $CNOT$ respects this separation of X- and Z-type operators. So, among the Clifford group only $CNOT$s can be performed at the logical level using preserving code maps. 

While we maintain that the category $\Chains$ does not fully capture CSS codes along with all relevant physical operations, we will continue to work within this category because it provides a comprehensive understanding of how to perform logical CSS operations. The category of chain complexes is well-known and allows us to leverage the properties of the functor $H_1$, which remains an excellent tool for analysing the logical data.

\section{CSS surgery: merges and splits} \label{sec:css_surgery}

In \cite{Cowtan_2024} it was shown that the chain complex formulation of CSS codes yields an elegant construction
for the merge of two CSS codes as a \emph{pushout}\footnote{A pushout is a general categorical construction that captures the idea of gluing two compatible pieces of data along the parts where they agree \cite{leinster2016basiccategorytheory}. Since we do not use this formulation, we will not explicit it any further.} of the corresponding
chain complexes. In this work, we show that \cite{Cowtan_2024}'s construction can be equivalently reformulated as a \emph{quotient} of the chain complexes. The specific advantage of our approach is that it allows us to leverage results from homological algebra to compute the homology of merges and splits. Our framework imposes weaker conditions on the codes than \cite{Cowtan_2024}, meaning we can determine the logical effect of much more general surgery procedures.

We first give a series of motivating examples which illustrate the
key differences between the standard lattice surgery for surface codes
(\cref{sec:lattice_surgery}) and the chain-complex formulation of
surgery, as well as the difficulties one encounters when trying to generalise
lattice surgery naively. Two main questions will guide us throughout
this section: \emph{1) How do we choose physical qubits to be ``merged'' and
``split''?}; and, \emph{2) What logical operations do the resulting merges and
splits implement?} Answering the first question leads to the concept of
\emph{subcodes}, while the second is answered using a homological analysis of the resulting transformations.

\subsection{Motivating examples} \label{seq:motivating_examples}

Before jumping into the general framework of CSS surgery, we introduce relevant examples of surgery protocols each highlighting a different subtlety inherent to surgery. The theoretical framework we introduce in the next subsection natively captures all these subtleties.

\subsubsection{Surface code welding}

Although the formalism of preserving code maps can describe lattice surgery, the construction of merges and splits that we present is conceptually much closer to \textbf{surface code welding} \cite{michnicki20123dquantumstabilizercodes}. Specifically, our approach relies on taking quotients of chain complexes, as introduced in def.~\ref{def:quotientZmerge}. While this quotient structure aligns naturally with the merge operations in surface code welding, it does not directly correspond to those used in lattice surgery.

Consider two patches of surface codes $C$ and $D$, depicted below. Surface code welding consists in merging these two patches along the edges in the orange shaded area. Once
merged, the two surface patches form one big rectangular surface patch $M$. 

\begin{equation}
  \label{eq:merge_split_als}
  \vcenter{\hbox{\includegraphics[width=0.6\linewidth]{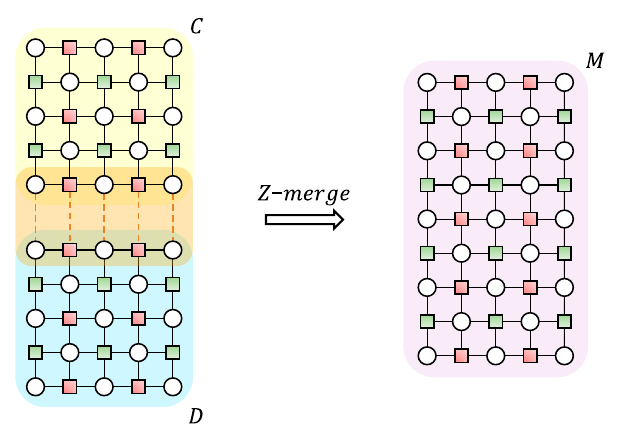}}}
\end{equation}
The Z-merge operation depicts the ``gluing'' of each pair
of qubits and each pair of X-checks at the boundaries of the two surface
patches. In practice, it is carried out by the physical CSS operation
\begin{equation}
  \input{figures/logical_merge_green.tikz} = \ket{0}\bra{00} + \ket{1}\bra{11}
\end{equation}
applied to each pair of qubits to be merged. Notably, this physical operation can be written as the physical interpretation of a Z-preserving code map.

\subsubsection{The subtleties of surgery for arbitrary CSS codes}
\label{sssec:surgery_subtleties}

In light of surface code welding, and in particular the merge of \cref{eq:merge_split_als}, one might expect that merges should
be formalised as simple ``gluings'' of Tanner graph. While our framework based on preserving code maps does incorporate gluing, it ultimately involves more refined operations. Nonetheless, visualising merges as Tanner graph gluings provides valuable intuition, especially for understanding the core challenges in developing a general theory of merges and splits. In order to keep the examples as simple as possible, we use surface codes, while emphasising that the framework we are about to present applies to any CSS codes, and in particular, qLDPC codes.

\paragraph{Many merges lead to the same code.}\label{para:many_splits}
Starting with different pairs of codes, we can get the same merged code by performing different merges. For instance, consider two pairs of surface codes $(C,D)$ and $(C',D')$. Patches $C$, $C'$ and $D'$ are surface codes of size $3 \times 3$ while $D$ is of size $4 \times 3$. Below are represented two different Z-merges performed on the two pairs. The merge occurs within the orange shaded regions of the two codes and, in this case, simply corresponds to a gluing.
\begin{equation}\label{eq:split_one_to_many}
  \vcenter{\hbox{\includegraphics[width=.90\linewidth]{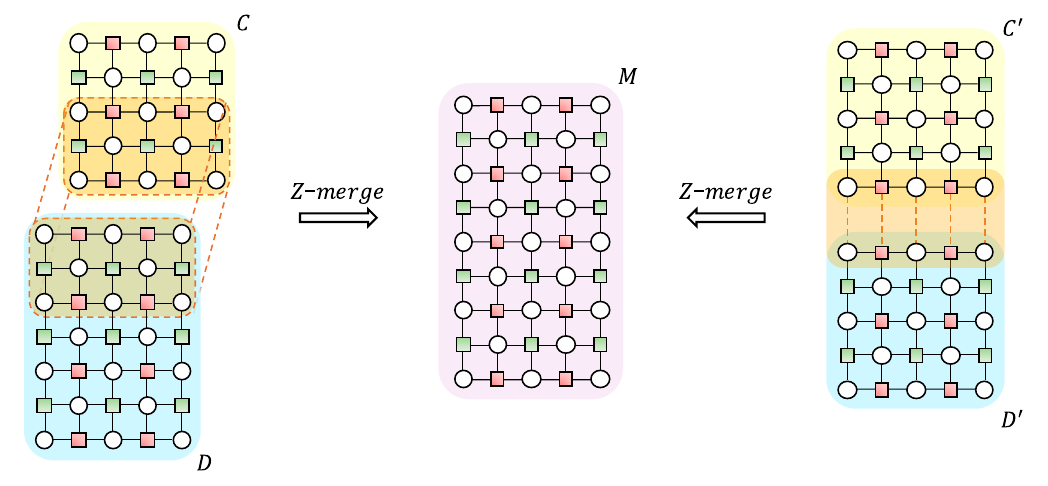}}}
\end{equation} 
Note that the merged code $M$ resulting from these two merges is the same. This leads to a difficulty in the definition of the split: whereas a merge only requires an orange merge-zone to be determined, a split requires a split-zone in addition to a precise specification of the resultant split codes. Whereas there are only a finite number of merges given two CSS codes, there are many splits given \emph{one} CSS code.

\paragraph{Internal merges are possible.}
It is possible for a code to be merged with itself. To illustrate, consider the orange shaded area on the surface code patch $C$.
\begin{equation}\label{eq:internal_merge}
  \vcenter{\hbox{\includegraphics[width=0.55\linewidth]{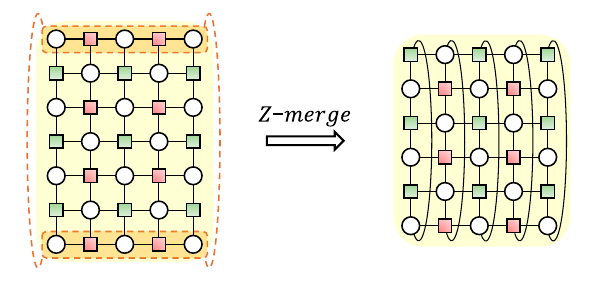}}}
\end{equation}
Gluing the qubits in this area pairwise results in a merged code that is topologically equivalent to a (discretised) cylinder. We could push this further by merging the two boundaries of the cylinder, yielding the famous toric code.

\paragraph{Not all gluings of Tanner graphs yield valid merges.}
So far, the merge operation appears to reduce to choosing a set of qubits and a set of checks to be merged, and then gluing them together. It is unfortunately more subtle. Consider the following ``merge'' where we glue the two qubits in the orange shaded area together.
\begin{equation} \label{eq:wrong_merge}
  \vcenter{\hbox{\includegraphics[width=.90\linewidth]{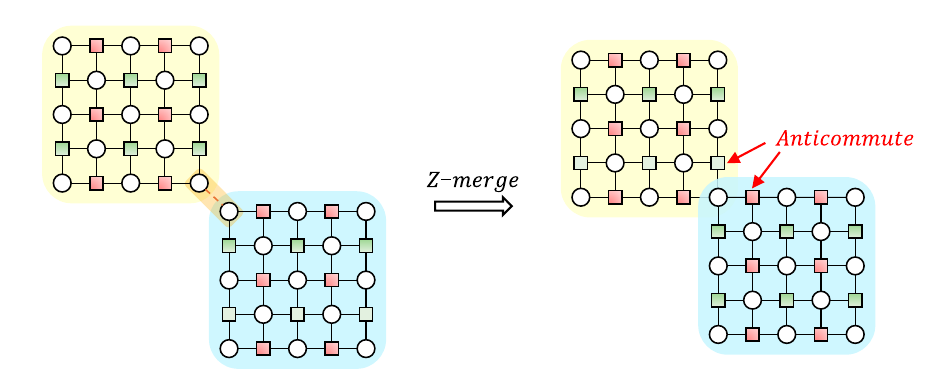}}}
\end{equation}
The resulting merged code is invalid. The Z-check and the X-check marked by the red arrows anti-commute as they overlap only on one qubit. This raises the question: what are the rules for selecting the orange shaded area?

\paragraph{Not all merges can be represented as graph
gluings.}\label{para:not_only_gluings} Tanner graphs implicitly forget the
linear structure of the stabiliser groups at hand.  Whilst the stabiliser
groups, and their linear structure, can always be recovered from the Tanner
graph, this means that not all natural operations at the level of stabiliser
groups have natural representations in terms of the Tanner graphs. Consider
the case of a CSS code with Pauli Z-operators \(Z_1\) and \(Z_2\),
where we want to perform a merge with the Z-operator \(Z_3\) of a second
code. The operator \(Z' \coloneqq Z_1   Z_2\)  is a perfectly valid Z-operator, and we can merge \(Z_3\) with \(Z'\). This merge will almost
certainly not respect the structure of the parity-check matrix written in the
\(\{Z_1,Z_2\}\) basis, and therefore the structure of the Tanner graph will
not be preserved. Nevertheless, this is a valid merge. Giving a more explicit example requires the use of the language
of subcodes which we develop in the following sections, so we defer it to
example~\ref{example:virtual_merge}.

\subsection{Arbitrary merges and splits of arbitrary CSS codes}

Cowtan and Burton's formulation of lattice surgery in terms of pushouts
of chain complexes satisfactorily addresses the difficulties we have just
discussed \cite{Cowtan_2024}. However, whilst their surgery procedure describes
valid merges and splits at the physical level, they were unable to compute the induced logical operation---determined by the
homology---of their procedure without imposing an irreducibility condition on the logical operators of the codes to be merged. The
logical operation implemented by more general surgery procedures remains unknown. We
introduce a subtle refinement of their construction, via the concept of
\emph{subcodes}, which makes it possible to use standard techniques from homological algebra to compute the homology. Since the formalisms are essentially equivalent in how they describe merges, our results also allows one to compute the homology for Cowtan and Burton's formulation, although we chose not to do so here.

\subsubsection{Subcodes}

A problem that arises when representing surgery of CSS codes in terms of gluing and cutting of Tanner graphs is that Tanner graphs forget the linear structure of the stabiliser group since they prioritise tracking the parity checks of the code. On the other hand, stabiliser groups (without a chosen generating set) lack the operational meaning of Tanner graphs, as described in remark~\ref{rem:chain_complex_operational}. Chain complexes track both of these data: the parity-checks are tracked by the boundary maps \(P_Z^T\) and \(P_X\), and the linear structure of stabilisers is tracked by the complexes \(C_2,C_1,C_0\) which are \(\F_2\)-vector spaces. The question that then arises is: what does it mean to merge, or ``glue'', two chain complexes?

Subcodes emerge as an answer to the question posed by the erroneous ``merge'' in \cref{eq:wrong_merge}. Being themselves CSS codes, they are the keystone of merges. As usual, we can define two types of subcodes using either a chain or a cochain complex.

\begin{definition} \label{def:subcode}
  A CSS code
  $(V_\bul,V^\bul)$ is called a \textbf{Z-subcode} of a CSS code $(E_\bul, E^\bul)$ if the inclusion map \(i_\bullet :
  V_\bul \hookrightarrow E_\bul\), where each \(i_n\) is a natural inclusion, defines a valid chain map.
  Analogously, $(V_\bul,V^\bul)$ is called an \textbf{X-subcode} of $(E_\bul, E^\bul)$ if the inclusion map \(i^\bullet : V^\bul \hookrightarrow E^\bul\), with each \(i^n\) an inclusion, defines a valid cochain map.
\end{definition}

\noindent Explicitly, $(V_\bul,V^\bul)$ is a Z-subcode of $(E_\bul, E^\bul)$ if:
\begin{itemize}
    \item $V_n$ is a subvector space of $E_n$ for all $n \in \{0,1,2\}$,
    \item if $x \in V_n$, then $\del_n^Ex \in V_{n-1}$ for  $n \in \{1,2\}$,
\end{itemize}
and the boundary maps $\del_n^V = \del_n^E\restriction_{V_n}$ are restrictions
of the ones of $E_\bul$. The explicit X-subcode description is similar swapping chain and cochain complexes.

The following example aims at building an intuition around the concept of X- and Z-subcode.

\begin{example} 
    We can formally describe the subcode chosen in \cref{eq:internal_merge}. Consider the rectangular surface code $(C_\bul, C^\bul)$, 
    \begin{equation} \label{eq:internal_merge_labeled}     
        \vcenter{\hbox{\includegraphics[width=.3\linewidth]{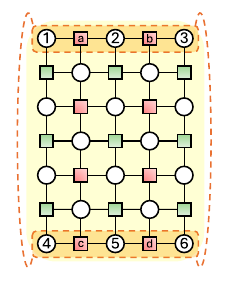}}}
    \end{equation}
    The orange shaded area represents a subcode. Formally, let $z_1,\ldots,z_6$ and $s_a,s_b,s_c, s_d$ be some canonical bases vectors of $C_1$ and $C_0$ corresponding to the qubits and X-checks on \cref{eq:internal_merge}, respectively. The Z-subcode $(V_\bul,V^\bul)$ of $(C_\bul, C^\bul)$ is described by the chain complex
    \begin{equation}
        \begin{tikzcd}
            V_\bul : \ 0 \arrow[r] & \Span\{z_1+z_4,z_2+z_5,z_3+z_6\} \arrow[r] & \Span\{s_a+s_c, s_b+s_d\}
        \end{tikzcd}
    \end{equation}
    One can easily check that $(V_\bul,V^\bul)$ is a valid Z-subcode: the first of the two constraints detailed below def.~\ref{def:subcode} is immediately satisfied, while the second requires to be verified for each basis vector of $V_1$ (for instance, $z_1+z_4 \in V_1$ and $\del_1^C(z_1+z_4) = s_a + s_b$ belongs to $V_0$).
\end{example}

\begin{example}
    Consider the 7-qubits Steane code $(C_\bul, C^\bul)$ pictured below:
    \begin{equation} \label{eq:steane_code}     
    \vcenter{\hbox{\includegraphics[width=.40\linewidth]{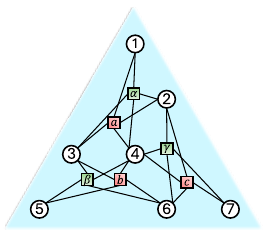}}}
    \end{equation}
    Let $\{s_\alpha,s_\beta,s_\gamma\}$,$\{z_1,z_2,z_3,z_4,z_5,z_6,z_7\}$ and $\{s_a,s_b,s_c\}$ be the canonical bases of $C_2$, $C_1$ and $C_0$, and $\{s^\alpha,s^\beta,s^\gamma\}$, $\{x^1,x^2,x^3,x^4,x^5,x^6,x^7\}$ and $\{s^a,s^b,s^c\}$ be the canonical bases of $C^2$, $C^1$ and $C^0$, respectively. These choices are made up to match the Tanner graph representation \cref{eq:steane_code}, i.e., the boundary maps are indeed the adjacency matrices of the Z-checks and the qubits, and of the qubits and the X-checks. We now give some examples of subcodes. \\
    The chain complex
    \begin{equation}\label{CC}\begin{tikzcd}
        \Span\{s_\alpha+s_\gamma\} \arrow[r] & \Span\{z_1+z_3, z_6+z_7\} \arrow[r] & \Span\{s_b\}
    \end{tikzcd}\end{equation}
    represents a Z-subcode of $(C_\bul, C^\bul)$, while the following cochain complex represents an X-subcode:
    \begin{equation}\begin{tikzcd}
        0 & \arrow[l] \Span\{x^1,x^4,x^5,x^7\} & \arrow[l] \Span\{s^a+s^b+s^c\}
    \end{tikzcd}.
    \end{equation}
    However, the following chain complex does not represent a Z-subcode of $(C_\bul, C^\bul)$
    \begin{equation}\begin{tikzcd}
        0 \arrow[r] & \Span\{z_1+z_2\} \arrow[r] & \Span\{s_a+s_c\}
    \end{tikzcd}.\end{equation}
    Indeed, $\del_1^C(z_1+z_2) = \del_1^C(z_1) + \del_1^C(z_2) = s_a+s_a+s_c=s_c \notin \Span\{s_a+s_c\}$---the second condition below def.~\ref{def:subcode} is not fulfilled.
\end{example}

\begin{remark} \label{rem:lin_comb}
    These examples show the importance of considering linear combinations and
    not just subsets of the basis vectors of the complexes. As such, developing
    this theory using only Tanner graphs would be very inconvenient. When designing more complicated merges where the Tanner representation is
    inadequate, we will abandon the Tanner graph representation in favour of
    the chain complex representation.
\end{remark}

As we are about to explain, subcodes make it possible to express merges as a
simple \emph{quotienting} operation for chain complexes. However, in order
to be able to describe a merge between two distinct CSS codes, we need to
explicit where to look for the corresponding subcodes. Given two CSS codes,
we need to give a single chain complex that represents the joint structure of
each code. This is analogous to how the tensor product \(\mathcal{H} \otimes
\mathcal{J}\) gives a \emph{single} Hilbert space describing the joint state
space of the Hilbert spaces \(\mathcal{H}\) and \(\mathcal{J}\):
\begin{definition} \label{def:direct_sum}
    The \textbf{direct sum} $((C \oplus D)_\bul,(C \oplus D)^\bul)$ of two
    CSS codes $(C_\bul,C^\bul)$ and $(D_\bul,D^\bul)$ is defined using the
    direct sums of the (co)chain complexes:
    \begin{equation}
        \begin{tikzcd}
            (C \oplus D)_\bul : \ C_2 \oplus D_2 \arrow[r, "\del_2^C \oplus \del_2^C"] & C_1 \oplus D_1 \arrow[r, "\del_1^C \oplus \del_1^C"] & C_0 \oplus D_0
        \end{tikzcd},
    \end{equation}
    \begin{equation}
        \begin{tikzcd}
            (C \oplus D)^\bul : \ C^2 \oplus D^2 \arrow[r, "\del^2_C \oplus \del^2_C"] & C^1 \oplus D^1 \arrow[r, "\del^1_C \oplus \del^1_C"] & C^0 \oplus D^0 
        \end{tikzcd}.
    \end{equation}
\end{definition}

\subsubsection{Merges}

Inspired by \cite{Cowtan_2024} and equipped with subcodes, we now define the building blocks of any CSS surgery protocol: merges and splits. These are of two types: X-type and Z-type. This distinction fully leverages the fact that a CSS code can be equivalently represented by either a chain complex or its cochain complex. 

\begin{definition}
    (Z-merge) A Z-merge of a CSS code $(E_\bul,E^\bul)$ is a surjective Z-preserving code map $p_\bul : E_\bul \rightarrow Q_\bul$ with $(Q_\bul,Q^\bul)$ being the \textbf{Z-merged code}. In other words, each of the three maps $p_2,p_1$ and $p_0$ needs to be surjective.
\end{definition}

While for any Z-merge we can analyse the logical operation induced by such a Z-preserving code map as presented in prop.~\ref{H_merge}, this fairly abstract definition does not provide a systematic way to construct a Z-merge. For this reason, we propose a second approach of the Z-merges, the \emph{quotient merges}, that leverages the notion of subcodes to provide a systematic construction of an epic Z-preserving code map.

\begin{definition} \label{def:quotientZmerge}
    (Quotient Z-merge) Consider a CSS code $(E_\bul,E^\bul)$, and a Z-subcode $(V_\bul, V^\bul)$ of $(E_\bul,E^\bul)$. A quotient Z-merge is a Z-merge $p_\bul : E_\bul \rightarrow (E/V)_\bul$ such that the Z-merged code $(E/V)_\bul$ is defined as
    \begin{equation}
        \begin{tikzcd}
            E_2 /V_2 \arrow[r, "\del_2^/"] & E_1/V_1 \arrow[r,"\del_1^/"] & E_0 / V_0
        \end{tikzcd}
    \end{equation}
    where for all $n \in \{0,1,2\}$ and $x \in E_n$,
    \begin{equation}
        \del_n^/[x] := [\del_n^Ex].
    \end{equation}
    The maps $p_n:E_n \rightarrow (E/V)_n$ are simply given by the respective
    projection onto the equivalence class: for any \(x \in E_n\), $p_n(x):=[x]_{V_n}$.
\end{definition}

A few remarks can be made about \cref{def:quotientZmerge}. First, $(V_\bul,V^\bul)$ needs to be a Z-subcode of $(E_\bul,E^\bul)$ for $(E/V)_\bul$ to be well defined. Second, according to prop.~\ref{hilb_representation_chain}, since $p_1$ is surjective, a merge cannot introduce any auxiliary physical qubits. This clarifies why lattice surgery merges cannot be directly represented as merges within our formalism. Third, for all $n \in \Z$, $\del_n^/$ is not defined as a matrix but as a linear map instead. To turn it into a matrix, we would have to choose a basis for each quotient space $E_n/V_n$. However, there is no natural choice of such bases, even though bases of $E_n$'s and $V_n$'s are provided. So, this choice has to be made arbitrarily. Although, the code dimension is not affected by this choice, some other parameters will. For instance, the code distance and the weight of the generators of the stabiliser group of $((E/V)_\bul, (E/V)^\bul)$ will be impacted. Also, the Z-merge $p_\bul$ will result in different matrices depending on the choices of basis (but which are all equivalent up to a multiplication on the left by an invertible matrix). Overall, these choices impact the physical implementation of $p_\bul$ as well as the QEC properties of $((E/V)_\bul, (E/V)^\bul)$. This is why this problem has to be dealt with on a case-by-case basis, depending on the fault-tolerant properties we want our circuit to have. 

The question that now arises is: can any merge be seen as a quotient merge for an appropriate subcode? The answer is negative. More precisely, \emph{every Z-merge differs from a quotient Z-merge by an isomorphic chain map} as it is proven in \cref{ap:merge-quotient}. According to prop.~\ref{hilb_representation_chain}, this implies that any Z-merge is a quotient Z-merge followed by an array of CNOTS (represented by the isomorphic chain map). While this is important from a fault-tolerant perspective, as two isomorphic chain complexes do not necessarily describe CSS codes with the same properties (code distance, weight of the stabilisers...), in the rest of this work, we mainly focus on quotient Z-merges for both the simplicity of their construction and the fact that they carry the non invertible structure of any Z-merge. 

Plugging $((C \oplus D)_\bul,(C \oplus D)^\bul)$ in prop.~\ref{def:quotientZmerge}, we get the definition of a \textbf{quotient Z-merge between $(C_\bul,C^\bul)$ and $(D_\bul,D^\bul)$ along $(V_\bul,V^\bul)$}. We now provide some examples to develop intuition on quotient merges.

\begin{example} \label{ex:merge_als}
    Consider the case of surface code welding.
    Such a merge can be visualised using Tanner graphs: 
    \begin{equation}
      \begin{tikzcd}
        \vcenter{\hbox{\includegraphics{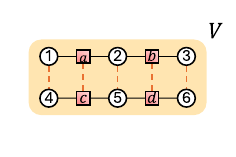}}} \arrow[hook, r] & \vcenter{\hbox{\includegraphics{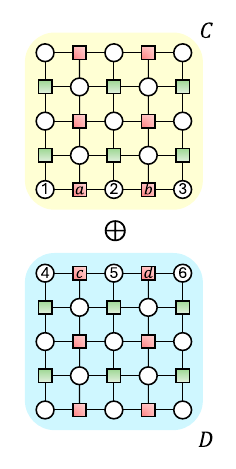}}} \arrow[two heads, r] & \vcenter{\hbox{\includegraphics{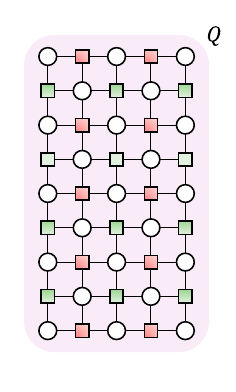}}}
      \end{tikzcd}
    \end{equation}
    Here, $(V_\bul,V^\bul)$ is a Z-subcode of $((C \oplus D)_\bul, (C \oplus   D)^\bul)$ described by
    \begin{equation}
        \begin{tikzcd}
            V_\bul : \ 0 \arrow[r] & \Span\{z_1+z_4,z_2+z_5,z_3+z_6\} \arrow[r] & \Span\{s_a+s_c, s_b+s_d\}
        \end{tikzcd}.
    \end{equation}
    So, in the Z-merged code $(Q_\bul, Q^\bul)$, as $[z_1]=[z_4]$ in $Q_1=(C_1 \oplus D_1)/\Span\{z_1+z_4,z_2+z_5,z_3+z_6\}$, the physical operators $z_1$ and $z_4$ are equals. This explains why we usually say that ``the qubits 1 and 4 have been merged together'': applying a Z pauli to the qubit 1 and merging will perform the same physical operation as applying a Z pauli to the qubit 4 and then merging.
\end{example}

\begin{example}\label{ex:merge_partial_boundary}
    We now consider the same surface code patches $(C_\bul, C^\bul)$ and $(D_\bul, D^\bul)$, but choose another Z-subcode $(V_\bul, V^\bul)$ as pictured by the orange Tanner graph below
    \begin{equation}\begin{tikzcd}
        \vcenter{\hbox{\includegraphics{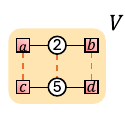}}} \arrow[hook, r] & \vcenter{\hbox{\includegraphics{figures/als_direct_sum.pdf}}} \arrow[two heads, r] & \vcenter{\hbox{\includegraphics{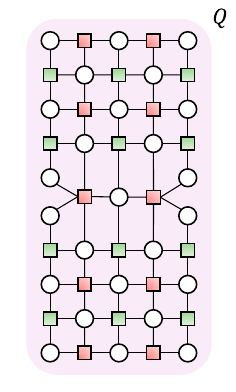}}}
    \end{tikzcd}\end{equation}
    $(V_\bul,V^\bul)$ is again a Z-subcode of $((C \oplus D)_\bul, (C \oplus   D)^\bul)$ described by
    \begin{equation}
        \begin{tikzcd}
            V_\bul : \ 0 \arrow[r] & \Span\{z_2+z_5\} \arrow[r] & \Span\{s_a+s_c, s_b+s_d\}
        \end{tikzcd}.
    \end{equation}
    We can verify that $(V_\bul,V^\bul)$ is a valid Z-subcode as $\del_1^{C \oplus D}(z_2+z_5)=s_a+s_c+s_b+s_d \in V_0$. The fact that the entire boundary is not merged has some interesting consequences for the logical operators of the Z-merged code $(Q_\bul, Q^\bul)$. We defer discussion of this to example~\ref{ex:homology_of_merge_partial_boundary}.
\end{example}

\begin{example}
\label{example:virtual_merge}
We are finally ready to give an example of a \emph{virtual} merge.
Our construction of merges in terms of chain complexes naturally captures the following
example, and many more which are impossible to represent at all in terms of
Tanner graph gluing. This makes our construction fundamentally different from code welding \cite{michnicki20123dquantumstabilizercodes} whilst generalising it. Consider the two CSS codes $(C_\bul,C^\bul)$ and $(D_\bul,D^\bul)$ defined by the chain complexes
\begin{equation}
    \left\{\begin{array}{cc}
        C_\bul :& \Span\{s_\alpha\} \xrightarrow{\begin{pmatrix}
            1 & 1
        \end{pmatrix}^T}  \Span\{z_1,z_2\} \longrightarrow  0  \\
        D_\bul :& 0 \longrightarrow  \Span\{z_3\} \xrightarrow{(1)}  \Span\{s_a\}
    \end{array}\right.
\end{equation}
with $z_1 = \begin{pmatrix}
    1 & 0
\end{pmatrix}^T$, $z_2 = \begin{pmatrix}
    0 & 1
\end{pmatrix}^T$ in $\F_2^2$, and $s_a=s_\alpha=z_3=(1) \in \F_2$. Now choose the Z-subcode $(V_\bul,V^\bul)$ of $((C\oplus D)_\bul,(C \oplus D)^\bul)$ defined by the chain complex
\begin{equation}
    V_\bul : 0 \longrightarrow  \Span\{(z_1+z_2,z_3)\} \xrightarrow{0_{\F_2^2} \oplus (1)}  \Span\{(0,s_a)\}
\end{equation}
where $(z_1+z_2,z_3) \in \F_2^2 \oplus \F_2$, $(0,s_a) \in \F_2 \oplus \F_2$ and $0_{\F_2^2} \oplus (1) = \begin{pmatrix}
    0&0&0\\
    0&0&0\\
    0&0&1
\end{pmatrix}$. Then, the corresponding quotient Z-merge $p_\bul : (C\oplus D)_\bul \rightarrow Q_\bul$ is 
\begin{equation}
    \begin{tikzcd}
        \Span\{s_\alpha\}\oplus 0 \arrow[r] \arrow[d, two heads, "p_2"] & \Span\{z_1,z_2\} \oplus \Span\{z_3\} \arrow[r] \arrow[d, two heads, "p_1"] & 0_{\F_2^2} \oplus\Span\{s_a\} \arrow[d, two heads, "p_0"] \\
        \Span\{s_\alpha\}\oplus 0 \arrow[r] & \frac{\Span\{z_1,z_2\} \oplus \Span\{z_3\}}{\Span\{(z_1+z_2,z_3)\}} \arrow[r] & 0
    \end{tikzcd}
\end{equation}
where $(Q_\bul, Q^\bul)$ is the quotient Z-merged code. To explicit $p_2$, $p_1$ and $p_0$ as matrices, we consider $\{[z_1],[z_2]\}$ as a basis of $Q_1$. Using prop.~\ref{hilb_representation_chain}, as $p_2$ is surjective, the Z-merge represents the physical operation
\begin{equation}
    \input{figures/ex_ZX0.tikz} \ = \ \input{figures/ex_ZX.tikz}.
\end{equation}
If we were to try to visualize the merging operation with Tanner graphs, it would be pictured like this
\begin{equation}
      \begin{tikzcd}
        \vcenter{\hbox{\includegraphics[scale=1.2]{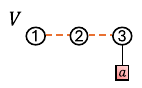}}} \arrow[r, hook] & \vcenter{\hbox{\includegraphics[scale=1.2]{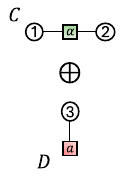}}} \arrow[r, two heads] & \vcenter{\hbox{\includegraphics[scale=1.2]{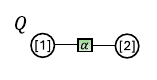}}}
      \end{tikzcd}
    \end{equation}
where the orange dashed lines in $(V_\bul,V^\bul)$ aim at picturing the Z-operator $(z_1+z_2,z_3)$. A major subtlety is that the qubit $3$ has been glued neither to qubit $1$ nor to qubit $2$. However, we did not simply discard qubit $3$ and X-check $a$. The nature of the operation to get $(Q_\bul,Q^\bul)$ from $((C\oplus D)_\bul,(C \oplus D)^\bul)$ can only be understood using chain complexes.
These new kinds of operations going beyond welding introduced in \cite{michnicki20123dquantumstabilizercodes} leave room for creating more sophisticated physical (as well as logical) operations.
\end{example}

X-merges are defined similarly to be surjective X-preserving code maps.

\begin{definition} \label{def:Xmerge}
    (X-merge) An X-merge of a CSS code $(E_\bul,E^\bul)$ is a surjective X-preserving code map $p^\bul : E^\bul \rightarrow Q^\bul$ with $(Q_\bul,Q^\bul)$ being the \textbf{X-merged code}.
\end{definition}

We can also define quotient X-merges using the dual of the prop.~\ref{def:quotientZmerge}.

\begin{definition} (Quotient X-merge)
    Let $(E_\bul, E^\bul)$ be a CSS code, and consider the X-subcode $(V_\bul, V^\bul)$ of $(E_\bul, E^\bul)$. The quotient X-merge is an X-merge $p^\bul: E^\bul \rightarrow (E/V)^\bul$ such that the X-merged code $((E/V)_\bul,(E/V)^\bul)$ is defined by the cochain complex 
    \begin{equation}
        \begin{tikzcd}
            E^2 /V^2 & \arrow[l, "\del^2_/"'] E^1/V^1 & \arrow[l,"\del^1_/"'] E^0 / V^0
        \end{tikzcd}
    \end{equation}
    where for all $n \in \{0,1,2\}$ and $x \in E^n$,
    \begin{equation}
        \del^n_/[x] := [\del^n_Ex] 
    \end{equation}
    The $p^n: E^n \rightarrow (E/V)^n$ are projections onto the equivalence class.
\end{definition}

The discussion made about the relation between Z-merges and quotient Z-merges similarly holds for X-merges and quotient X-merges.

\begin{remark}
    This construction is slightly different from the one given in \cite{Cowtan_2024}. It constrains $(V_\bul, V^\bul)$ to be a subcode of $((C \oplus D)_\bul, (C \oplus   D)^\bul)$.
    We prove in \cref{ap:eq_construction} that this construction is equivalent to the one in \cite{Cowtan_2024}, namely that all the merges constructed using general spans can be derived from a jointly monic span, i.e. using a subcode as the apex of the span.
\end{remark}

\subsubsection{Splits} \label{sec:split}

The definition of splits is very similar to the one of merges with the difference that the preserving code map must be injective.

\begin{definition}
    (Z/X-splits) A Z- (resp. X-) split of $(E_\bul,E^\bul)$ is an injective Z- (resp. X-) preserving code map $s_\bul : E_\bul \rightarrow T_\bul$ (resp. $s^\bul : E^\bul \rightarrow T^\bul$). In other words, each of the three maps $s_2,s_1$ and $s_0$ (resp. $s^2,s^1$ and $s^0$) must be injective.
\end{definition}

Splits are inherently more difficult to design than merges. For a given chain complex \( E_\bullet \), there are only finitely many chain complexes \( Q_\bullet \) and surjective chain maps \( p_\bullet : E_\bullet \to Q_\bullet \), due to dimensional constraints. In contrast, there are infinitely many possible pairs \( (T_\bullet, s_\bullet) \) where \( s_\bullet : E_\bullet \to T_\bullet \) is an injective chain map. This asymmetry can be seen from a basic dimensional argument: surjectivity of $p_\bul$ requires \( \dim Q_n \leq \dim E_n \) for $n \in \{0,1,2\}$, which bounds the number of target complexes, whereas injectivity of $s_\bul$ only requires \( \dim T_n \geq \dim E_n \), leaving infinitely many options. This echoes \cref{para:many_splits}.

So, even though a symmetric of def.~\ref{def:quotientZmerge} could be derived for splits, it would require to introduce the pending of subcodes for splits\footnote{When looking at the exact sequence in appendix~\ref{ap:merge-quotient}, we see that a subcode plays the role of the kernel of a merge. For splits, one would need to define the equivalent of subcodes for cokernels.}, that would likely be less intuitive than subcodes.
For this reason, most of the time, we prefer to restrict splits to be the dual of merges. In fact, from a Z-merge $p_\bul$ (resp. X-merge $p^\bul$), one can define the X-split $p^\bul$ (resp. Z-split $p_\bul$)~\footnote{This approach of splits is the same as the one presented in \cite{Cowtan_2024}.}.

\subsection{Homology of the merges and splits}

In the previous subsection, we established how the physical operations of merging and splitting can be described using preserving code maps. The
next step is to understand these transformations at the logical level. We
begin by considering the general case of an arbitrary \(Z\)-merge operation,
and show that results from homological algebra can be applied to make this
computation. This analysis has its limits: in the general case, the exact form of the logical
operation can only be obtained by specific analysis of the matrices involved
in the preserving code maps. However, in some specific cases, for example when the homology space
\(H_0(V_\bullet)\) is trivial, we can say a lot more about the physical
operation without requiring this matrix analysis. This setting will
be a key assumption of our implementation of the CNOT gates in \cref{sec:ft_cnot}.

\subsubsection{The general case}

Suppose that we have performed a Z-merge of $(E_\bul, E^\bul)$ along the Z-subcode
$(V_\bul, V^\bul)$ to get the Z-merged code $(Q_\bul, Q^\bul)$. We wish to
know how the logical operators of the Z-merged code $(Q_\bul, Q^\bul)$, that
is, $H_1(Q_\bul)$ and $H^1(Q_\bul)$, can be expressed in terms of those of
$(E_\bul, E^\bul)$ and the Z-subcode $(V_\bul, V^\bul)$. Unfortunately,
the relation between these spaces isn't given by a simple equality
or isomorphism. We will need more sophisticated tools from homological
algebra called \emph{exact sequences} (see prop.~\ref{exact_sequence}
for a refresher).

\begin{proposition} \label{H_merge}
    Let $(V_\bul, V^\bul)$ be a Z-subcode of $(E_\bul, E^\bul)$, and $(Q_\bul, Q^\bul)$ be the Z-merged code of $(E_\bul, E^\bul)$ along $(V_\bul, V^\bul)$. Consider the inclusion chain map $i_\bul : V_\bul \hookrightarrow E_\bul$ and the Z-merge $p_\bul$.
    The following short sequence is exact
    \begin{equation} \label{eq:short_exact_sequence}
        \begin{tikzcd}
            0_\bul \arrow[r] & V_\bul \arrow[r, "i_\bul"] & E_\bul \arrow[r, "p_\bul"] & Q_\bul \arrow[r] & 0_\bul
        \end{tikzcd}.
    \end{equation}
    It lifts to a long exact sequence on homology spaces
    \begin{equation} \label{eq:long_exact_sequence_homology}
        \begin{tikzcd}
            \ldots \arrow[r] & H_1(V_\bul) \arrow[r, "i_{1*}"] & H_1(E_\bul) \arrow[r, "p_{1*}"] & H_1(Q_\bul) \arrow[r, "\del_0"] & H_0(V_\bul) \arrow[r] & \ldots
        \end{tikzcd}.
    \end{equation}
\end{proposition}

For the purpose of this work, we do not need an explicit definition of $\del_0$. We are only interested in the fact that such a map exists.

\begin{proof}
  By definition of the inclusion map $i_\bul$ and the Z-merge (quotient map)
  $p_\bul$, for all $n \in \{0,1,2\}$, $i_n$ and $p_n$ are injective and
  surjective, respectively. Furthermore, the quotient
  is constructed such that $\ker(p_n)=\im(i_n)$, so the sequence in
  \cref{eq:short_exact_sequence} is exact. The rest of the proof follows as
  a direct consequence of theorem~\ref{Hatcher} \cite{Hatcher:478079}.
\end{proof}

\begin{remark} (Irreducibility)
    Note that for the above proposition, no assumptions about irreducibility are required to describe the logical operation performed. This is one of the biggest advantages of picturing surgery using chain complexes instead of Tanner graphs. It echos directly to \cref{para:not_only_gluings} where the physical operation cannot be captured as a Tanner graph gluing.
\end{remark}

The question that remains relates to how this exact sequence on
the homology spaces can be used to understand the logical operation performed
by the merge. We are interested in two unknowns in the exact
sequence \cref{eq:long_exact_sequence_homology}: the logical map $p_{1*}$
and the Z-logical operators of the Z-merged code $H_1(Q_\bul)$.  We can deduce
many properties of $H_1(Q_\bul)$ and $p_{1*}$ based on our knowledge of
$H_1(V_\bul)$, $H_1(E_\bul)$ and $H_0(V_\bul)$. In particular, when either
$H_1(V_\bul)=0$ or $H_0(V_\bul)=0$, the logical operation and $H_1(Q_\bul)$
can be straightforwardly inferred. For now, we provide examples of how to
use the exact sequence \cref{eq:long_exact_sequence_homology} based on the
examples of merges given in the previous section. Subsection~\ref{sec:H_-1=0}
will then give a detailed analysis in the case when $H_0(V_\bul)=0$ (and by duality, when $H^2(V_\bul)=0$).

\begin{example} (Following example~\ref{ex:merge_als})
    Let us write $Z_C = [z_1+z_2+z_3]$ and $Z_D = [z_4+z_5+z_6]$ as the Z-logical operators of $(C_\bul, C^\bul)$ and $(D_\bul, D^\bul)$, respectively. As $\im(\del_1^{C \oplus D} \restriction_{V_1}) = V_0$, we have $H_0(V_\bul)=V_0/\im(\del_1^V)=0$. So, according to the exact sequence, $p_{1*}$ is surjective. 
    
    Moreover, $(V_\bul, V^\bul)$ contains a Z-logical operator, $Z_C+Z_D$, that is non-trivial when embedded in the Z-logical operators of $((C \oplus D)_\bul, (C \oplus D)^\bul)$, that is, it is not a stabiliser. In the language of homology, this means that $Z_C+Z_D \in H_1(V_\bul)$ and $i_{1*}(Z_C+Z_D) \neq [0]$. As a matter of facts, it is the only one in the image: $\im(i_{1*})=\Span \{Z_C + Z_D\}$\footnote{$Z_C$ and $Z_D$ are to be considered as vectors of $H_1((C \oplus D)_\bul)$.}. Therefore, as $\ker(p_{1*}) = \im(i_{1*})$, we deduce that $p_{1*}(Z_C+Z_D)=[Z_C] + [Z_D] = [0]$ in $H_1(Q_\bul)$. This is the usual relation of lattice surgery: $[Z_C] = [Z_D]$.
\end{example}

\begin{example} \label{ex:homology_of_merge_partial_boundary} (Following example~\ref{ex:merge_partial_boundary})
    In this example, the subcode $(V_\bul,V^\bul)$ does not have any Z-logical operator, which translates in $H_1(V_\bul) = 0$. It implies that $p_{1*}$ is injective, so $\dim H_1((C \oplus D)_\bul) \leq \dim H_1(Q_\bul)$. Moreover, $H_0(V_\bul) \neq 0$. It implies that $\dim H_1((C \oplus D)_\bul) < \dim H_1(Q_\bul)$, hence, a Z-logical operator has been created by the merge. 
    
    This new Z-logical operator can easily be visualised in the Tanner graph. One of its representatives is pictured by the green line below
    \begin{equation}            
        \vcenter{\hbox{\includegraphics{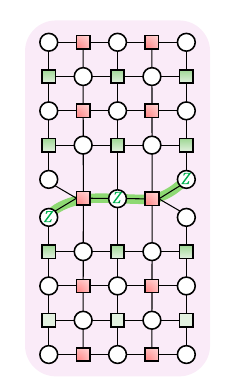}}}
    \end{equation} 
    So, a merge does not only quotient logical operators, it can also create some new ones.
\end{example}
    
\begin{example}
    We have seen with the two previous examples that in some cases, we are able to predict which logical operators appear or vanish. This allows us to give a basis of Z-logical operators of $(Q_\bul, Q^\bul)$, and in some simple cases, to relate this basis with the one of $((C \oplus D)_\bul, (C \oplus D)^\bul)$, thus describing the logical operation performed. However, it is not always as easy. For instance, consider the Z-merge where $H_0(V_\bul) \neq 0$ and $H_1(V_\bul) \neq 0$:
    \begin{equation}\begin{tikzcd}
        \vcenter{\hbox{\includegraphics{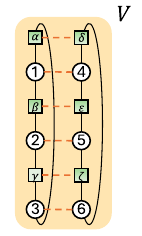}}} \arrow[hook, r] & \vcenter{\hbox{\includegraphics{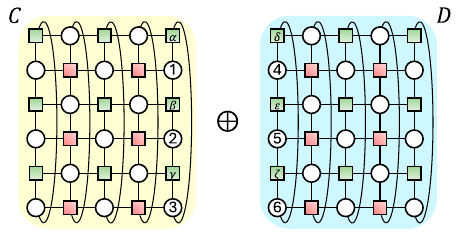}}} \\
        \arrow[two heads, r] & \vcenter{\hbox{\includegraphics{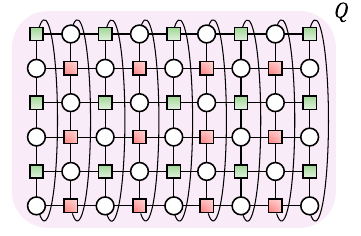}}}
    \end{tikzcd}\end{equation}
    The codes $(C_\bul, C^\bul)$ and $(D_\bul, D^\bul)$ are ``cylinder codes''. The result of the merge is a bigger cylinder code. As both $H_0(V_\bul) \neq 0$ and $H_1(V_\bul) \neq 0$, we cannot make any direct statement about the relation between $H_1((C \oplus D)_\bul)$ and $H_1(Q_\bul)$. Conducting a more in depth analysis of $p_{1*}$ is then necessary.
\end{example}

Although the exact sequence \cref{eq:long_exact_sequence_homology} does not suffice alone to determine $H_1(Q_\bul)$ and $p_{1*}$ in general, it still provides some valuable information on the logical operators of $(Q_\bul,Q^\bul)$. We can classify the Z-logical operators in $H_1(Q_\bul) = \im(p_{1*}) \oplus \ker(\del_0)^\perp$.
The ones in $\im(p_{1*})$ are inherited from $H_1(E_\bul)$, whereas the ones in $\ker(\del_0)^\perp$ are created by the merge. 

It can be easily pictured using def.~\ref{def:log_hilb_rep}, namely the representation of $p_{1*}$ on the logical Hilbert space. The Z-logical operators in $\im(i_{1*}) = \ker(p_{1*})$ vanish while some new X-logical operators in $\im(\del^0) = \ker(p^1_*)$ appear. In fact, for $[u] \in \im(i_{1*})$ and $[v] \in \im(\del^0)$, we have
\begin{align}
    \input{p_0star_left_true.tikz}  = \input{p_0star_right2.tikz} = \input{p_0star.tikz} \\
    \input{p_0star.tikz} = \input{p_0star_leftred.tikz} = \input{p_0star_rightred.tikz}.
\end{align}
However, the deletion or creation of X- or Z-logical operators does not necessarily imply the addition or removal of logical qubits. There are two specific cases where we can confidently assert that logical qubits have been added or removed, that follow straightforwardly from \cref{eq:long_exact_sequence_homology}:
\begin{corollary}
    \label{lem:surjective_logical}
    The logical operation $p_{1*}$ is surjective if and only if $H_0(V_\bul) = 0$, and then $\dim H_1((C \oplus D)_\bul) \geq \dim H_1(Q_\bul)$; the number of logical qubits may have decreased.
\end{corollary}
\begin{corollary}
    \label{lem:injective_logical}
    The logical operation $p_{1*}$ is injective if and only if $H_1(V_\bul) = 0$, and then $\dim H_1((C \oplus D)_\bul) \leq \dim H_1(Q_\bul)$; the number of logical qubits may have increased.
\end{corollary}
If neither of these conditions is met, we cannot determine the number of logical qubits after the merge through the mere existence of the exact sequence of \cref{eq:long_exact_sequence_homology}. By duality, all the results previously established hold similarly for X-merges.

\begin{proposition} \label{Hom_X_merge}
    Let $(V_\bul,V^\bul)$ be an X-subcode of $(E_\bul, E^\bul)$, and $(Q_\bul, Q^\bul)$ be the X-merged code of $(E_\bul, E^\bul)$ along $(V_\bul, V^\bul)$. Consider $i^\bul : V^\bul \hookrightarrow E^\bul$ the inclusion cochain map, and the X-merge $p^\bul$. Then, the following short exact sequence 
    \begin{equation}
        \begin{tikzcd}
            0^\bul \arrow[r] & V^\bul \arrow[r, "i^\bul"] & E^\bul \arrow[r, "p^\bul"] & Q^\bul \arrow[r] & 0^\bul
        \end{tikzcd}
    \end{equation}
    lifts to a long sequence on the cohomology spaces
    \begin{equation}\begin{tikzcd}
        \ldots \arrow[r] & H^1(V_\bul) \arrow[r, "i^1_*"] & H^1(E_\bul) \arrow[r, "p^1_*"] & H^1(Q_\bul) \arrow[r, "\del^2"] & H^2(V_\bul) \arrow[r] & \ldots
    \end{tikzcd}\end{equation} 
\end{proposition}

\begin{remark}
    (Application to code switching) Code switching is a technique for implementing a transversal universal gate set originally introduced by Paetznick and Reichardt in \cite{Paetznick_2013} and then demonstrated experimentally in \cite{pogorelov2024experimentalfaulttolerantcodeswitching}. Recall that transversal gates are natively fault-tolerant and that there is no code having a transversal universal gate set \cite{Eastin_2009}. However, using two codes with complementary transversal gates, if we can switch from one to the other fault-tolerantly, we can achieve a transversal universal gate set. The formalism here helps to design the switch. Consider the first code $(C_\bul, C^\bul)$. Adding some auxiliary qubits encoded with $(D_\bul, D^\bul)$, we can merge them along a subcode $(V_\bul, V^\bul)$ to get the second code $(Q_\bul, Q^\bul)$. Then, we can track the logical operation performed using prop.~\ref{H_merge}. Although this approach does not make easier the search for family of codes to switch between, it greatly facilitates the analysis of the logical operation induced by the switch. We provide in \cref{appendix:code_switching} a replication of a variant of the switch protocol between the Reed-Muller code $\llbracket 15,1,3 \rrbracket$ and the Steane code $\llbracket 7, 1, 3\rrbracket$ \cite{Anderson_2014} proposed in \cite{heußen2024efficientfaulttolerantcodeswitching}.
\end{remark}

\subsubsection{The case $H_0(V_\bul)=0$}\label{sec:H_-1=0}

As it will facilitate the derivation of the CNOT protocol presented in the next section, we expand more on the case $H_0(V_\bul)=0$. The exactness of the homology sequence implies that $\im(i_{1*}) = \ker(p_{1*})$ and $\im(p_{1*}) = H_1(Q_\bul)$, hence
\begin{equation}
    H_1(Q_\bul) \cong H_1(E_\bul)/ \im(i_{1*}).
\end{equation}
We assume that a basis $\mathcal{B} = \{u_1, \ldots, u_k\}$ of $H_1(E_\bul)$ is already fixed. The challenge is then to find a basis of $H_1(Q_\bul)$ in which the matrix $P_{1*}$ takes a simple form.

Let $\{v_1, \ldots, v_r\}$ be an echelonized and reduced basis of $\im(i_{1*})$. This means that the matrix in which the $i$-th row is the decomposition of $v_i$ in the basis $\mathcal{B}$ should take the following form:
\begin{equation}M=\left( \begin{array}{ccccccc}
    1 & 0 & * & 0 & * & \ldots & * \\
    \cline{1-1}
    0 & 1 & * & 0 & * & \ldots & * \\
    \cline{2-3}
    0 & 0 & 0 & 1 & * & \ldots & * \\
    \cline{4-5}
         & & & & \ddots & & \\
    0 & 0 & 0 & 0 & 0 & 1 & * 
\end{array} \right).\end{equation}
Each column containing a pivot is full of zeros except at the entry of the pivot.    
Let $I$ be the set of indices of the $r$ pivots. Then, it is straightforward that $\mathcal{C}' = \{[u_i] : i \in \llbracket 1,k \rrbracket \setminus I \}$ is a basis of $H_1(E_\bul)/ \im(i_{1*})$.
Finally, applying $p_{1*}$ to the representative of each equivalence class raises a similar basis $\mathcal{C} = \{[u_i] : i \in \llbracket 1,k \rrbracket \setminus I \}$ of $H_1(Q_\bul)$~\footnote{Applying $p_{1*}$ to $[u_i] \in \mathcal{C}'$ gives $[p_{1*}(u_i)] = [u_i] \in H_1(Q_\bul)$}. 

Henceforth, as $P_{1*}=Mat_\mathcal{C}(\mathcal{B})$ projects the basis vectors of $\mathcal{B}$ onto their equivalence class written in the basis $\mathcal{C}$, $P_{1*}$ resembles the identity except at each column $i \in I$ where the $1$ entries are given by the decomposition of the vectors $v_i$'s in $\mathcal{B}$.

To make this fairly abstract construction clearer, let us consider a concrete example. Consider $\mathcal{B} = \{u_1, \ldots, u_4\}$ and the echelonized and reduced basis $\{v_1 = u_1+u_3, v_2 = u_2+u_3+u_4\}$. The decomposition of these vectors in $\mathcal{B}$ gives the matrix
\begin{equation}M=\begin{pmatrix}
    1 & 0 & 1 & 0 \\
    0 & 1 & 1 & 1
\end{pmatrix}\end{equation}
The pivots are $u_1$ and $u_2$, so $I=\{1,2\}$. In other words, in $H_1(E_\bul)/\im i_{1*}$, $[u_1]=[u_3]$ as $[v_1]=0$, and $[u_2]=[u_3]+[u_4]$ as $[v_2]=0$. So, $[u_1]$ and $[u_2]$ can be obtained from $[u_3]$ and $[u_4]$. That is why we choose the basis $\mathcal{C} = \{[u_3], [u_4]\}$ of $H_1(Q_\bul)$ such that
\begin{equation}
    P_{1*} = \begin{pmatrix}
        1 & 1 & 1 & 0 \\
        0 & 1 & 0 & 1
    \end{pmatrix}.
\end{equation}
    
\section{Designing a logical CNOT with fault-tolerance guarantees} \label{sec:ft_cnot}

In this section, we leverage the theoretical tools previously introduced to construct a logical CNOT gate between arbitrary logical qubits in any CSS code. We believe that the main advantage of CSS surgery is to navigate between CSS codes while keeping track of the logical operations. In this way, we can exploit the advantages of each code we pass through to perform a desired computation. This approach is commonly known in the literature and has recently been shown to be successful \cite{pogorelov2024experimentalfaulttolerantcodeswitching, Davydova_2024}. We now propose a highly general protocol for implementing a logical CNOT which is rigid in the sense that we aim to return to the same code we started with. This final section serves to showcase how the subcode surgery formalism we have developed can be applied in practice.

The setup we consider is as follows: let $(C_\bul, C^\bul)$ be a CSS code for which logical qubits are defined with respect to the dual bases of logical operators $\mathcal{B}_C=\{ [z_i] \}_{1 \leq i \leq k}$ of $H_1(C_\bul)$ and $\mathcal{B}^C=\{ [x_j] \}_{1 \leq j \leq k}$ of $H^1(C_\bul)$ (i.e. $[z_i] \cdot [x_j] = \delta_{i,j}$, cf. \cref{seq:logop}). The goal is to find a physical operation that implements a logical CNOT between any pair of logical qubits. Without loss of generality, we will consider the first logical qubit to be the control, and the second to be the target. This problem turns out to be challenging for many reasons. In addition to finding a physical operation that implements a logical CNOT, we must also ensure that it does not interfere with any other logical qubits. In other words, among the logical qubits of $(C_\bul,C^\bul)$, \emph{only} the first (target) and second (control) ones should undergo a logical operation. Furthermore, we desire our implementation to have some guarantee of fault-tolerance.

To achieve our ends, we use CSS surgery to implement a decomposition of the CNOT gate---identical to the one used in the standard lattice surgery protocol \cref{eq:cnot_ls}
\begin{equation} \label{eq:cnot2}
  \input{cnot2.tikz}.
\end{equation}
There are three main challenges:
\begin{itemize}
    \item First, we need to \textbf{introduce an auxiliary logical qubit}. Our goal is to minimise the overhead in auxiliary physical qubits required to introduce the auxiliary logical qubit.
    \item Second, we need to \textbf{choose subcodes for both merges} which guarantee that the overall logical operation performed by the succession of merges and splits will result in \cref{eq:cnot2}.
    \item Finally, the merges and splits should be \textbf{fault-tolerant} operations.
\end{itemize} 

As the splits in \cref{eq:cnot2} should be the transpose operation of the merges, all the splits will be defined to be the dual preserving code map of the previous merge. 
Furthermore, by corollary~\ref{lem:surjective_logical}, and noting that the logical merges of eqs.~\eqref{eq:physical_surgery} and \eqref{eq:physical_surgery_X} (defined in \cref{eq:logical_surgery}) are surjective maps, we only need to consider Z- (resp. X-) subcodes $(V_\bul, V^\bul)$ (resp. $(W_\bul, W^\bul)$) satisfying $H_0(V_\bul)=0$ (resp. $H^2(W_\bul)=0$). This greatly simplifies the analysis of the logical operation performed by the merge following the methodology outlined in \cref{sec:H_-1=0}.

In this section, first, we introduce a naive non-fault-approach to implementing a logical CNOT, with the aim of building intuition about how to introduce the auxiliary logical qubit and what subcodes to choose to perform the desired merges shown in \cref{eq:cnot2}. Following this, we outline a general protocol that aims to minimise the number of physical qubits while ensuring guarantees of fault-tolerance. \\

\subsection{The non-fault-tolerant approach} \label{sec:naive_approach}

In the first instance, we ignore fault-tolerant design considerations. The main goal is to achieve the desired logical operation, namely, the CNOT. The protocol is pictured in \cref{fig:cnot2} a).

\begin{figure}
    \centering
    \input{surg_cnot.tikz}
    \caption{\textbf{Surgery logical $CNOT$ with different auxiliary codes.} 
    The general structures of the Tanner graphs are illustrated for three different approaches to introduce an auxiliary logical qubit, each intended to support the protocol of \cref{eq:cnot2}.  
    \textbf{a)} The auxiliary logical qubit corresponds directly to a physical qubit, resulting in a trivial encoding.  
    \textbf{b)} An auxiliary code is introduced that hosts a single logical qubit. Unlike in (a), the minimal-weight representatives of the auxiliary code’s X- and Z-logical operators are expected to have weight at least equal to the distance of the original code \((C_\bullet, C^\bullet)\), improving protection against errors.  
    \textbf{c)} To mitigate the physical qubit overhead incurred in b), the auxiliary logical qubit is embedded within the original code \((C_\bullet, C^\bullet)\). This can be achieved either by starting from a code that already contains an additional logical qubit, or by augmenting \((C_\bullet, C^\bullet)\) via CSS surgery to introduce a new logical qubit with minimal added physical resources.
    }
    \label{fig:cnot2}
\end{figure}

\paragraph{The auxiliary logical qubit.}
The very first step is to choose a way to introduce an auxiliary logical qubit. As a first approach, we consider this logical qubit to be encoded in a trivial auxiliary CSS code $(A_\bul,A^\bul)$ composed of one physical qubit \textit{a} without any stabilisers. This code encodes one logical qubit with logical operators $z_a \in H_1(A_\bul)=A_1$ and $x_a \in H^1(A_\bul)=A^1$ corresponding to the physical Z and X Paulis acting on qubit \textit{a}. It is instantiated in the $\ket{+}$ state according to \cref{eq:cnot2}.

\paragraph{The Z-merge/X-split.}
Now, we have to choose a Z-subcode $(V_\bul, V^\bul)$ to define the Z-merged code $(Q_\bul, Q^\bul)$ and the Z-merge/X-split $p_\bul$ and $p^\bul$. According to def.~\ref{def:subcode} of a Z-subcode, we need to pick three subspaces $V_2,V_1$ and $V_0$ of $(C \oplus A)_2,(C\oplus A)_1$ and $(C\oplus A)_0$, respectively, satisfying the subcode condition $\del_n^C(V_n) \subset V_{n-1}$. For simplicity, we fix $V_2=0$. This choice aims at leaving the group of Z-stabilisers unchanged and does not hinder the implementation of the desired logical operation for the Z-merge \cref{eq:cnot2}. To satisfy the constraint $H_0(V_\bul)=V_0/\im(\del_1^V)=0$, knowing that $\del_1^V=\del_1^C \restriction_{V_1}$, we have to choose $V_0=\del_1^C(V_1)$. So, once $V_1$ is chosen, $V_0$ is fixed~\footnote{This $V_0$ has one major drawback: the weight of the X-stabiliser generators of the Z-merged code might increase a lot compared to the ones of $(C_\bul, C^\bul)$ and $(A_\bul, A^\bul)$. This does not represent an issue in our protocol as we implement the Z-merge and the X-split in a row, without performing a round of error correction on the code $(Q_\bul, Q^\bul)$.}. It remains to choose $V_1$. This choice is dictated by the logical operation we are aiming for with the Z-merge.

According to \cref{sec:H_-1=0}, the logical operation of the Z-merge is a quotient map $p_{1*} : H_1((C \oplus A)_\bul) \rightarrow H_1(Q_\bul)$ where $H_1(Q_\bul) \cong H_1((C \oplus A)_\bul)/\im(i_{1*})$ and $i_{1*} : H_1(V_\bul) \rightarrow H_1((C \oplus A)_\bul)$ is an embedding. Hence, according to def.~\ref{def:log_hilb_rep}, to achieve the ZX-diagram of the Z-merge \cref{eq:cnot2}, we need to choose $V_1$ such that $\im(i_{1*})=\Span\{[z_1]+z_a\}$. This way, $p_{1*}$ is the identity on all Z-logical operators of $(C_\bul, C^\bul)$ except for $[z_1]$ which is identified with $z_a$. In fact, if we choose $\mathcal{B}_Q=\{[z_1]\sim z_a, [z_2], \ldots , [z_k]\}$ as a basis of $H_1(Q_\bul)$, $P_{1*}$ takes the form
\begin{equation}
    P_{1*} = Mat_{(\mathcal{B}_C \cup \{z_a\},\mathcal{B}_Q)}(p_{1*})= \left(
\begin{array}{cccc|c}
1 & 0 & \ldots & 0 & 1 \\
0 & 1 & \ldots & 0 & 0 \\
\vdots  &   & \ddots &   & \vdots \\
0 & 0 & \ldots & 1 & 0
\end{array}
\right).
\end{equation}
Considering the logical interpretation of $p_{1*}$ (def.~\ref{def:log_hilb_rep}), the logical operation induced by the Z-merge $p_\bul$ corresponds to the desired ZX-diagram \cref{eq:cnot2}.
There are \emph{many} choices of $V_1$ that lead to $\im(i_{1*})=\Span\{[z_1]+z_a\}$.
Let us make the most straightforward choice, $V_1=\Span\{z_1+z_a\}$. To summarise, we choose the Z-subcode $(V_\bul, V^\bul)$ given by
\begin{equation}
\begin{tikzcd}
    V_\bul : 0 \arrow[r] & \Span\{z_1+z_a\} \arrow[r] & 0
\end{tikzcd}.
\end{equation}
Because both $z_1$ and $z_a$ are representatives of Z-logical operators, we have $\del_1^C(V_1)=0$. It is easily verified that $z_1+z_a \in H_1(V_\bul)$ and that $i_{1*}(z_1+z_a)=[z_1]+z_a \neq 0$ in $H_1((C\oplus A)_\bul)$ as $z_1+z_a \notin \im(\del_2^C) \oplus\im(\del_2^A)$. In other words, $z_1+z_a$ is not a stabiliser of $((C\oplus A)_\bul,(C\oplus A)^\bul)$. Hence,  $\im(i_{1*})=\Span\{[z_1]+z_a\}$ as desired. 

\begin{remark} \label{rem:ft_code}
    We stress that \emph{the choice of $V_1$ is not unique}. Different choices lead to
    \emph{different physical implementations} whilst implementing a fixed
    logical operation. That is why from a fault-tolerance perspective,
    studying different choices of $V_1$ may be of interest, especially for
    limiting the spread of errors. This is discussed in \cref{sec:ft}.
\end{remark}

On one hand, according to props.~\ref{hilb_representation_chain} and \ref{hilb_representation_cochain}, because $p_2$ is surjective (as it is a quotient map), the physical operation corresponding to the sequential Z-merge $p_\bul$/X-split $p^\bul$ is described by the SZX-diagram:
\begin{equation} \label{eq:Z_merge-split}
    \input{Z_merge-split_SZX.tikz}.
\end{equation}
Note that for the specific choice $V_1=\Span\{z_1+z_a\}$, the final projection is not required as $V_0=0$ and so $p^0$ is surjective (cf. the second remark below prop.~\ref{hilb_representation_chain}).
On the other hand,
we achieve the desired logical operation:
\begin{equation}
    \input{Z_merge-split_SZX_logical.tikz}
    = \input{Z-spider_cnot.tikz}.
\end{equation}
We postpone the analysis of the physical operation to \cref{sec:ft} and \cref{sec:postselection}.

\paragraph{The X-merge/Z-split}
To complete the diagram \cref{eq:cnot2} with the X-spiders, we need an X-subcode $(W_\bul,W^\bul)$ to define the X-merged code $(R_\bul,R^\bul)$ and the X-merge/Z-split $q^\bul$ and $q_\bul$. All the choices made have similar motivations: we fix for simplicity $W^0=0$ and to guarantee that $H^2(W_\bul)=0$ we pick $W^2=\del^2_C(W^1)$. Thus, $H^1(W_\bul) \cong H^1((C \oplus A)_\bul)/\im(j^1_*)$ with $j^1_* : H^1(W_\bul) \rightarrow H^1((C \oplus A)_\bul)$ being the embedding map. As again we aim for $\im(j^1_*)=\Span\{[x_2]+x_a\}$, many choices of $W^1$ are valid but not equivalent from a fault-tolerance perspective. We pick the simplest $W^1$ such that
\begin{equation}
    \begin{tikzcd}
        W^\bul : 0 & \arrow[l] \Span\{x_2+x_a\} & \arrow[l] 0
    \end{tikzcd}.
\end{equation}
The X-merge $q^\bul$ results in the map $q^1_*$ at the logical level, which, given the basis $\mathcal{B}^R=\{[x_1], [x_2]\sim x_a, \ldots , [x_k]\}$ of $H^1(R_\bul)$, gives the matrix
\begin{equation}
    Q_*^1 = Mat_{(\mathcal{B}^C \cup \{x_a\},\mathcal{B}^R)}(q^0_*)= \left(
\begin{array}{cccc|c}
1 & 0 & \ldots & 0 & 0 \\
0 & 1 & \ldots & 0 & 1 \\
\vdots  &   & \ddots &   & \vdots \\
0 & 0 & \ldots & 1 & 0
\end{array}
\right)
\end{equation}

So, applying the physical operation described by the ZX-diagram
\begin{equation}\label{eq:X_merge-split}
    \input{X_merge-split_SZX.tikz},
\end{equation}
we achieve the desired logical operation:
\begin{equation}
    \input{X_merge-split_SZX_logical.tikz}=\input{X-spider_cnot.tikz}.
\end{equation}
To complete the CNOT, we measure the log-ical qubit of $(A_\bul, A^\bul)$ in the Z-basis. In \cref{eq:cnot2}, the measurement outcome is assumed to be positive, but in the event of a negative outcome, we can simply apply a logical X on the second logical qubit of $(C_\bul,C^\bul)$, namely $[x_2]$.

\subsection{Fault-tolerance of the physical implementation} \label{sec:ft}

A rigorous analysis of the fault-tolerance of the protocol turns out to be challenging. Essentially, the set of potential errors happening during the protocol can differ a lot depending on the physical implementation. The latter can change in two ways:
\begin{itemize}
    \item First, using ZX-calculus, one can change the physical implementation while realising the same operation. This way, we can choose the implementation best suited for a given hardware.
    \item Second, as pointed out in remark~\ref{rem:ft_code}, depending on the code $(C_\bul,C^\bul)$ (and $(A_\bul, A^\bul)$), one may choose another subcode $(V_\bul, V^\bul)$ improving in some way the physical implementation non-trivially (meaning that is not simply a former ZX-diagram rewritten differently).
\end{itemize}
That's why in this section we focus on presenting aspects of the protocol that guarantee or can be used to guarantee the fault-tolerance. These are three-fold: first we argue that no matter the specific physical implementation we choose among all the possible ones, it does not propagate pre-existing Pauli errors. Second, we discuss how to improve the introduction of the auxiliary logical qubit. Finally, we showcase how in specific cases, we can perform the merges/splits using only local operations in parallel, implying that internal Pauli errors do not propagate drastically.

\subsubsection{Pre-existing Pauli errors do not propagate}

Although it is difficult to analyse the propagation of Pauli errors occurring within the physical implementation of the merges and splits in full generality, the case of pre-existing Pauli errors is tractable.
\begin{proposition}
    Every Pauli error commutes with the Z-merge/X-split and with the X-merge/Z-split.
\end{proposition}
This property inherited from the SZX-calculus is particularly interesting for fault-tolerance as it tells us that the protocol can be performed on errored qubits without making the error worse.

\begin{proof}
    Let us start with the SZX-diagram \cref{eq:Z_merge-split}. Every Z-Pauli commutes with this diagram and flips the probability outcomes of the final stabiliser measurements. In fact, for every $v \in \F_2^n$,
    \begin{equation}
        \input{Z_merge-split_commutation1.tikz} = \input{Z_merge-split_commutation2.tikz} = \input{Z_merge-split_commutation3.tikz}.
    \end{equation}
    The commutation for X-Pauli errors is less straightforward.
    As $p_\bul$ is a Z-merge, the matrix $p_1 \in \F_2^{m \times n}$ is surjective. So, adding rows to this matrix gives the bijection $\Tilde{p}_1 \in \F_2^{n \times n}$. In other words, we retrieve $p_1=\iota \Tilde{p}_1$ where $\iota = \text{diag}(1, \ldots, 1, 0, \ldots, 0)$ is the inclusion matrix containing $m$ ones on the first entries of the diagonal. So, for any $u \in \F_2^n$, let $u_0 \in \F_2^n$ such that $\Tilde{p}^1u_0=u$. We get
    \begin{equation}
        \input{X_merge-split_commutation1.tikz} = \input{X_merge-split_commutation2.tikz} = \input{X_merge-split_commutation3.tikz}
    \end{equation}
    where $\iota^c = id-\iota$. As a consequence, the X-error flips the probabilities of the measurement outcomes. Hence, with post-selected measurement, it does not have any effect and commutes with the diagram.

    The same reasoning can be applied to \cref{eq:X_merge-split} by swapping the colours.
\end{proof} 

\subsubsection{Auxiliary code distance and fault-tolerance}

We can reasonably assert that the protocol described in the previous section suffers from a major flaw: any physical error occurring on the auxiliary qubit outside of the merges and splits will effectively apply a \emph{logical} error: a Z- (resp. X-) error will propagate to a Z- (resp. X-) logical error on the first (resp. second) logical qubit, namely $[z_1]$ (resp. $[x_2]$).
The reason for this is that a physical Pauli error on the auxiliary qubit amounts to a logical error in the trivial code $(A_\bullet, A^\bullet)$. The consequence of such a logical error on the logical data of $(C_\bullet, C^\bullet)$ can immediately be assessed using \cref{eq:cnot2}. The following ZX-diagram describes how such physical error would propagate on the logical data.
\begin{equation} \label{eq:cnot3}
  \begin{aligned}
    &\input{cnot3.tikz} \\
    &\hspace{1cm}= \quad \input{cnot4.tikz} \quad = \quad \input{cnot5.tikz}
  \end{aligned}
\end{equation}

To alleviate this problem, the auxiliary code $(A_\bullet, A^\bullet)$ must be selected with a distance at least as large as that of $(C_\bul, C^\bul)$. This ensures that the probability of a logical error occurring in the auxiliary code is not greater than the likelihood of a logical error in $(C_\bul, C^\bul)$. 
That is why we present two alternative methods for introducing an auxiliary logical qubit, as illustrated in \cref{fig:cnot2} b) and c). 

The first approach (\cref{fig:cnot2}b) involves constructing a \textbf{larger auxiliary code} by using more physical qubits to increase the distance. While this straightforward method guarantees a sufficient code distance for \( (A_\bullet, A^\bullet) \), it may prove costly in terms of qubit overhead. This can be justified in two ways:
\begin{itemize}
    \item \textit{``Good'' codes with a single logical qubit do not exist.} Intuitively, error correction codes benefit from a scaling effect: as more logical (qu)bits are encoded, fewer physical (qu)bits are needed per logical (qu)bit to achieve a significant code distance. This is analogous to classical coding theory, where the best $[n,1,d]$ code for encoding a single logical bit remains the repetition code. In the quantum setting, finding a family of CSS codes with a single logical qubit, tunable distance, and LDPC properties that uses significantly fewer physical qubits than the surface code (i.e. $d^2+(d-1)^2$ for a distance $d$) is challenging---if not impossible---depending on the desired sparsity of the parity check matrix.
    \item \textit{The overall code $((C\oplus A)_\bullet, (C \oplus A)^\bullet)$ cannot be optimal.} Error correction relies on a trade-off between the rate ($\frac{k}{n}$), and the distance per physical qubit ($\frac{d}{n}$), making difficult to define the optimality of a code. However, we argue in appendix~\ref{ap:singleton_bound} that regardless of the choice of the auxiliary code $(A_\bullet, A^\bullet)$, the overall code $((C\oplus A)_\bullet, (C \oplus A)^\bullet)$ cannot be an MDS (Maximal Distance Separable) code \cite{rains1997nonbinaryquantumcodes}, meaning it cannot achieve optimality in the sense of the Singleton bound.
\end{itemize}

We propose a final additional method to reduce the size of the auxiliary code that involves \textbf{embedding the auxiliary code \( (A_\bullet, A^\bullet) \) directly into \( (C_\bullet, C^\bullet) \)}, as depicted in \cref{fig:cnot2}c. This can be achieved in two ways: by initially selecting a code \( (C_\bullet, C^\bullet) \) with an additional logical qubit or by augmenting the existing code through a CSS surgery operation to introduce a new logical qubit.  In practice, the viability of this approach depends on the code \( (C_\bullet, C^\bullet) \) considered, making this method fairly abstract. However, the idea can always be considered. In the next example, for simplicity, we consider \( (C_\bullet, C^\bullet) \) to be a toric code.

\begin{example}
    A third logical qubit can be introduced in the toric code by creating a hole. This process can be interpreted as a surgery operation, where the subcode corresponds to the sub-Tanner graph formed by the removed qubits and checks.
\end{example}

The main challenge with this approach is ensuring that the introduction of an auxiliary logical qubit does not (significantly) degrade the code distance. This opens an interesting research direction: the development of \textbf{modular CSS codes}—CSS codes designed to accommodate additional logical qubits via CSS surgery while preserving distance, all with minimal physical qubit overhead.

Last but not least, it is important to remark that all the analysis conducted in \cref{sec:naive_approach} is still valid even though $z_a$ and $x_a$ are now vectors of larger size than one.

\subsubsection{Choosing $V_1$ and $W^1$ to guarantee only local interactions} \label{sec:physical_implementation}

In this section, we elaborate on remark~\ref{rem:ft_code}.
The main way to prevent the spread of internal errors is to limit as much as possible the number of physical qubits which interact. If an error occurs on a single qubit that interacts with many other physical qubits through, say CNOTs, the single qubit error is likely to turn into a high weight error. We argue that with appropriate choices of $V_1$ in section~\ref{sec:naive_approach}, the Z-merge/X-split can be chosen to consist only of parallel local interactions. By the Z/X-duality of (co)chain complexes, an analogous conclusion can, of course, be reached for $W^1$ and the X-merge/Z-split.

To make the statement more explicit, let us first examine the straightforward choice $V_1=\Span\{z_1+z_a\}$ made in section~\ref{sec:naive_approach}. Although this choice always achieves the desired logical operation, it typically induces a physical implementation that turns single internal Pauli errors into very height weight Pauli errors. We start by expliciting \cref{eq:Z_merge-split}.

Recall that $p_1: C_1 \oplus A_1 \rightarrow (C_1 \oplus A_1)/span\{z_1+z_a\}$ is a projection map. We first need to turn it into a matrix by defining a basis of $(C_1 \oplus A_1)/span\{z_1+z_a\}$. 

In the general case (illustrated in \cref{fig:cnot2}b and c), the vector \( z_a \) does not act on a single physical qubit but instead has support on multiple qubits.
Let \( \mathcal{C} = \{ \delta_1, \ldots, \delta_{n_C + n_A} \} \) be the canonical basis of \( C_1 \oplus A_1 \cong \mathbb{F}_2^{n_C + n_A} \), where each \( \delta_i \) is the standard basis vector with a 1 in the \( i \)-th position and 0 elsewhere.
Without loss of generality (up to a permutation of the physical qubits), we assume that:
\begin{itemize}
    \item \( \delta_1 \cdot z_a = 1 \), i.e., the first entry of \( z_a \) is non-zero,
    \item \( \delta_1 \cdot z_1 = 0 \), i.e., the first entry of \( z_1 \) is zero.
\end{itemize}
This implies that the vector \( z_1 + z_a + \delta_1 \) has a zero in its first entry and therefore lies in the span of \( \{ \delta_2, \ldots, \delta_{n_C + n_A} \} \). Consequently, the class \( [\delta_1] \in (C_1 \oplus A_1)/\mathrm{span}\{ z_1 + z_a \} \) is equal to \( [z_1 + z_a + \delta_1] \), which belongs to the span of \(\left\{ [\delta_2], \ldots, [\delta_{n_C + n_A}] \right\}\).
Hence, we can choose the set \( \mathcal{C}' = \{ [\delta_2], \ldots, [\delta_{n_C + n_A}] \} \) as a basis for the quotient space
\((C_1 \oplus A_1)/\mathrm{span}\{ z_1 + z_a \} \cong \mathbb{F}_2^{n_C + n_A - 1}\).
With respect to the bases \( \mathcal{C} \) and \( \mathcal{C}' \), the projection map \( p_1 \) can be written in matrix form as
\begin{equation}\label{eq:p_1_general}
p_1 =
\begin{pmatrix}
\overline{z_1 + z_a + \delta_1} \ \Big| \
I_{n_C + n_A - 1}
\end{pmatrix}.
\end{equation}
Here, \( \overline{z_1 + z_a + \delta_1} \) denotes the vector \( z_1 + z_a + \delta_1 \) with its first coordinate removed, so it has length \( n_C + n_A - 1 \).

So, expanding the SZX-diagram corresponding to the physical interpretation of the Z-merge $p_\bul$, we get
\begin{equation}
    \input{SZX_p_0.tikz} = \input{Z-merge_ZX.tikz}
\end{equation}
where $\supp z_1 + z_a + \delta_1$ is the set of physical qubits in the support of $z_1 + z_a + \delta_1$, and $|z_1|$ and $|z_a|$ are the weights of $z_1$ and $z_a$, respectively. Hence, \cref{eq:Z_merge-split} turns into
\begin{equation}\label{eq:physical_implementation}
    \input{Z_merge-split_SZX.tikz}=\input{Z-merge-split_ZX.tikz}.
\end{equation}
As pointed out below \cref{eq:Z_merge-split}, because $V_0=0$, the final projection is unnecessary---it reduces to the identity as $\del^\Omega$ is an empty matrix.
The main issue with this physical implementation is that a single Z-Pauli error on the last qubit leads to the logical error $[z_1]+[z_a]$ as shown below:
\begin{equation}
    \input{Z-merge-split_ZX1.tikz}=\input{Z-merge-split_ZX2.tikz}.
\end{equation}
This can be avoided by decomposing $z_1+z_a$ as the sum of many vectors of smaller weight $z_1+z_a=\sum_jv_j$ and defining $V_1=\Span\{v_j\}_j$. The only constraint is that $im(i_{1*})$ should remain equal to $\Span\{[z_1]+[z_a]\}$ to achieve the desired logical operation. Because it depends on the specific code $((C\oplus A)_\bul,(C\oplus A)^\bul)$ (or $(C_\bul,C^\bul)$ if $(A_\bul,A^\bul)$ is embedded in the original code) under consideration, we present how to apply this strategy with a small example. Consider the two logical operator representatives $\overline{Z}_1$ and $\overline{Z}_a$ on respectively 3 and 2 qubits.

\begin{equation}
    \input{ex_V_1.tikz}
\end{equation}
To stress the importance that the chosen decomposition should respect $im(i_{1*})=\Span\{[z_1]+[z_a]\}$, we assume that $\overline{Z}_1=Z_1Z_3Z_4$ is not irreducible and that only $\overline{Z}_3=Z_1Z_3$ lies within its support.

In order to obtain a more fault-tolerant implementation, the rules of the game are simple: using the canonical bases vectors $\{\delta_1,\delta_2,\delta_3,\delta_4,\delta_5\}$, find a decomposition such that only $[z_1]+[z_a] \in \im(i_{1*})$. In this toy example where we assume that no stabilisers is contained in the support of the Z-logical operators, it boils down to choosing $V_1$ such that among the logical operator representatives, only $z_1+z_a \in V_1$. The decomposition $V_1=\Span\{\delta_1+\delta_2,\delta_3+\delta_4,\delta_5\}$ meets the requirements. 

Similarly to what has been done from \cref{eq:p_1_general} to \cref{eq:physical_implementation}, we deduce that the associated Z-merge $p_\bul$/X-split $p^\bul$ is implemented physically via
\begin{equation} \label{eq:ft_physical_implementation}
    \input{Z_merge-split_SZX.tikz}=\input{physical_implementation_ZX_variant.tikz}.
\end{equation}
Note that the final X-stabiliser measurement involves all the physical qubits of the code $((C\oplus A)_\bul,(C\oplus A)^\bul)$. Whereas in the previous physical implementation, an internal error on the last physical qubit (here qubit 5) would have spread into a logical error, it can now only propagate via the final stabiliser measurements. So, in the case of a qLDPC code, it will propagate to a limited amount of physical qubits.

This toy example shows how we can improve the implementation of the merges/splits to make it fault-tolerant to some extent. The key aspect here is to analyse the representatives of the logical operators to be merged in order to further decompose $V_1$. Although we did not find a direct connection between the dimension of $V_1$ and the fault-tolerance of the resulting implementation, this toy example shows that it is interesting to consider raising the dimension of $V_1$ on a case by case basis.

\subsection{A note on postselection and measurement errors}
\label{sec:postselection}

Throughout this paper, we have assumed that the measurements are post-selected on the ideal (\(+1\)) outcome. This is somewhat unavoidable because, in general, the measurement errors one obtains without postselection are difficult to take care of. Nevertheless, in the case of the physical implementation of the CNOT described in \cref{eq:physical_implementation}, it is fairly straightforward to take care of these measurement errors.
Recall that the physical implementation is
\begin{equation*}
    \input{Z_merge-split_SZX.tikz}=\input{Z-merge-split_ZX.tikz}.
\end{equation*}

This operation involves only one X-measurement represented by the single red spider at the bottom of the diagram. In the event of a negative measurement outcome, this red spider takes a $\pi$ phase. We can commute this phase with the encoders to see it as a supplementary X-logical operation applied before and after the Z-merge/X-split on the first logical qubit as illustrated on the following diagrams
\begin{equation}\label{eq:measurement_correction}
  \input{fault-tolerance.tikz}.
\end{equation}
Let us unwrap the transformations. As $x_1$ and $z_1$ are representatives of Pauli logical operators acting on the same logical qubit, according to prop.~\ref{prop:dual_bases} they have odd overlap. As a consequence, $\supp(x_1)\cap \supp(z_1) \neq \emptyset$. Moreover, $\supp(z_1) \subset \supp(z_1+z_a+\delta_1)$, implying that $\supp(x_1)\cap\supp(z_1+z_a+\delta_1) \neq \emptyset$. Choose a qubit in the overlap of these two supports to make the $\pi$ phase commute before and after the physical implementation of the Z-merge/X-split. From there, we can easily commute these $\pi$ red spiders with the encoders as shown \cref{eq:measurement_correction}.

A similar reasoning applies for the X-merge/Z-split, making possible to take into account negative measurement outcomes. As the resulting logical operation is a CNOT, we can safely correct all these logical Pauli errors, similarly to what has been described in \cref{eq:cnot3}. However, whether a similar strategy can handle measurement outcomes across different physical implementations---like those discussed in section~\ref{sec:physical_implementation}---remains an open question.

\section*{Conclusion}

Building on the work initiated in \cite{Cowtan_2024}, and through our reformulation of CSS surgery in terms of \emph{quotients by subcodes}, we have been able to compute the logical operations induced by very general surgery operations. This allows us to describe a methodology for implementing CNOTs with fault-tolerance guarantees between arbitrary CSS codes. 
We are also able to describe (some) code-switching through the resulting formalism, such as code-switching between the 3D colour code and the Steane code \cite{Anderson_2014, heußen2024efficientfaulttolerantcodeswitching}.

We briefly mention some research directions for future work. We assumed throughout the paper that all measurements are postselected on the ideal outcome. This restriction is obviously non-realistic and further work needs to be done to provide guarantees that negative measurement outcomes do not make the computation fail.
As described in \cref{sec:postselection}, this is possible in the case of the first CNOT physical implementation we have described \cref{eq:physical_implementation}. However, this implementation cannot be considered fault-tolerant as internal single qubit Pauli errors propagate to logical errors. How to generalise the method applied in \cref{sec:postselection} to fault-tolerant implementations like the one presented in \cref{eq:ft_physical_implementation} remains an open question.
It would be interesting to explore how the results of this work can be related to the complementary work of \cite{de_beaudrap_pauli_2020}, which describes when logical errors obtained from merges in surgery procedures can be corrected.

A second direction concerns the computation in full generality of the logical operation, i.e., of the (co)homology of the X- and Z-merges (\cref{eq:long_exact_sequence_homology}).  We rely on the additional technical assumption that the homology of the subcode $(V_\bul,V^\bul)$ satisfies either $H_0(V_\bul)=0$ or $H_1(V_\bul)=0$. We have not found a simple way to express the logical operation without this technical condition. This implies that we must carefully construct the subcode $(V_\bul,V^\bul)$ such that it satisfies this condition.

Third, a perennial problem in QEC is that measurement outcomes are themselves not reliable and can be subject to errors. A way to approach this problem is by introducing meta-checks, that consist in plugging a classical code on the syndrome measurements. This is captured straightforwardly in the homological picture of CSS codes by extending the length of the chain complex. 
It should be feasible to incorporate metachecks in the formalism to potentially enable CSS surgery for single-shot codes \cite{Quintavalle_2021,hillmann2024singleshotmeasurementbasedquantumerror}.

While our work focusses on qubits, the categorical nature of the theory allows it to be extended to other homological codes, including those involving qudits or more exotic quantum systems such as rotors and bosonic codes  \cite{cowtan_qudit_2022, vuillot2023homologicalquantumrotorcodes, xu2024lettingtigercagebosonic}. In this direction, the main difficulty is to adapt prop.~\ref{hilb_representation_chain} and theorem~\ref{theorem:soundness}, namely, to find a consistent physical implementation of the (co)chain maps. In a similar vein, this formalism may also provide insights into the study of CSS Floquet or space-time codes \cite{Hastings_2021,delfosse2023spacetimecodescliffordcircuits}.

A downside of the theory developed in this work is that it is somewhat restricted in terms of the operations it can describe: it cannot even describe general Clifford unitaries. The category $\Chains$ is commonly used in quantum error correction because there is currently no equivalent way to represent stabiliser codes as elegantly as CSS codes using chain complexes. Further investigation is needed to construct a category that captures stabiliser codes together with a functor leading to the logical space. Promising work in this direction, using symplectic algebra, has already begun \cite{booth2024graphicalsymplecticalgebra}.

\section*{Acknowledgements}

C.P. thanks Adithya Sireesh, Armanda O. Quintavalle and Simon Burton for insightful discussions in the early stages of this project. C.P. also thanks Anthony Leverrier and Christophe Vuillot for valuable advice, for suggesting the use of SZX-calculus 
and finding how to adapt the CNOT protocol with negative measurement outcomes
, Arthur Pesah for his assistance in designing small examples of surgery protocols and Alexander Cowtan for helpful discussions on the interpretation of the long exact sequence \cref{eq:long_exact_sequence_homology} and for finding the correct physical interpretation of preserving code maps \cref{eq:hilb_representation_chain}. C.P. also acknowledges Esha Swaroop for a helpful discussion on prior works in code surgery. Finally, we thank Șerban Cercelescu for pointing out that the choice of basis in a quotient merged code affects the code distance. We also thank Zhiyang He and Dominic J. Williamson for their valuable insights regarding prior work on code surgery.

C.P. acknowledges funding from the Plan France 2030 through the project ANR-22-PETQ-0006.

 J.R. is funded by by EPSRC Grants
EP/T001062/1 and EP/X026167/1.

A final thanks goes to the entire Quantum Software Lab team in Edinburgh for their support and kindness, as well as to the organizers of the Fault-Tolerant Quantum Computing 2024 workshop and the Centro de Ciencias de Benasque, where many discussions took place.

\printbibliography

\appendix

\section{Logical operations of lattice surgery}
\label{app:lattice_surgery}

We aim at proving that the physical operation presented \cref{fig:lattice_surgery} implements the logical operations presented \cref{eq:physical_surgery}. As the reasoning is very similar in the four cases (X- or Z-merge and X- or Z-split), we will focus our attention on the Z-merge.

Consider the Z-merge physical operation represented in thumbnail 2 \cref{fig:lattice_surgery}, where we measure new Z-stabilisers. We want to show that this physical operation commutes with the encoder in the following manner:
\begin{equation}
    \input{physical_Z-merge.tikz}.
\end{equation}
Consider the quantum channel $\Phi_p'$ representing the merge operation. It is defined by its Kraus operators $\{K_p'\}_{p \in \F_2^{n_a}}$ where $n_a$ is the number of auxiliary qubits in between the two surface patches \cref{fig:lattice_surgery}. Each Kraus operator $K_p'$ is a projection onto the Z-Pauli eigenspace of the corresponding combination of measurement outcomes $p$. Depending on the measurement outcomes $p$, one can directly interpret the resulting state as being a (larger) surface code---up to a logical Pauli correction. More details on this correction can be found in \cite{de_Beaudrap_2020}. Here, we will consider the correction applied when needed such that $\Phi_p$ defined by $\{K_p\}_p$ where $K_p=P^{correct}_pK_p'$ does already account for this correction. Thus, the state resulting from $\Phi_p$ is in any cases a (larger) surface code. We make the \emph{choice} to interpret it that way by defining the encoder $E$. 
Hence, the induced logical operation is represented by the quantum channel $\Phi_l$ defined by its Kraus operators $\{K_l\}_l$ where $K_l=E^\dagger K_p(E'\tens E')$. All the information can be recast in terms of a commuting diagram:
\begin{equation}
    \begin{tikzcd}
        \input{two_surface.tikz} \arrow[r, "\Phi_p"] & \input{big_surface.tikz} \\
        {(\C^2)}^{\tens 2} \arrow[u,"E' \tens E'"] \arrow[r, "\Phi_l"'] & \C^2 \arrow[u, "E"']
    \end{tikzcd}
\end{equation}
where the upper row objects are representations of the surface code Tanner graphs~\footnote{Formally, they should be ${(\C^2)}^{\tens n}$ and ${(\C^2)}^{\tens n'}$, but we prefer this graphical representation that stress on the Tanner graph representation of the merge.}. The question is now how to effectively deduce the $K_l$'s? The most efficient approach is to relate the logical operators before and after the merge. For simplicity, we are only presenting the method when no logical Pauli correction is required. A similar reasoning can be done with the corrections \cite{de_Beaudrap_2020}.

Consider $K_l=E^\dagger K_p(E'\tens E')$ such that no correction is required (for instance, we got only positive measurement outcomes $p=0_{\F_2^{n_a}}$). Then, by looking at the Tanner graph \cref{fig:lattice_surgery}, we can assert that the X-logical operators of the two original surface patches are mapped to the only X-logical operator in the resulting larger surface code---i.e. $K_p\overline{X}_1\overline{X}_2K_p^\dagger=\overline{X}$ as $\overline{X}_1\overline{X}_2$ and $K_p$ commute. Similarly, for Z-logical operators, we have $K_p\overline{Z}_1K_p^\dagger=\overline{Z}=K_p\overline{Z}_2K_p^\dagger$. Using the fact that $\overline{X}_i=E'X_iE'^\dagger$, $\overline{Z}_i=E'Z_iE'^\dagger$ and $\overline{X}=EXE^\dagger$, $\overline{Z}=EZE^\dagger$, we get the relations:
\begin{align}
    XK_lX_1X_2&=K_l \\
    ZK_lZ_1&=K_l \\
    ZK_lZ_2&=K_l.
\end{align}
Using the Choi-Jamiołkowski isomorphism $\mathcal{J}:\mathcal{L}(\mathcal{H}) \rightarrow \mathcal{H}^*\tens \mathcal{H}$, this system of equations can be turned into
\begin{align}
    (X \tens X_1X_2) \mathcal{J}(K_l)&=\mathcal{J}(K_l) \\
    (Z\tens Z_1)\mathcal{J}(K_l)&=\mathcal{J}(K_l) \\
    (Z\tens Z_2)\mathcal{J}(K_l)&=\mathcal{J}(K_l)
\end{align}
Hence, $\mathcal{J}(K_l)$ is a codeword of the CSS code defined by the stabilisers $\langle (X \tens X_1X_2), (Z\tens Z_1), (Z\tens Z_2) \rangle$. This code is maximal meaning there exists a single codeword. Using \cite{kissinger2022phasefree} (or simply by solving the system), we get 
\begin{equation}
    \mathcal{J}(K_l)=\input{Z-spider_JK_l.tikz}.
\end{equation}
Thus, 
\begin{equation}
    K_l=\input{Z-spider_K_l.tikz}.
\end{equation}

\section{The category $\Chains$ and homology} \label{hom&cat}

In this section, we introduce all the relevant mathematical notions used in this work. More details on these concepts can be found in \cite{Hatcher:478079, leinster2016basiccategorytheory}.

\subsection{Homological algebra}

The purpose of this subsection is to provide an overview of the category $\Chains$ and its dual $\Cochains$, and second to present the homology and cohomology functors. To achieve our ends, we need first to define the category $\MatF$:
\begin{itemize}
    \item each object is a vector space $V$ equipped with a specified basis $\Tilde{V}$ such that $V = \F_2 \Tilde{V}$, 
    \item maps are matrices over $\F_2$ and the composition is the usual matrix multiplication.
\end{itemize}
We will also utilise the known category $\Vect$, where objects are vector spaces and maps are linear maps with the usual composition.
We can now introduce chain complexes over $\MatF$.

\begin{definition}
    An object of $\Chains$, called a chain complex $C_\bul$, is a sequence of objects---called ``complexes''---and maps of $\MatF$ indexed by $\Z$
    \begin{equation}
    \begin{tikzcd}
        \ldots \arrow[r] & C_{n+1} \arrow[r, "\del_{n+1}"] & C_n \arrow[r, "\del_n"] & C_{n-1} \arrow[r] & \ldots
    \end{tikzcd}
    \end{equation}
    such that $\del_n \circ \del_{n+1} = 0$ for all $n \in \Z$.

    \noindent Furthermore, we define the cochain complex $C^\bul$ of $C_\bul$ to be
    \begin{equation}
    \begin{tikzcd}
        \ldots & C^{n+1} \arrow[l] & C^n \arrow[l, "\del^{n+1}"'] & C^{n-1} \arrow[l, "\del^n"'] & \ldots \arrow[l]
    \end{tikzcd}
    \end{equation}
    where $C^n \cong C_n$ is the dual of $C_n$ and $\del^n = \del_n^T$ for any $n \in \Z$. These are objects of $\Cochains$.
\end{definition}

The point of defining chain complexes is to be able to study them through the lens of homology and cohomology.

\begin{definition}
    Let $C_\bul$ be a chain complex over $\MatF$, we define for all $n \in \Z$ 
    \begin{equation}
    \begin{tikzcd}
        H_n(C_\bul) := \ker(\del_n) / \im(\del_{n+1}) & H^n(C_\bul) := \ker(\del_{n+1}^T) / \im(\del_n^T)
    \end{tikzcd}
    \end{equation}
    to be the $n$-th homology and cohomology spaces respectively.
\end{definition}

There is a duality between the homology and the cohomology spaces that is described by the two following propositions.

\begin{proposition} \label{H_n isom H^n}
    Let $C_\bul$ be a chain complex over $\MatF$. For all $n \in \Z$,
    \begin{equation}
    H_n(C_\bul) \cong H^n(C_\bul)
    \end{equation}
\end{proposition}

\begin{proof}
    First, using the two equalities $\ker \del_{n+1}^T = {\im \del_{n+1}}^\perp$ and $\im \del_n^T = {\ker \del_n}^\perp$, we get
    \begin{equation}
        \ker \del_{n+1}^T / \im \del_n^T = {\im \del_{n+1}}^\perp / {\ker \del_n}^\perp
    \end{equation}
    Then, we consider $p : {\im \del_{n+1}}^\perp \to \ker \del_n$ the orthogonal projection onto $\ker \del_n$. As $\ker p = {\ker \del_n}^\perp$, we get the injective map $\Tilde{p} : {\im \del_{n+1}}^\perp / {\ker \del_n}^\perp \to \ker \del_n$. The last step is to quotient out the codomain by $\im \del_{n+1}$. In fact, using the chain complex property $\im \del_{n+1} \subseteq \ker \del_n$, we get the decomposition into direct sum $\ker \del_n = \im \del_{n+1} \oplus {\im \del_{n+1}}^\perp \cap \ker \del_n$. Hence, the isomorphism $\Tilde{p} : {\im \del_{n+1}}^\perp / {\ker \del_n}^\perp \to \ker \del_n / \im \del_{n+1}$.
\end{proof}

\begin{proposition} \label{canonical_bases}
    The duality pairing
    \begin{align*}
        \cdot : C^n \times C_n & \to \F_2 \\
        x^T,z & \mapsto x^Tz
    \end{align*}
    lifts to $H^n(C_\bul) \times H_n(C_\bul)$ in the following way 
    \begin{align*}
        \cdot : H^n(C_\bul) \times H_n(C_\bul) & \to \F_2 \\
        [x],[z] & \mapsto x \cdot z
    \end{align*}

    Thus, the duality pairing of $C^n \times C_n $ lifts to $H^n(C_\bul) \times H_n(C_\bul)$, meaning that considering a basis $\{[z_j]\}_{1 \leq j \leq k}$ of $H_n(C_\bul)$, there exist a unique basis $\{[x_i]\}_{1 \leq i \leq k}$ of $H^n(C_\bul)$ such that $[x_i]\cdot [z_j] = \delta_{i,j}$.
\end{proposition}

\begin{proof}
    We shall prove that the lift of the scalar product is well defined, i.e., does not depend on the representative. This is comes as a consequence of $\im(\del_{n+1}) \perp \ker(\del^{n+1})$, $\im(\del^n) \perp \ker(\del_n)$ and $\im(\del_{n+1}) \perp \im(\del^n)$. If $x,z \in \ker \del_n \times \ker \del^{n+1}$ and $s,t \in \im \del_{n+1} \times \im \del^n$, then $(x+s) \cdot (z+t) = x \cdot z$.

    Consider now a basis $\{[z_j]\}_{1 \leq j \leq k}$ of $H_n(C_\bul)$. We want to construct the dual basis $\{[x_i]\}_{1 \leq i \leq k}$ of $H^n(C_\bul)$ such that $[x_i]\cdot [z_j] = \delta_{i,j}$.
    The construction is based on the decomposition $\F_2^n \cong H_n(C_\bul) \oplus \im \del_{n+1} \oplus {\ker \del_n}^\perp$.
    Define the matrix  
    \begin{equation}
    L_Z=\left( \begin{array}{c|c|c}
        z_1 & \ldots & z_k
    \end{array} \right)
    \end{equation}
    where the $j$-th column is a representative of $[z_j]$. Consider a generating matrix $\Tilde{\del}_{n+1}$ of $\im \del_{n+1}$ (that can be obtained by removing columns of $\del_{n+1}$ until the matrix becomes injective), and a generating matrix $\Tilde{\del}^n$ of ${\ker \del_n}^\perp$ (that again can be obtained by removing columns of $\del^n$ until the matrix becomes injective).
    Now define the \emph{invertible} matrix
    \begin{equation}
        \left(\begin{array}{c|c|c}
            L_Z & \Tilde{\del}_{n+1} & \Tilde{\del}^n  
        \end{array}\right)
    \end{equation}
    and compute the left inverse
    \begin{equation}
        \left(\begin{array}{c}
            L_X \\
            \hline* \\ \hline * \\
        \end{array}\right)\left(\begin{array}{c|c|c}
            L_Z & \Tilde{\del}_{n+1} & \Tilde{\del}^n  
        \end{array}\right)=I.
    \end{equation}
    We can verify that the $k$ first rows $\{x_i\}_{1 \leq i \leq k}$ of the left inverse---i.e. the rows of $L_X$---are representatives of a basis of $H^n(C_\bul)$. For each row $x_i$ of $L_X$, $x_1\Tilde{\del}_{n+1}=0$, implying that $x_1 \in \im {\del_{n+1}}^\perp=\ker \del^{n+1}$. Finally, $\{[x_i]\}_{1 \leq i \leq k} \in H^n(C_\bul)$ are linearly independent as $L_X\Tilde{\del}_{n+1}=0$.
\end{proof}

So far, we have focused our attention on the objects of the category $\mathtt{(co)}\Chains$---namely, the chain complexes. We now turn our attention to maps of $\mathtt{(co)}\Chains$. These are called chain maps.

\begin{definition}
    A chain map $f_\bul : C_\bul \rightarrow D_\bul$ consists in a collection of matrices $\{f_i : C_i \rightarrow D_i \}_{i\in \Z}$ over $\F_2$ such that each square of the following diagram commutes\footnote{The $n$-th square commutes iff $f_{n-1} \circ \del_n^C  = \del_n^D \circ f_n$. All the squares commute iff this relation holds for all $n \in \Z$.}.
    \begin{equation}\begin{tikzcd}\cdots \arrow[r] & C_{n+1}\arrow[r, "\del^C_{n+1}"]\arrow[d, "f_{n+1}"] & C_{n}\arrow[r, "\del^C_{n}"]\arrow[d, "f_{n}"] & C_{n-1}\arrow[r]\arrow[d,"f_{n-1}"] & \cdots\\
    \cdots \arrow[r] & D_{n+1}\arrow[r, "\del^D_{n+1}"] & D_{n}\arrow[r, "\del^D_{n}"] & D_{n-1}\arrow[r] & \cdots\end{tikzcd}.
    \end{equation}
    
    Similarly, a cochain map $g^\bul : C^\bul \rightarrow D^\bul$ consists in a collection of matrices $\{g^i : C^i \rightarrow D^i \}_{i\in \Z}$ over $\F_2$ such that all the squares commute
    \begin{equation}
    \begin{tikzcd}
        \cdots & \arrow[l] \arrow[d,"g^{n+1}"] C^{n+1} & \arrow[l, "\del_C^{n+1}"']\arrow[d, "g^n"] C^n & \arrow[l, "\del_C^{n}"']\arrow[d, "g^{n-1}"] C^{n-1} & \arrow[l] \cdots\\
        \cdots & \arrow[l] D^{n+1} & \arrow[l, "\del_D^{n+1}"'] D^n & \arrow[l, "\del_D^{n-1}"'] D^{n} & \arrow[l] \cdots
    \end{tikzcd}.
    \end{equation}
    Finally, we define the composition of chain maps by pairwise matrix multiplication, and similarly for the composition of cochain maps.
\end{definition}

Note that from any chain map $f_\bul : C_\bul \rightarrow D_\bul$, we can define a cochain map $f^\bul : D^\bul \rightarrow C^\bul$ where $f^n=f_n^T$ for all $n \in \Z$. 

The axiomatisation of $\{H_n : \Chains \rightarrow \Vect\}_{n \in \Z}$ and $\{H^n : \Cochains \rightarrow \Vect\}_{n \in \Z}$ are functorial. In fact, every chain map lifts to a linear map between the homology spaces.

\begin{theorem} \label{th:lift-chain_maps}
    Let $f_\bul : C_\bul \rightarrow D_\bul$ be a chain map. Then, $f_{n*} := H_n(f_\bul) : H_n(C_\bul) \rightarrow H_n(D_\bul)$ is defined for every $[x] \in H_n(C_\bul)$ to be $f_{n*}([x]) = [f_n(x)] \in H_n(D_\bul)$. \\
    Considering a second chain map $g_\bul : D_\bul \rightarrow E_\bul$ The composition is defined as $(g_n \circ f_n)_* = g_{n*} \circ f_{n*}$.
\end{theorem}

\begin{proof}
    We have to prove that $f_{n*}$ is well defined, i.e., that $f_{n*}([x])$ does not dependend on the representative $x \in \ker(\del_n^C)$, and that $f_n(x) \in \ker(\del_n^D)$. A small diagram chasing using the commutativity of the squares guarantees that $f_n(\ker(\del_n^C)) \subseteq \ker(\del_n^D)$ and $f_n(\im(\del_{n+1}^C)) \subseteq \im(\del_{n+1}^D)$.
\end{proof}

The same goes for the cohomology functor with cochain maps. Finally, we introduce the notion of direct sums of chain complexes.

\begin{definition}
    Let $C_\bul$ and $D_\bul$ be two chain complexes over $\MatF$. The direct sum $(C \oplus D)_\bul$ of $C_\bul$ and $D_\bul$ is the chain complex
    \begin{equation}
    \begin{tikzcd}
        \ldots \arrow[r] & C_{n+1} \oplus D_{n+1} \arrow[r, "\del_{n+1}"] & C_n \oplus D_n \arrow[r, "\del_n"] & C_{n-1} \oplus D_{n-1} \arrow[r] & \ldots
    \end{tikzcd}
    \end{equation}
    where for all $n \in \Z$, $\del_n := \del_n^C \oplus \del_n^D$, meaning that
    for $c+d \in C_n \oplus D_n$, $\del_n (c+d) := \del_n^Cc + \del_n^Dd$.
\end{definition}

\subsection{Short exact sequences and homology}

\begin{definition} \label{exact_sequence}
    Consider a sequence of maps
    \begin{equation}
        \begin{tikzcd}
            \ldots \arrow[r] & V_{n+1} \arrow[r, "\phi_{n+1}"] & V_n \arrow[r, "\phi_n"] & V_{n-1} \arrow[r] & \ldots
        \end{tikzcd}
    \end{equation}
    The sequence is exact if $\im(\phi_{n+1})=\ker(\phi_n)$ for all $n \in \Z$.
\end{definition}

\begin{theorem} \label{Hatcher}
    Let
    \begin{equation}\begin{tikzcd}
    0_\bul \arrow[r] & A_\bul \arrow[r, "\iota_\bul"] & B_\bul \arrow[r, "\pi_\bul"] & C_\bul \arrow[r] & 0_\bul
    \end{tikzcd}\end{equation}
    be a short exact sequence of chain complexes~\footnote{This means that replacing $\bul$ with any $n \in \Z$, the sequence is exact.}. This exact sequence lifts to a long exact sequence on the homology spaces
    \begin{equation}\begin{tikzcd}
        \ldots \arrow[r] & H_n(A_\bul) \arrow[r, "\iota_{n*}"] & H_n(B_\bul) \arrow[r, "\pi_{n*}"] & H_n(C_\bul) \arrow[r, "\del_n"] & H_{n-1}(A_\bul) \arrow[r] & \ldots
    \end{tikzcd}\end{equation}
    for some sequence of homomorphism $(\del_n)_n$. It does also lift to a long exact sequence on cohomology spaces

    \begin{equation}\begin{tikzcd}
        \ldots & \arrow[l] H^n(A_\bul) & \arrow[l, "\iota_*^n"'] H^n(B_\bul) & \arrow[l, "\pi_*^n"'] H^n(C_\bul) & \arrow[l, "\del^n"'] H^{n-1}(A_\bul) & \arrow[l] \ldots
    \end{tikzcd}\end{equation}
\end{theorem}

\begin{proof}
    The long exact sequence on homology can be found in \cite{Hatcher:478079} at p.116. The second on cohomology spaces is obtained using prop.~\ref{H_n isom H^n}.
\end{proof}

We can derive from this theorem a similar version with short exact sequences in $\Cochains$. To stay concise, we do not present it here.

\subsection{Pushouts, pullbacks and exact sequences}

Here, we introduce the key ingredient for proving the equivalence between our construction and that of \cite{Cowtan_2024}. The proof relies on establishing a connection between pushouts, pullbacks, and exact sequences. Notably, throughout the following discussion, we consider pushouts and pullbacks without explicitly verifying their existence. This is justified by the fact that the category 
$\Chains$ (as well as $\Cochains$) is abelian \cite{lane1978categoriesfortheworkingmathematician}.

\begin{theorem} \label{theorem:stacks}
Consider the following commuting square in the category $\Chains$;
\begin{equation}\begin{tikzcd}
V_\bul \arrow[r, "g_\bul"]\arrow[d, "f_\bul"'] & D_\bul \arrow[d,"l_\bul"]\\
C_\bul \arrow[r,"k_\bul"'] & Q_\bul 
\end{tikzcd}\end{equation}
This diagram is a pullback square iff the corresponding sequence
\begin{equation}\begin{tikzcd}
    0_\bul \arrow[r] & V_\bul \arrow[r, "f_\bul \oplus g_\bul"] & (C\oplus D)_\bul \arrow[r, "k_\bul - l_\bul"] & Q_\bul
\end{tikzcd}\end{equation}
is exact. Similarly, it is a pushout square iff the sequence
\begin{equation}\begin{tikzcd}
    V_\bul \arrow[r, "f_\bul \oplus g_\bul"] & (C\oplus D)_\bul \arrow[r, "k_\bul - l_\bul"] & Q_\bul \arrow[r] & 0_\bul
\end{tikzcd}\end{equation}
is exact. Hence, the diagram is both a pullback and a pushout square iff the sequence
\begin{equation}\begin{tikzcd}
    0_\bul \arrow[r] & V_\bul \arrow[r, "f_\bul \oplus g_\bul"] & (C\oplus D)_\bul \arrow[r, "k_\bul - l_\bul"] & Q_\bul \arrow[r] & 0_\bul
\end{tikzcd}\end{equation}
is exact.
\end{theorem}

\begin{proof}
    We refer to Lemma 12.5.2 \cite{stacks-project}.
\end{proof}

\section{Soundness of the interpretation} \label{ap:interpretation}

In this section, we present all the proofs regarding the interpretation of the (co)chain maps and of their lifts to the (co)homology. The proofs are presented for chain complexes/maps, but can be adapted dually for cochain complexes/maps.

\begin{proof}[Proof of Proposition~\ref{hilb_representation_chain}] 
    We aim to prove that this representation is consistent with the interpretation of the chain complexes we gave in \cref{chain_complex}. In other words, starting from a state in the codespace of $(C_\bul,C^\bul)$---namely, stabilised by all the stabilisers of $(C_\bul,C^\bul)$---and applying the interpretation of $f_\bul$ gives a state in the codespace of $(D_\bul,D^\bul)$.
    It essentially comes from the properties of SZX-calculus \cite{https://doi.org/10.4230/lipics.mfcs.2019.55}.

    \textit{Z-stabilisers: }In SZX-calculus, a stabiliser Z-operator of $(D_\bul,D^\bul)$---i.e. a vector $\del_2^Dy \in \im(\del_2^D)$---is represented by a thick Z-spider labelled $\pi \del_2^Dy$. According to the rules of SZX-calculus, we can make this Z-operator commute with the physical operation. Considering $D_2=\im(f_2)\oplus \Omega$, there are two cases; either $y \in \im(f_2)$ such that $y=f_2x$ and using the commutativity of the left square in \cref{eq:chain_map}, we get
    \begin{equation}
        \input{SZX_stab_to_stab_f_03.tikz} = \input{SZX_stab_to_stab_f_02.tikz} = \input{SZX_stab_to_stab_f_01.tikz}.
    \end{equation}
    Or $y \in \Omega$, then $\del_2^Dy=\del_\Omega y$ and we get
    \begin{equation}
        \input{SZX_stab_to_stab_f_032.tikz} = \input{SZX_stab_to_stab_f_012.tikz}=\input{SZX_f_0.tikz}.
    \end{equation}
    
    \textit{X-stabilisers: }Similarly, in SZX-calculus, a stabiliser X-operator $\del^1_Dy \in \im(\del^1_D)$ is a thick X-spider labelled $\pi \del^1_Dy$. It turns into a thick X-spider labelled $\pi f^1 \del^1_D y \in \im(\del^1_C)$  
    \begin{equation}
        \input{SZX_syn_to_syn_f_03.tikz} = \input{SZX_syn_to_syn_f_02.tikz} = \input{SZX_syn_to_syn_f_01.tikz}
    \end{equation} 
    thus guaranteeing again that any state stabilised by the X-stabilisers of $(C_\bul,C^\bul)$ will result in a state stabilised by the X-stabilisers of $(D_\bul,D^\bul)$.
\end{proof}

\begin{remark}
    The interpretation can be formalised as a functor $\Xi$ from the sub-category of $\Chains$ composed of length 2 chain complexes/maps to a sub-category of $\FHilb$ where objects are CSS codespaces and maps are phase-free ZX-diagrams. It is not hard to prove that $\Xi(f_\bul \circ g_\bul)=\Xi(f_\bul) \circ \Xi(g_\bul)$ and that $\Xi(id_\bul)=\input{SZX_wire.tikz}.$
\end{remark}

We now turn our attention to the proof of theorem~\ref{theorem:soundness} which states that the interpretation of the lifts of chain maps coincides with the logical operation induced by a Z-preserving code map.

\begin{proof}[Proof of Theorem~\ref{theorem:soundness}:] 
    In order to prove the equality, we need to explicit the link between the encoders $E_C$ and $E_D$ and the bases chosen to express $F_{1*}$.
    This can be done using prop.~\ref{prop:encoder-basis}; consider the unique choice of basis map $\varepsilon_C:\F_2^{k_C} \oplus \F_2^{k_C} \rightarrow H^1(C_\bul) \oplus H_1(C_\bul)$ (resp. $\varepsilon_D:\F_2^{k_D} \oplus \F_2^{k_D} \rightarrow H^1(D_\bul) \oplus H_1(D_\bul)$) corresponding to $E_C$ (resp. $E_D$). Because the encoders are type preserving (cf. remark~\ref{rem:type_preserving}) we can define $\varepsilon^X_C$ and $\varepsilon^Z_C$ to be
    \begin{align}
        \varepsilon_C^X : \ \F_2^{k_C} &\rightarrow H^1(C_\bul) & \varepsilon_C^Z : \ \F_2^{k_C} &\rightarrow H_1(C_\bul) \\
        u &\mapsto \varepsilon_C(u|0_{k_C})\restriction^{H^1(C_\bul)} & v &\mapsto \varepsilon_C(0_{k_C}|v)\restriction^{H_1(C_\bul)}
    \end{align}
    where there is a restriction on the codomain. We have the equivalent for $\varepsilon_D$. These restrictions can equivalently be regarded as 
    \begin{equation}
        \varepsilon=\left( \begin{array}{c|c}
            \varepsilon^X & 0 \\
            \hline 0 & \varepsilon^Z
        \end{array} \right)
    \end{equation}
    similarly to what is done for the split of the parity check matrix of a CSS code into $P_X$ and $P_Z$ \cref{eq:CSS_parity_check}. Now we can rewrite $F_{1*}$ as $F_{1*}={\varepsilon_D^Z}^{-1} \circ f_{1*} \circ \varepsilon_C^Z$.

    The strategy of the proof is to apply the Choi-Jamiołkowski isomorphism to both diagrams to get two quantum states $\ket{f}$ and $\ket{F}$ over $k_C+k_D$ qubits, and then show that both are stabilised by $k_C+k_D$ independent stabilisers, implying that these two states need to be the same. To find the stabilisers, we take advantage of the fact that Z-logical operators of $(C_\bul,C_\bul)$ (resp. X-logical operators of $(D_\bul,D_\bul)$) are mapped to Z-logical operators of $(D_\bul,D_\bul)$ (resp. to X-logial operators of $(C_\bul,C_\bul)$) via the physical interpretation of the chain map $f_\bul$.
    
    \sloppy \textit{Constructing Z-stabilisers:} 
    It is easy to see that the state $\ket{F}$ is stabilised by the group $\left\{Z^{\delta_i}\tens Z^{F_{1*}\delta_i} : i \in \llbracket1,k_C\rrbracket\right\}$ using the rules of SZX-calculus:
    \begin{equation}
        \input{proof_bent_FZ1.tikz}=\input{proof_bent_FZ2.tikz}
    \end{equation}
    
    Now consider $z$ to be a representative of $\varepsilon_C^Z(\delta_i) \in H_1(C_\bul)$. Then, by definition of $\varepsilon_C$ in prop.~\ref{prop:encoder-basis}
    \begin{align}
        \input{interpretation_proof_bent1.tikz}&=\input{interpretation_proof_bent2.tikz}=\input{interpretation_proof_bent3.tikz}.
    \end{align}
    Because $z \in \ker(\del_1^C)$ and that $f_\bul$ is a chain map, $f_1z \in \ker(\del_1^D)$ meaning it is a Z-logical operator of the code $(D_\bul,D_\bul)$. So, we can make it commute with the encoder as follows
    \begin{align}
        \input{interpretation_proof_bent3.tikz}&=\input{interpretation_proof_bent4.tikz}.
    \end{align}
    By definition of $f_{1*}$, we get ${\varepsilon_D^Z}^{-1}[f_1z]={\varepsilon_D^Z}^{-1}f_{1*}[z]={\varepsilon_D^Z}^{-1}f_{1*}\varepsilon_C^Z(\delta_i)=F_{1*}\delta_i$. Hence, both states are stabilised by the group $\left\{Z^{\delta_i}\tens Z^{F_{1*}\delta_i} : i \in \llbracket1,k_C\rrbracket\right\}$ containing $k_C$ independent Pauli operators.

    \textit{Constructing X-stabilisers:} 
    Similarly, on one hand, the state $\ket{F}$ is stabilised by the group $\left\{X^{F_*^1\delta_j}\tens Z^{\delta_j} : j \in \llbracket1,k_D\rrbracket\right\}$:
    \begin{equation}
        \input{proof_bent_FX1.tikz}=\input{proof_bent_FX2.tikz}.
    \end{equation}
    
    On the other hand, let $x$ be a representative of $\varepsilon_D^X(\delta_j) \in H^1(D_\bul)$, then
    \begin{equation}
        \input{interpretation_proof_bent4X.tikz}=\input{interpretation_proof_bent3X.tikz}=\input{interpretation_proof_bent2X.tikz}
    \end{equation}
    where $\del_\Omega^Tx=0$ because $x \in \ker \del^2_D$. Moreover, as $f^\bul$ is a cochain map, we get $f^1x \in \ker\del^2_C$. So, the X-Pauli commutes with the encoder to give:
    \begin{equation}
        \input{interpretation_proof_bent2X.tikz}=\input{interpretation_proof_bent1X.tikz}.
    \end{equation}
    Similarly, ${\varepsilon_C^X}^{-1}[f^1x]={\varepsilon_C^X}^{-1}f^1_*\varepsilon_D^X(\delta_j)$. Concluding that it corresponds to $F^1_*\delta_j$ is less straightforward than in the previous case, and requires to exploit the fact that $\varepsilon_C$ as well as $\varepsilon_D$ preserve the symplectic structure, i.e., the commutation relations. In the case of type-preserving encoder, it boils down to the equation ${\varepsilon^Z}^T\varepsilon^X=I_k$. Indeed, for all $u,v \in \F_2^k$,
    \begin{equation}
        \langle\varepsilon^Zu, \varepsilon^Xv\rangle=\langle u,v\rangle=\langle u,{\varepsilon^Z}^T\varepsilon^Xv\rangle
    \end{equation}
    where the scalar product is computed using prop.~\ref{canonical_bases}. Going back to the X-Pauli, we get ${\varepsilon_C^X}^{-1}f^1_*\varepsilon_D^X(\delta_j)=\left({\varepsilon^Z_C}^T\varepsilon^X_C\right){\varepsilon_C^X}^{-1}f^1_*\varepsilon_D^X\left({\varepsilon^Z_D}^T\varepsilon^X_D\right)^{-1}(\delta_j)={\varepsilon_C^Z}^Tf_*^T\left({\varepsilon_D^Z}^{-1}\right)^T(\delta_j)=F^1_*\delta_j$. Hence, both states are stabilised by the group $\left\{X^{F_*^1\delta_j}\tens Z^{\delta_j} : j \in \llbracket1,k_D\rrbracket\right\}$ containing $k_D$ independent Pauli operators.
    
\end{proof}

\section{A counter example} \label{ap:counterexample}

Consider the two codes $(C_\bul,C^\bul)$ and $(D_\bul,D^\bul)$ pictured by their Tanner graphs
\begin{equation}
  (C_\bul,C^\bul) : \vcenter{\hbox{\includegraphics[width=0.1\linewidth]{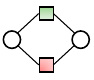}}} \quad \quad (D_\bul,D^\bul) : \vcenter{\hbox{\includegraphics[width=0.1\linewidth]{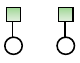}}}.
\end{equation}
Consider the chain map $f_\bul : C_\bul \rightarrow D_\bul$ defined on the following diagram

\begin{equation}\label{eq:counterexample}
    \begin{tikzcd}
        C_\bul \arrow[d,"f_\bul"] \\
        D_\bul
    \end{tikzcd}
    \ = \
    \begin{tikzcd}[ampersand replacement=\&]
        \F_2 
        \arrow[r, "{\scriptsize\left(\begin{smallmatrix}1 \\ 1\end{smallmatrix}\right)}"] 
        \arrow[d, "{\scriptsize\left(\begin{smallmatrix}1 \\ 1\end{smallmatrix}\right)}"'] 
        \& \F_2^2 
        \arrow[r, "{\scriptsize\left(\begin{smallmatrix}1 & 1\end{smallmatrix}\right)}"] 
        \arrow[d, "I_2"] 
        \& \F_2 
        \arrow[d, "0"] \\
        \F_2^2 
        \arrow[r, "I_2"] 
        \& \F_2^2 
        \arrow[r, "{\scriptsize\left(\begin{smallmatrix}0 & 0\end{smallmatrix}\right)}"] 
        \& 0
    \end{tikzcd}.
\end{equation}
It is clear that $f_2$ is not surjective. So if we omit the projection in the interpretation \cref{hilb_representation_chain},
because $f_1=I_2$, the physical implementation is:
\begin{equation}
    \input{SZX_id.tikz}=\input{SZX_wire.tikz}.
\end{equation}
Therefore, according to \cref{eq:counterexample}, applying the identity to a state stabilised by $(C_\bul,C^\bul)$ gives a state in $(D_\bul,D^\bul)$. In other words, each state in the codespace of $(C_\bul,C^\bul)$ should be in the codespace of $(D_\bul,D^\bul)$. This is obviously not the case as $\frac{1}{\sqrt{2}}(\ket{00}+\ket{11})$ belongs to the codespace of $(C_\bul,C^\bul)$ but not to that of $(D_\bul,D^\bul)$ as $Z_1\frac{1}{\sqrt{2}}(\ket{00}+\ket{11}) \neq \frac{1}{\sqrt{2}}(\ket{00}+\ket{11})$.

\section{Equivalent constructions} \label{ap:eq_construction}

The objective of this section is to prove the equivalence between the construction of the merges and splits in \cite{Cowtan_2024} and our construction. We start of by recalling the construction of \cite{Cowtan_2024}. We focus on the Z-merge as the proof of the equivalence for the X-merge follows with the same arguments.

\begin{definition}(Z-merge/X-split in \cite{Cowtan_2024}) \label{def:merge_split_cowtan}
    Let's consider two CSS codes $(C_\bul, C^\bul)$ and $(D_\bul, D^\bul)$. Let $V_\bul$ be a chain complex such that we have two chain maps $f_\bul$ and $g_\bul$ defined as
    \begin{equation}
        \begin{tikzcd}
            V_\bul \arrow[r, "f_\bul"] \arrow[d, "g_\bul"'] & C_\bul \\
            D_\bul & 
        \end{tikzcd}.
    \end{equation}
    Then, the \textbf{Z-merged code} $(Q_\bul,Q^\bul)$ is the result of the pushout
    \begin{equation}
        \begin{tikzcd}
            V_\bul \arrow[r, "g_\bul"]\arrow[d, "f_\bul"'] & D_\bul \arrow[d,"l_\bul"]\\
            C_\bul \arrow[r,"k_\bul"'] & Q_\bul\arrow[ul, phantom, "\usebox\pushout", very near start]
        \end{tikzcd}.
    \end{equation}
    The Z-preserving code map $(k - l)_\bul : (C \oplus D)_\bul \rightarrow Q_\bul$ is called a \textbf{Z-merge between $(C_\bul, C^\bul)$ and $(D_\bul, D^\bul)$ along $(V_\bul, V^\bul)$}. 
\end{definition}

The dual of the above definition provides the definition of an X-merge.
In the original construction of merges \cite{Cowtan_2024}, the apex of the span $V_\bul$ in def.~\ref{def:merge_split_cowtan} did not need to be a Z-subcode. Similarly, the chain maps $f_\bul$ and $g_\bul$ did not require to be projections. The following proposition proves that if the apex $V_\bul$ of the span is not a Z-subcode, then there exists another span meeting this requirement and producing the same Z-merged code and Z-merge.

\begin{proposition} \label{subcodes}
    Let $(C_\bul, C^\bul)$ and $(D_\bul, D^\bul)$ be two CSS codes. Consider the span $C_\bul \xleftarrow{f_\bul'} V_\bul' \xrightarrow{g_\bul'} D_\bul$ and the pushout 
    \begin{equation}
    \begin{tikzcd}
        V_\bul' \arrow[r, "g_\bul'"]\arrow[d, "f_\bul'"'] & D_\bul \arrow[d,"l_\bul"]\\
        C_\bul \arrow[r,"k_\bul"'] & Q_\bul\arrow[ul, phantom, "\usebox\pushout", very near start]
    \end{tikzcd}
    \end{equation}
    There exists a Z-subcode of $((C \oplus D)_\bul, (C \oplus   D)^\bul)$ such that the code $(Q_\bul, Q^\bul)$ is a Z-merged code.
\end{proposition}

\begin{proof}
    The claim is that there exists another span $C_\bul \xleftarrow{f_\bul} V_\bul \xrightarrow{g_\bul} D_\bul$---where $(V_\bul, V^\bul)$ is a Z-subcode of $((C \oplus D)_\bul, (C \oplus   D)^\bul)$ and $f_\bul$ and $g_\bul$ are projections onto $C_\bul$ and $D_\bul$, respectively –- leading to the same triple $(Q_\bul, k_\bul, l_\bul)$.
    Forgetting about $(V_\bul', f_\bul', g_\bul')$, we construct a new span by considering the following pullback
    \begin{equation} \label{eq:square}
        \begin{tikzcd}
            V_\bul \arrow[r, "g_\bul"]\arrow[d, "f_\bul"']  \arrow[rd, phantom, "\usebox\pullback", very near start] & D_\bul \arrow[d,"l_\bul"]\\
            C_\bul \arrow[r,"k_\bul"'] & Q_\bul
        \end{tikzcd}.
    \end{equation}
    As a result of a pullback \cite{Cowtan_2024}, $(V_\bul, V^\bul)$ is a Z-subcode of $((C \oplus D)_\bul, (C \oplus   D)^\bul)$, and for $n \in \N$ and $c+d \in C_n \oplus D_n$,
    \begin{equation}f_n(c+d) = c \quad ;\quad g_n(c+d) = d\end{equation}
    Furthermore, the pullback square above is also a pushout. In fact, according to theorem ~\ref{theorem:stacks}, the first pushout gives the exact sequence
    \begin{equation}\begin{tikzcd}
        V_\bul' \arrow[r, "f_\bul' \oplus g_\bul'"] & (C\oplus D)_\bul \arrow[r, "k_\bul - l_\bul"] & Q_\bul \arrow[r] & 0_\bul
    \end{tikzcd}.\end{equation}
    As a pushout is unique up to isomorphism, we have $\left((C\oplus D)/V\right)_\bul \cong Q_\bul$. Moreover, the pullback ensures that the following sequence is also exact
    \begin{equation}\begin{tikzcd}
        0 \arrow[r] & V_\bul \arrow[r, "f_\bul \oplus g_\bul"] & (C\oplus D)_\bul \arrow[r, "k_\bul - l_\bul"] & Q_\bul 
    \end{tikzcd}.\end{equation} 
    Hence, the square in \cref{eq:square} is both a pushout and a pullback.
\end{proof}

The same reasoning goes for an X-merge with X-subcodes. 

\section{Relation between merges and quotient merges}\label{ap:merge-quotient}

This appendix aims to establish a connection between merges and quotient merges. Specifically, we show that any merge can be decomposed into a quotient merge followed by an isomorphism, as formalised in the following proposition.

\begin{proposition}
    Consider a Z-merge $p_\bul:E_\bul \rightarrow Q_\bul$. There exists a Z-subcode $(V_\bul,V^\bul)$ such that $p_\bul=\sigma_\bul \circ \Tilde{p}_\bul$, where $\Tilde{p}_\bul:E_\bul \rightarrow (E/V)_\bul$ is a quotient Z-merge and $\sigma_\bul : (E/V)_\bul \rightarrow Q_\bul$ is an isomorphic chain map, i.e., each map $\sigma_n$ for $n \in \{0,1,2\}$ is an isomorphism.
\end{proposition}

\begin{proof}
    It is straightforward to verify that $(E/V)_\bul$ is in fact a valid chain complex and that the boundaries are well defined, since $(V_\bul,V^\bul)$ is a Z-subcode of $(E_\bul,E^\bul)$. We give the construction of $V_\bul$ and $\sigma_\bul$. The rest of the proof is essentially a diagram chasing.
    
    Choose $V_\bul=\ker p_\bul (= \ker \Tilde{p}_\bul)$. Concretely, $V_\bul$ is defined for all $n \in \{0,1,2\}$ as $V_n=\ker p_n$ and $\del_n^V=\del_n^E \restriction_{V_n}$. We can verify straightforwardly that $V_\bul$ is a well defined Z-subcode by playing diagram chasing on 
    \begin{equation}
        \begin{tikzcd}
            \ker p_2 \arrow[d, hook] & \ker p_1 \arrow[d, hook] & \ker p_0 \arrow[d, hook] \\
            E_2 \arrow[d, "p_2"] \arrow[r,"\del_2^E"] & E_1 \arrow[d, "p_1"] \arrow[r,"\del_1^E"] & E_0 \arrow[d, "p_0"] \\
            Q_2 \arrow[r,"\del_2^Q"] & Q_1 \arrow[r,"\del_1^Q"] & Q_0 
        \end{tikzcd}.
    \end{equation}
    So, we can construct the quotient Z-merge $\Tilde{p}_\bul:E_\bul \rightarrow (E/V)_\bul$. The second step is to construct $\sigma_\bul$. To do so, as for each $n \in \{0,1,2\}$, $V_n=\ker p_n$, we use the decomposition
    \begin{equation}
        \begin{tikzcd}
            E_n \arrow[r,"\Tilde{p}_n"] \arrow[dr, "p_n"'] & E_n/V_n \arrow[d,"\sigma_n"] \\
             & Q_n 
        \end{tikzcd}
    \end{equation}
    where $\sigma_n$ is the isomorphism defined as $\sigma_n([x]):=p_n(x)$. Finally, as $\Tilde{p}_\bul$ is a chain map as well as $p_\bul$, it is straightforward to verify that $\sigma_\bul$ is a valid chain map.

    All the constructions we made can be recast in a short exact sequence:
    \begin{equation}
      \begin{tikzcd}
         0_\bul \arrow[r] & \ker p_\bul \arrow[r, hook, "i_\bul"] & E_\bul \arrow[dr,"p_\bul"'] \arrow[r, "\Tilde{p}_\bul"] & (E/V)_\bul \arrow[r] & 0_\bul \\
         & & & Q_\bul \arrow[u,"\sigma_\bul"'] &
      \end{tikzcd}
    \end{equation}
    where $i_\bul$ is simply an inclusion.
\end{proof}

\section{Steane Reed-Muller code switching} \label{appendix:code_switching}

The code switching described here is formally identical to that in \cite{heußen2024efficientfaulttolerantcodeswitching}, but reinterpreted through our framework, which simplifies tracking the logical operation.
The analysis is very similar to the one that has been done in \cref{sec:naive_approach}. We encourage the reader to first read this section before reading this appendix.

The set up is the following: consider the (7-qubits) Steane code $(C_\bul, C^\bul)$. We are aiming to turn this code into the 15-qubits Reed-Muller code in which the $T$ gate is transversal. To do so, we propose to use CSS surgery. The switch from the Steane code to the Reed-Muller code goes in two steps: 
\begin{itemize}
    \item First instantiate the auxiliary code $(A_\bul, A^\bul)$ to be a 15-qubits Reed-Muller code in the logical $\ket{+}$ state.
    \item Perform a Z-merge (which simply consists in a gluing in that case) to identify a face of the auxiliary 15-qubits Reed-Muller code $(A_\bul, A^\bul)$ and of the Steane code $(C_\bul, C^\bul)$, as pictured \cref{fig:steane_reed-muller}.
\end{itemize}
The switch back is achieved by performing the corresponding X-split and measuring the logical auxiliary qubit encoded in the 15-qubit Reed-Muller code $(A_\bul, A^\bul)$ in the X-basis, with a potential correction applied if necessary. The step-by-step logical operations implemented by the switch---both ultimately reducing to the identity---are illustrated in \cref{fig:logop_code_switch}.

\begin{figure}
    \centering
    \includegraphics[width=1\linewidth]{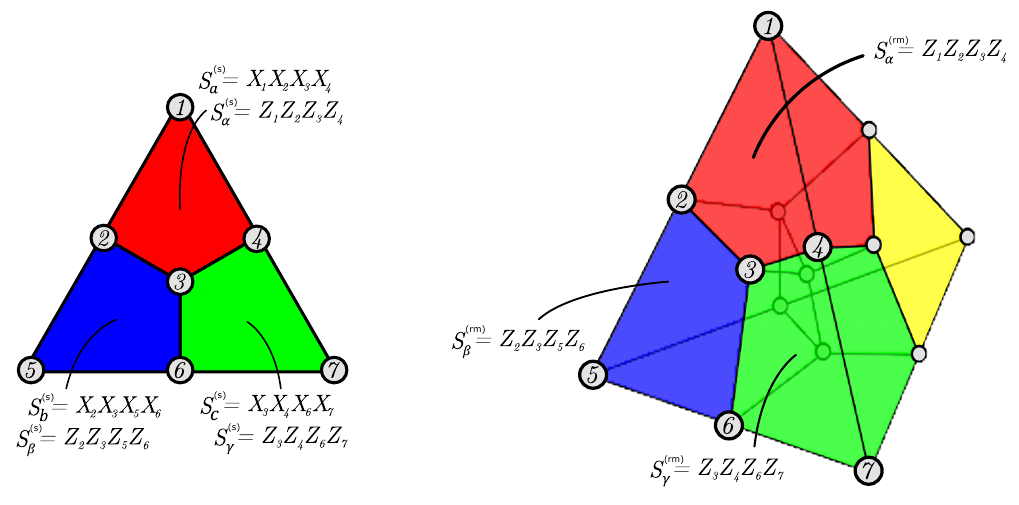}
    \caption{\textbf{Steane and 15-qubits Reed-Muller Tanner graph representations.} 
    In the Steane code (to the left), each face corresponds to two checks---an X- and a Z-type---acting on the neighbouring qubits. It is similar for the 15-qubits Reed-Muller code (to the right) where each coloured face corresponds a Z-check acting on the 4 neighbouring qubits, while each coloured cell corresponds to an X-check acting on the 8 neighbouring qubits (these are not written on the figure).}
    \label{fig:steane_reed-muller}
\end{figure}

\begin{figure}
    \centering
    \includegraphics[width=0.15\linewidth]{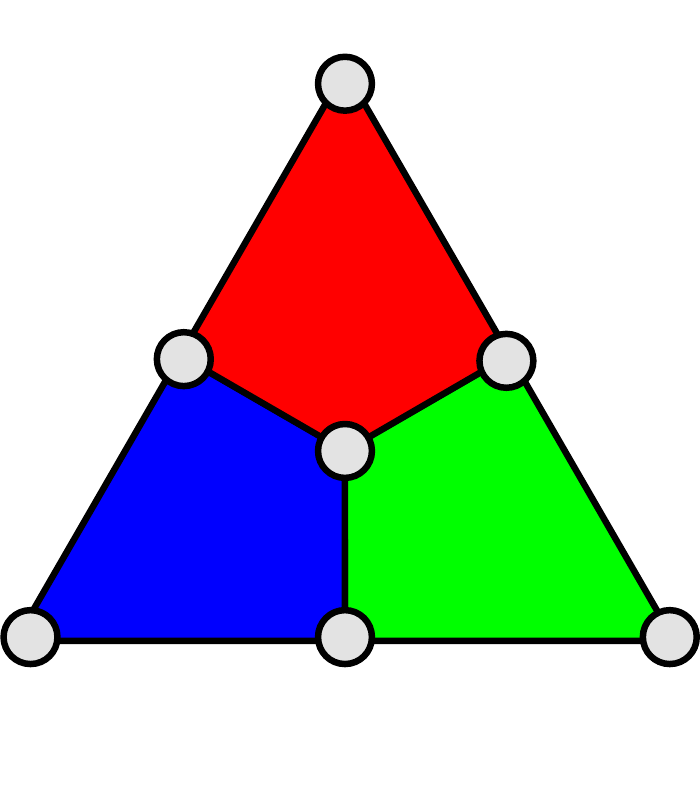}
    \quad
    \input{figures/code-switch-1.tikz}
    \quad
    \includegraphics[width=0.15\linewidth]{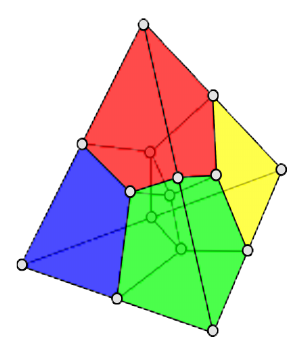}
    \quad
    \input{figures/code-switch-2.tikz}
    \quad
    \includegraphics[width=0.15\linewidth]{figures/steane_code_colour.pdf}
    \caption{
    \textbf{Switch between the Steane code and the 15-qubits Reed-Muller code.}
    The ZX-diagram on the left-hand side illustrates the step-by-step logical operations involved in transitioning to the 15-qubit Reed-Muller code. As expected in code switching, this logical operation ultimately reduces to the identity. The logical operation for switching back is simply the transpose of the previous one, which also results in the identity.
    }
    \label{fig:logop_code_switch}
\end{figure}

We now focus our attention on the Z-merge. Again, the main difficulty consists in choosing a Z-subcode such that the associated Z-merge induces the desired logical operation \cref{fig:logop_code_switch} and such that the Z-merged code is, in fact, a 15-qubits Reed-Muller code. The Z-subcode $(V_\bul,V^\bul)$ is defined to be the pairwise qubits and pairwise checks (connected to the same qubits) of the Steane code and of a face of the Reed-Muller code:
\begin{equation} \label{eq:subcode_switch}
    V_\bul \ : \ \Span \left\{s_x^{(s)}+s_x^{(rm)}\right\}_{x \in \{\alpha,\beta,\gamma\}} \longrightarrow \Span \left\{z_i^{(s)}+z_i^{(rm)}\right\}_{1\leq i \leq 7} \longrightarrow \Span \left\{s_y^{(s)}+s_y^{(rm)}\right\}_{y \in \{a,b,c\}}
\end{equation}
where the labelling of the qubits and checks is given \cref{fig:steane_reed-muller}, except for $s_a^{(rm)}$, $s_b^{(rm)}$ and $s_c^{(rm)}$ that correspond to the red, blue and green cell X-checks, respectively. 

The Z-subcode is isomorphic to a Steane code. Consequently, since the X-parity matrix of the Steane code is surjective, we get $H_0(V_\bul)=0$. This allows us to leverage \cref{sec:H_-1=0} for analyzing the logical operation.
As the vectors in the spans \cref{eq:subcode_switch} consist in sums of canonical vectors pairwise, quotienting out $(C \oplus A)_n$ by $V_n$ for $n \in \{0,1,2\}$ results in identifying every Z-operator and check pairwise. Thus, the Z-merge in this case reduces to a simple gluing operation, i.e., a code welding \cite{michnicki20123dquantumstabilizercodes}.
Reasoning with the Tanner graphs, it is clear that the Z-merged code $(Q_\bul,Q^\bul)$ is isomorphic to a 15-qubit Reed-Muller code, as intended.

Both the Steane code and the 15-qubits Reed-Muller code contain a single logical qubit, and therefore a single Z-logical operator that we call $[z_s] \in H_1(C_\bul)$ and $[z_{rm}] \in H_1(A_\bul)$, respectively. Note that $[z_s]+[z_{rm}] \in H_1(V_\bul)$ is a non-trivial Z-logical operator of $(V_\bul, V^\bul)$ once embedded in the code $((C\oplus A)_\bul,(C\oplus A)^\bul)$ via the inclusion map $i_{1*}$. Hence, $\im i_{1*}=\Span \{[z_s]+[z_{rm}]\}$. So, according to \cref{sec:H_-1=0}, 
\begin{equation}
    H_1(Q_\bul)\cong H_1((C \oplus A)_\bul)/\im i_{1*}=\Span \{[z_s],[z_{rm}]\} / \Span \{[z_s]+[z_{rm}]\}.
\end{equation}
As a consequence, the projection $p_{1*} : H_1((C \oplus A)_\bul) \rightarrow H_1(Q_\bul)$ is always represented by $P_{1*}=\begin{pmatrix}
    1 & 1
\end{pmatrix}$, regardless of the choice of bases. Hence, the logical operation
\begin{equation}
    \input{p_0star.tikz}=\input{Z-spider_K_l.tikz} \ .
\end{equation}

\section{Non optimality of split codes regarding the Singleton bound} \label{ap:singleton_bound}

Recall that the claim made in \cref{sec:ft} is that regardless of the choices of two CSS codes $(C_\bul,C^\bul)$ and $(A_\bul,A^\bul)$, the code $((C\oplus A)_\bul,(C\oplus A)^\bul)$ cannot be an MDS code, i.e., cannot saturate the quantum Singleton bound (in the qubit case) that applies to any stabiliser code $[[n,k,d]]$ \cite{rains1997nonbinaryquantumcodes}:
\begin{equation} \label{eq:singleton_bound}
    n-k \geq 2(d-1).
\end{equation}
We denote $[[n_A,k_A,d_A]]$, $[[n_C,k_C,d_C]]$ and $[[n_{C\oplus A},k_{C\oplus A},d_{C\oplus A}]]$ the code parameters of $(A_\bul,A^\bul)$, $(C_\bul,C^\bul)$ and $((C\oplus A)_\bul,(C\oplus A)^\bul)$, respectively. According to def.~\ref{def:direct_sum}, it is straightforward to see that these quantities are linked by the equations
\begin{equation}
    n_{C\oplus A}=n_C + n_A \quad ; \quad k_{C\oplus A}=k_C + k_A \quad ; \quad d_{C\oplus A}=\min (d_C,d_A).
\end{equation}
Applying the Singleton bound to the codes $(C_\bul,C^\bul)$ and $(A_\bul,A^\bul)$ gives:
\begin{equation}
    \left\{\begin{array}{c}
        n_C-k_C \geq 2(d_C-1) \\
        n_A-k_A \geq 2(d_A-1)
    \end{array}\right. \quad \Rightarrow \quad n_{C\oplus A}-k_{C\oplus A} \geq 2(d_C+d_A-1).
\end{equation}
Hence, because $d_C+d_A > \min (d_C,d_A)$, we have the strict inequality $n_{C\oplus A}-k_{C\oplus A} >2(d_{C\oplus A}-1)$. Thus, the code $((C\oplus A)_\bul,(C\oplus A)^\bul)$ can never claim to saturate the Singleton bound \cref{eq:singleton_bound}.

\end{document}

%% file: style/zx.tikzstyles
\tikzstyle{boundary vertex}=[inner sep=0mm, minimum size=1mm, shape=circle, draw=black, fill=black]
\tikzstyle{grey_dot}=[fill={rgb,255: red,191; green,191; blue,191}, draw={rgb,255: red,191; green,191; blue,191}, shape=circle, minimum size=1mm, inner sep=0mm]
\tikzstyle{blue_dot}=[fill={rgb,255: red,202; green,251; blue,255}, draw=black, shape=circle, minimum size=1mm, inner sep=0mm]
\tikzstyle{white_dot}=[fill=white, draw=black, shape=circle, minimum size=1.5mm, inner sep=0mm]

\tikzstyle{arrow}=[->]
\tikzstyle{red_arrow}=[->, draw=red]
\tikzstyle{cyan_arrow}=[->, draw=cyan]
\tikzstyle{red_dash}=[-, dashed, draw=red]
\tikzstyle{grey dash}=[-, fill=none, draw={rgb,255: red,191; green,191; blue,191}, dashed]
\tikzstyle{dashed arrow}=[->, dashed]
\tikzstyle{blue_edge}=[-, draw={rgb,255: red,46; green,126; blue,255}]
\tikzstyle{red_edge}=[-, draw=red]

\tikzstyle{gate}=[shape=rectangle, text height=1.5ex, text depth=0.25ex, yshift=0.5mm, fill=white, draw=black, minimum height=5mm, yshift=-0.5mm, minimum width=5mm, font={\small}, tikzit category=circuit]
\tikzstyle{big gate}=[shape=rectangle, text height=1.5ex, text depth=0.25ex, yshift=0.5mm, fill=white, draw=black, minimum height=18mm, yshift=-0.5mm, minimum width=5mm, font={\small}, tikzit category=circuit]
\tikzstyle{Z dot}=[inner sep=0mm, minimum size=2mm, shape=circle, draw=black, fill={rgb,255: red,221; green,255; blue,221}, tikzit category=zx]
\tikzstyle{Z phase dot}=[minimum size=5mm, font={\footnotesize\boldmath}, shape=rectangle, rounded corners=2mm, inner sep=0.2mm, outer sep=-2mm, scale=0.8, tikzit shape=circle, draw=black, fill={rgb,255: red,221; green,255; blue,221}, tikzit draw=blue, tikzit category=zx]
\tikzstyle{X dot}=[Z dot, shape=circle, draw=black, fill={rgb,255: red,255; green,136; blue,136}, tikzit category=zx]
\tikzstyle{X phase dot}=[Z phase dot, tikzit shape=circle, tikzit draw=blue, fill={rgb,255: red,255; green,136; blue,136}, font={\footnotesize\boldmath}, tikzit category=zx]
\tikzstyle{hadamard}=[fill=yellow, draw=black, shape=rectangle, inner sep=0.6mm, minimum height=1.5mm, minimum width=1.5mm, tikzit category=zx]
\tikzstyle{paulibox}=[fill={rgb,255: red,221; green,221; blue,255}, draw=black, shape=rectangle, inner sep=0.6mm, minimum height=5mm, minimum width=5mm, font={\footnotesize}, text height=1.5ex, text depth=0.25ex, tikzit category=zx]
\tikzstyle{vertex}=[inner sep=0mm, minimum size=1mm, shape=circle, draw=black, fill=black, tikzit category=misc]
\tikzstyle{vertex set}=[inner sep=0mm, minimum size=1mm, shape=circle, draw=black, fill=white, font={\footnotesize\boldmath}, tikzit category=misc]
\tikzstyle{small black dot}=[fill=black, draw=black, shape=circle, inner sep=0pt, minimum width=1.2mm, tikzit category=circuit]
\tikzstyle{cnot ctrl}=[fill=black, draw=black, shape=circle, inner sep=0pt, minimum width=1.2mm, tikzit category=circuit]
\tikzstyle{cnot targ}=[fill=white, draw=white, shape=circle, tikzit category=circuit, label={center:$\oplus$}, inner sep=0pt, minimum width=2.1mm, tikzit fill={rgb,255: red,102; green,204; blue,255}, tikzit draw=black]
\tikzstyle{ket}=[fill=white, draw=black, shape=regular polygon, regular polygon sides=3, regular polygon rotate=-30, scale=0.7, inner sep=1pt, tikzit category=circuit, tikzit shape=rectangle, tikzit fill=green]
\tikzstyle{bra}=[fill=white, draw=black, shape=regular polygon, regular polygon sides=3, regular polygon rotate=30, scale=0.7, inner sep=1pt, tikzit category=circuit, tikzit shape=rectangle, tikzit fill=red]
\tikzstyle{scalar}=[shape=rectangle, text height=1.5ex, text depth=0.25ex, yshift=0.5mm, fill=white, draw=black, minimum height=5mm, yshift=-0.5mm, minimum width=5mm, font={\small}]
\tikzstyle{clabel}=[fill=white, draw=none, shape=rectangle, tikzit fill={rgb,255: red,56; green,255; blue,242}, font={\footnotesize}, inner sep=1pt, tikzit category=labels]
\tikzstyle{empty diagram}=[draw={gray!40!white}, dashed, shape=rectangle, minimum width=1cm, minimum height=1cm, tikzit category=misc]

\tikzstyle{label dot}=[minimum size=5mm, font={\footnotesize\boldmath}, shape=rectangle, rounded corners=2mm, inner sep=0.2mm, outer sep=-2mm, scale=0.8, tikzit shape=circle, draw=black, tikzit draw=blue, tikzit category=zx]

\tikzstyle{hadamard edge}=[-, dashed, dash pattern=on 2pt off 0.5pt, thick, draw={rgb,255: red,68; green,136; blue,255}]
\tikzstyle{box edge}=[-, dashed, dash pattern=on 2pt off 0.5pt, thick, draw={rgb,255: red,203; green,192; blue,225}]
\tikzstyle{brace edge}=[-, tikzit draw=blue, decorate, decoration={brace,amplitude=1mm,raise=-1mm}]
\tikzstyle{diredge}=[->]
\tikzstyle{double edge}=[-, double, shorten <=-1mm, shorten >=-1mm, double distance=2pt]
\tikzstyle{gray edge}=[-, {gray!60!white}]
\tikzstyle{pointer edge}=[->, very thick, gray]
\tikzstyle{boldedge}=[-, line width=1.6pt, shorten <=-0.17mm, shorten >=-0.17mm]

\tikzstyle{rmat}=[draw, signal, fill=gray!30, signal to=east, signal from=west, inner sep=1.5pt, minimum height=9pt]
\tikzstyle{lmat}=[draw, signal, fill=gray!30, signal to=west, signal from=east, inner sep=1.5pt, minimum height=9pt]
\tikzstyle{ggreen}=[fill=green, draw=black, shape=circle, tikzit category=SZX, tikzit fill=green, tikzit draw=black, line width=1pt, inner sep=2pt]
\tikzstyle{gred}=[fill=red, draw=black, shape=circle, rounded corners=2mm,  tikzit category=SZX, inner sep=2pt, tikzit fill=red, line width=1pt]
\tikzstyle{divide}=[regular polygon, regular polygon sides=3, shape border rotate=90, draw=black,fill=gray!30, inner sep=1.6pt, tikzit category=scal, rounded corners=0.8mm]
\tikzstyle{gather}=[fill=gray!30, draw=black, tikzit category=scal, rounded corners=0.8mm, regular polygon, regular polygon sides=3, shape border rotate=-90, inner sep=1.6pt]
\tikzstyle{A}=[fill=white, shape=circle, tikzit category=scal, inner sep=1pt]

\tikzstyle{very thick}=[-, line width=1pt, tikzit draw=red]

%% file: style/circuits.tikzstyles
\tikzstyle{control}=[thick, fill=black, circle, scale=1, inner sep=.05cm]
\tikzstyle{target}=[inner sep=2.5pt, draw, circle, path picture={ \draw[black](path picture bounding box.east) -- (path picture bounding box.west) (path picture bounding box.south) -- (path picture bounding box.north);}]
\tikzstyle{gate}=[fill=white, draw=black, shape=rectangle]

\tikzstyle{bgate}=[-, fill=white]

%% file: figures/SZX_f_0_surjective.tikz
\begin{tikzpicture}
	\begin{pgfonlayer}{nodelayer}
		\node [style=lmat] (0) at (0, 0) {};
		\node [style=none] (1) at (-1.5, 0) {};
		\node [style=none] (2) at (1.5, 0) {};
		\node [style=none] (7) at (0, 0.8) {\footnotesize $f_1^T$};
	\end{pgfonlayer}
	\begin{pgfonlayer}{edgelayer}
		\draw [style=very thick] (1.center) to (2.center);
	\end{pgfonlayer}
\end{tikzpicture}

%% file: figures/two_surface.tikz
\begin{tikzpicture}
    \draw (-0.25,0.1) rectangle (0.25,0.6);
    \draw (-0.25,-0.1) rectangle (0.25,-0.6);
\end{tikzpicture}

%% file: figures/big_surface.tikz
\begin{tikzpicture}
    \draw (-0.25,-0.6) rectangle (0.25,0.6);
\end{tikzpicture}

%% file: figures/SZX_wire.tikz
\begin{tikzpicture}
	\begin{pgfonlayer}{nodelayer}
		\node [style=none] (1) at (-1, 0) {};
		\node [style=none] (2) at (1, 0) {};

	\end{pgfonlayer}
	\begin{pgfonlayer}{edgelayer}
		\draw [style=very thick] (1.center) to (2.center);
	\end{pgfonlayer}
\end{tikzpicture}

%% file: figures/SZX_id.tikz
\begin{tikzpicture}
	\begin{pgfonlayer}{nodelayer}
		\node [style=none] (0) at (-1, 0) {};
		\node [style=none] (1) at (1, 0) {};
		\node [style=lmat] (2) at (0, 0) {};
		\node [style=none] (3) at (0, 0.75) {\footnotesize $id$};
	\end{pgfonlayer}
	\begin{pgfonlayer}{edgelayer}
		\draw [style=very thick] (0.center) to (1.center);
	\end{pgfonlayer}
\end{tikzpicture}

%% file: sample.bib
@article{Breuckmann_2024,
   title={Fold-Transversal Clifford Gates for Quantum Codes},
   volume={8},
   ISSN={2521-327X},
   url={http://dx.doi.org/10.22331/q-2024-06-13-1372},
   DOI={10.22331/q-2024-06-13-1372},
   journal={Quantum},
   publisher={Verein zur Forderung des Open Access Publizierens in den Quantenwissenschaften},
   author={Breuckmann, Nikolas P. and Burton, Simon},
   year={2024},
   month=jun, pages={1372} }

@article{Quintavalle_2023,
   title={Partitioning qubits in hypergraph product codes to implement logical gates},
   volume={7},
   ISSN={2521-327X},
   url={http://dx.doi.org/10.22331/q-2023-10-24-1153},
   DOI={10.22331/q-2023-10-24-1153},
   journal={Quantum},
   publisher={Verein zur Forderung des Open Access Publizierens in den Quantenwissenschaften},
   author={Quintavalle, Armanda O. and Webster, Paul and Vasmer, Michael},
   year={2023},
   month=oct, pages={1153} }

@misc{sayginel2024faulttolerantlogicalcliffordgates,
      title={Fault-Tolerant Logical Clifford Gates from Code Automorphisms}, 
      author={Hasan Sayginel and Stergios Koutsioumpas and Mark Webster and Abhishek Rajput and Dan E Browne},
      year={2024},
      eprint={2409.18175},
      archivePrefix={arXiv},
      primaryClass={quant-ph},
      url={https://arxiv.org/abs/2409.18175}, 
}

@misc{malcolm2025computingefficientlyqldpccodes,
      title={Computing Efficiently in QLDPC Codes}, 
      author={Alexander J. Malcolm and Andrew N. Glaudell and Patricio Fuentes and Daryus Chandra and Alexis Schotte and Colby DeLisle and Rafael Haenel and Amir Ebrahimi and Joschka Roffe and Armanda O. Quintavalle and Stefanie J. Beale and Nicholas R. Lee-Hone and Stephanie Simmons},
      year={2025},
      eprint={2502.07150},
      archivePrefix={arXiv},
      primaryClass={quant-ph},
      url={https://arxiv.org/abs/2502.07150}, 
}

@misc{swaroop2024universaladaptersquantumldpc,
      title={Universal adapters between quantum LDPC codes}, 
      author={Esha Swaroop and Tomas Jochym-O'Connor and Theodore J. Yoder},
      year={2024},
      eprint={2410.03628},
      archivePrefix={arXiv},
      primaryClass={quant-ph},
      url={https://arxiv.org/abs/2410.03628}, 
}

@misc{hillmann2024singleshotmeasurementbasedquantumerror,
      title={Single-shot and measurement-based quantum error correction via fault complexes}, 
      author={Timo Hillmann and Guillaume Dauphinais and Ilan Tzitrin and Michael Vasmer},
      year={2024},
      eprint={2410.12963},
      archivePrefix={arXiv},
      primaryClass={quant-ph},
      url={https://arxiv.org/abs/2410.12963}, 
}

@misc{scruby2024highthresholdlowoverheadsingleshotdecodable,
      title={High-threshold, low-overhead and single-shot decodable fault-tolerant quantum memory}, 
      author={Thomas R. Scruby and Timo Hillmann and Joschka Roffe},
      year={2024},
      eprint={2406.14445},
      archivePrefix={arXiv},
      primaryClass={quant-ph},
      url={https://arxiv.org/abs/2406.14445}, 
}

@article{Campbell_2017,
   title={Roads towards fault-tolerant universal quantum computation},
   volume={549},
   ISSN={1476-4687},
   url={http://dx.doi.org/10.1038/nature23460},
   DOI={10.1038/nature23460},
   number={7671},
   journal={Nature},
   publisher={Springer Science and Business Media LLC},
   author={Campbell, Earl T. and Terhal, Barbara M. and Vuillot, Christophe},
   year={2017},
   month=sep, pages={172–179} }

@article{Webster_2022,
   title={Universal fault-tolerant quantum computing with stabilizer codes},
   volume={4},
   ISSN={2643-1564},
   url={http://dx.doi.org/10.1103/PhysRevResearch.4.013092},
   DOI={10.1103/physrevresearch.4.013092},
   number={1},
   journal={Physical Review Research},
   publisher={American Physical Society (APS)},
   author={Webster, Paul and Vasmer, Michael and Scruby, Thomas R. and Bartlett, Stephen D.},
   year={2022},
   month=feb }

@article{Eastin_2009,
   title={Restrictions on Transversal Encoded Quantum Gate Sets},
   volume={102},
   ISSN={1079-7114},
   url={http://dx.doi.org/10.1103/PhysRevLett.102.110502},
   DOI={10.1103/physrevlett.102.110502},
   number={11},
   journal={Physical Review Letters},
   publisher={American Physical Society (APS)},
   author={Eastin, Bryan and Knill, Emanuel},
   year={2009},
   month=mar }

@article{Bravyi_2024,
   title={High-threshold and low-overhead fault-tolerant quantum memory},
   volume={627},
   ISSN={1476-4687},
   url={http://dx.doi.org/10.1038/s41586-024-07107-7},
   DOI={10.1038/s41586-024-07107-7},
   number={8005},
   journal={Nature},
   publisher={Springer Science and Business Media LLC},
   author={Bravyi, Sergey and Cross, Andrew W. and Gambetta, Jay M. and Maslov, Dmitri and Rall, Patrick and Yoder, Theodore J.},
   year={2024},
   month=mar, pages={778–782} }

@article{Breuckmann_2021,
   title={Balanced Product Quantum Codes},
   volume={67},
   ISSN={1557-9654},
   url={http://dx.doi.org/10.1109/TIT.2021.3097347},
   DOI={10.1109/tit.2021.3097347},
   number={10},
   journal={IEEE Transactions on Information Theory},
   publisher={Institute of Electrical and Electronics Engineers (IEEE)},
   author={Breuckmann, Nikolas P. and Eberhardt, Jens N.},
   year={2021},
   month=oct, pages={6653–6674} }

@article{Tillich_2014,
   title={Quantum LDPC Codes With Positive Rate and Minimum Distance Proportional to the Square Root of the Blocklength},
   volume={60},
   ISSN={1557-9654},
   url={http://dx.doi.org/10.1109/TIT.2013.2292061},
   DOI={10.1109/tit.2013.2292061},
   number={2},
   journal={IEEE Transactions on Information Theory},
   publisher={Institute of Electrical and Electronics Engineers (IEEE)},
   author={Tillich, Jean-Pierre and Zemor, Gilles},
   year={2014},
   month=feb, pages={1193–1202} }

@misc{kissinger2022phasefree,
      title={Phase-free ZX diagrams are CSS codes (...or how to graphically grok the surface code)}, 
      author={Aleks Kissinger},
      year={2022},
      eprint={2204.14038},
      archivePrefix={arXiv},
      primaryClass={quant-ph}
}

@inproceedings{https://doi.org/10.4230/lipics.mfcs.2019.55,
  doi = {10.4230/LIPICS.MFCS.2019.55},
  
  url = {https://drops.dagstuhl.de/entities/document/10.4230/LIPIcs.MFCS.2019.55},
  
  author = {Carette, Titouan and Horsman, Dominic and Perdrix, Simon},
  
  language = {en},
  
  title = {SZX-Calculus: Scalable Graphical Quantum Reasoning},
  
  publisher = {Schloss Dagstuhl – Leibniz-Zentrum für Informatik},
  
  year = {2019},
  
  copyright = {Creative Commons Attribution 3.0 Unported license}
}

@article{Cowtan_2024,
   title={CSS code surgery as a universal construction},
   volume={8},
   ISSN={2521-327X},
   url={http://dx.doi.org/10.22331/q-2024-05-14-1344},
   DOI={10.22331/q-2024-05-14-1344},
   journal={Quantum},
   publisher={Verein zur Forderung des Open Access Publizierens in den Quantenwissenschaften},
   author={Cowtan, Alexander and Burton, Simon},
   year={2024},
   month=may, pages={1344} }

@inproceedings{Rengaswamy_2018,
   title={Synthesis of Logical Clifford Operators via Symplectic Geometry},
   url={http://dx.doi.org/10.1109/ISIT.2018.8437652},
   DOI={10.1109/isit.2018.8437652},
   booktitle={2018 IEEE International Symposium on Information Theory (ISIT)},
   publisher={IEEE},
   author={Rengaswamy, Narayanan and Calderbank, Robert and Pfister, Henry D. and Kadhe, Swanand},
   year={2018},
   month=jun }

@book{Hatcher:478079,
      author        = "Hatcher, Allen",
      title         = "{Algebraic topology}",
      publisher     = "Cambridge Univ. Press",
      address       = "Cambridge",
      year          = "2000",
      url           = "https://cds.cern.ch/record/478079",
}

@misc{cohen2022lowoverheadfaulttolerantquantumcomputing,
      title={Low-overhead fault-tolerant quantum computing using long-range connectivity}, 
      author={Lawrence Z. Cohen and Isaac H. Kim and Stephen D. Bartlett and Benjamin J. Brown},
      year={2022},
      eprint={2110.10794},
      archivePrefix={arXiv},
      primaryClass={quant-ph},
      doi={https://doi.org/10.1126/sciadv.abn1717},
      url={https://arxiv.org/abs/2110.10794}, 
}

@misc{xu2023constantoverheadfaulttolerantquantumcomputation,
      title={Constant-Overhead Fault-Tolerant Quantum Computation with Reconfigurable Atom Arrays}, 
      author={Qian Xu and J. Pablo Bonilla Ataides and Christopher A. Pattison and Nithin Raveendran and Dolev Bluvstein and Jonathan Wurtz and Bane Vasic and Mikhail D. Lukin and Liang Jiang and Hengyun Zhou},
      year={2023},
      eprint={2308.08648},
      archivePrefix={arXiv},
      primaryClass={quant-ph},
      url={https://arxiv.org/abs/2308.08648}, 
}

@misc{panteleev2022asymptoticallygoodquantumlocally,
      title={Asymptotically Good Quantum and Locally Testable Classical LDPC Codes}, 
      author={Pavel Panteleev and Gleb Kalachev},
      year={2022},
      eprint={2111.03654},
      archivePrefix={arXiv},
      primaryClass={cs.IT},
      url={https://arxiv.org/abs/2111.03654}, 
}

@article{Horsman_2012,
   title={Surface code quantum computing by lattice surgery},
   volume={14},
   ISSN={1367-2630},
   url={http://dx.doi.org/10.1088/1367-2630/14/12/123011},
   DOI={10.1088/1367-2630/14/12/123011},
   number={12},
   journal={New Journal of Physics},
   publisher={IOP Publishing},
   author={Horsman, Dominic and Fowler, Austin G and Devitt, Simon and Meter, Rodney Van},
   year={2012},
   month=dec, pages={123011} }

@article{de_Beaudrap_2020,
   title={The ZX calculus is a language for surface code lattice surgery},
   volume={4},
   ISSN={2521-327X},
   url={http://dx.doi.org/10.22331/q-2020-01-09-218},
   DOI={10.22331/q-2020-01-09-218},
   journal={Quantum},
   publisher={Verein zur Forderung des Open Access Publizierens in den Quantenwissenschaften},
   author={de Beaudrap, Niel and Horsman, Dominic},
   year={2020},
   month=jan, pages={218} }

@misc{vuillot2023homologicalquantumrotorcodes,
      title={Homological Quantum Rotor Codes: Logical Qubits from Torsion}, 
      author={Christophe Vuillot and Alessandro Ciani and Barbara M. Terhal},
      year={2023},
      eprint={2303.13723},
      archivePrefix={arXiv},
      primaryClass={quant-ph},
      url={https://arxiv.org/abs/2303.13723}, 
}

@misc{booth2024graphicalsymplecticalgebra,
      title={Graphical Symplectic Algebra}, 
      author={Robert I. Booth and Titouan Carette and Cole Comfort},
      year={2024},
      eprint={2401.07914},
      archivePrefix={arXiv},
      primaryClass={cs.LO},
      url={https://arxiv.org/abs/2401.07914}, 
}

@misc{delfosse2023spacetimecodescliffordcircuits,
      title={Spacetime codes of Clifford circuits}, 
      author={Nicolas Delfosse and Adam Paetznick},
      year={2023},
      eprint={2304.05943},
      archivePrefix={arXiv},
      primaryClass={quant-ph},
      url={https://arxiv.org/abs/2304.05943}, 
}

@article{Hastings_2021,
   title={Dynamically Generated Logical Qubits},
   volume={5},
   ISSN={2521-327X},
   url={http://dx.doi.org/10.22331/q-2021-10-19-564},
   DOI={10.22331/q-2021-10-19-564},
   journal={Quantum},
   publisher={Verein zur Forderung des Open Access Publizierens in den Quantenwissenschaften},
   author={Hastings, Matthew B. and Haah, Jeongwan},
   year={2021},
   month=oct, pages={564} }

@article{Panteleev_2022,
   title={Quantum LDPC Codes With Almost Linear Minimum Distance},
   volume={68},
   ISSN={1557-9654},
   url={http://dx.doi.org/10.1109/TIT.2021.3119384},
   DOI={10.1109/tit.2021.3119384},
   number={1},
   journal={IEEE Transactions on Information Theory},
   publisher={Institute of Electrical and Electronics Engineers (IEEE)},
   author={Panteleev, Pavel and Kalachev, Gleb},
   year={2022},
   month=jan, pages={213–229} }

@misc{audoux2018tensorproductscsscodes,
      title={On tensor products of CSS Codes}, 
      author={Benjamin Audoux and Alain Couvreur},
      year={2018},
      eprint={1512.07081},
      archivePrefix={arXiv},
      primaryClass={cs.IT},
      url={https://arxiv.org/abs/1512.07081}, 
}

@misc{cowtan2024ssipautomatedsurgeryquantum,
      title={SSIP: automated surgery with quantum LDPC codes}, 
      author={Alexander Cowtan},
      year={2024},
      eprint={2407.09423},
      archivePrefix={arXiv},
      primaryClass={quant-ph},
      url={https://arxiv.org/abs/2407.09423}, 
}

@misc{leinster2016basiccategorytheory,
      title={Basic Category Theory}, 
      author={Tom Leinster},
      year={2016},
      eprint={1612.09375},
      archivePrefix={arXiv},
      primaryClass={math.CT},
      url={https://arxiv.org/abs/1612.09375}, 
}

@misc{lane1978categoriesfortheworkingmathematician,
      title={Categories for the working mathematician}, 
      author={Saunders Mac Lane},
      year={1978},
      month=jan,
      ISSN={0072-5285},
      url={https://link.springer.com/book/10.1007/978-1-4757-4721-8},
      publisher={Springer New York, NY},
}

@article{Quintavalle_2021,
   title={Single-Shot Error Correction of Three-Dimensional Homological Product Codes},
   volume={2},
   ISSN={2691-3399},
   url={http://dx.doi.org/10.1103/PRXQuantum.2.020340},
   DOI={10.1103/prxquantum.2.020340},
   number={2},
   journal={PRX Quantum},
   publisher={American Physical Society (APS)},
   author={Quintavalle, Armanda O. and Vasmer, Michael and Roffe, Joschka and Campbell, Earl T.},
   year={2021},
   month=jun }

@thesis{gross_finite_2005,
	title = {Finite Phase Space Methods in Quantum Information},
	url = {https://www.thp.uni-koeln.de/gross/files/diplom.pdf},
	institution = {Universitat Postdam},
	type = {Master},
	author = {Gross, David},
	date = {2005},
}

@misc{acharya2024quantumerrorcorrectionsurface,
      title={Quantum error correction below the surface code threshold}, 
      author={Rajeev Acharya and Laleh Aghababaie-Beni and Igor Aleiner and Trond I. Andersen and Markus Ansmann and Frank Arute and Kunal Arya and Abraham Asfaw and Nikita Astrakhantsev and Juan Atalaya and Ryan Babbush and Dave Bacon and Brian Ballard and Joseph C. Bardin and Johannes Bausch and Andreas Bengtsson and Alexander Bilmes and Sam Blackwell and Sergio Boixo and Gina Bortoli and Alexandre Bourassa and Jenna Bovaird and Leon Brill and Michael Broughton and David A. Browne and Brett Buchea and Bob B. Buckley and David A. Buell and Tim Burger and Brian Burkett and Nicholas Bushnell and Anthony Cabrera and Juan Campero and Hung-Shen Chang and Yu Chen and Zijun Chen and Ben Chiaro and Desmond Chik and Charina Chou and Jahan Claes and Agnetta Y. Cleland and Josh Cogan and Roberto Collins and Paul Conner and William Courtney and Alexander L. Crook and Ben Curtin and Sayan Das and Alex Davies and Laura De Lorenzo and Dripto M. Debroy and Sean Demura and Michel Devoret and Agustin Di Paolo and Paul Donohoe and Ilya Drozdov and Andrew Dunsworth and Clint Earle and Thomas Edlich and Alec Eickbusch and Aviv Moshe Elbag and Mahmoud Elzouka and Catherine Erickson and Lara Faoro and Edward Farhi and Vinicius S. Ferreira and Leslie Flores Burgos and Ebrahim Forati and Austin G. Fowler and Brooks Foxen and Suhas Ganjam and Gonzalo Garcia and Robert Gasca and Élie Genois and William Giang and Craig Gidney and Dar Gilboa and Raja Gosula and Alejandro Grajales Dau and Dietrich Graumann and Alex Greene and Jonathan A. Gross and Steve Habegger and John Hall and Michael C. Hamilton and Monica Hansen and Matthew P. Harrigan and Sean D. Harrington and Francisco J. H. Heras and Stephen Heslin and Paula Heu and Oscar Higgott and Gordon Hill and Jeremy Hilton and George Holland and Sabrina Hong and Hsin-Yuan Huang and Ashley Huff and William J. Huggins and Lev B. Ioffe and Sergei V. Isakov and Justin Iveland and Evan Jeffrey and Zhang Jiang and Cody Jones and Stephen Jordan and Chaitali Joshi and Pavol Juhas and Dvir Kafri and Hui Kang and Amir H. Karamlou and Kostyantyn Kechedzhi and Julian Kelly and Trupti Khaire and Tanuj Khattar and Mostafa Khezri and Seon Kim and Paul V. Klimov and Andrey R. Klots and Bryce Kobrin and Pushmeet Kohli and Alexander N. Korotkov and Fedor Kostritsa and Robin Kothari and Borislav Kozlovskii and John Mark Kreikebaum and Vladislav D. Kurilovich and Nathan Lacroix and David Landhuis and Tiano Lange-Dei and Brandon W. Langley and Pavel Laptev and Kim-Ming Lau and Loïck Le Guevel and Justin Ledford and Kenny Lee and Yuri D. Lensky and Shannon Leon and Brian J. Lester and Wing Yan Li and Yin Li and Alexander T. Lill and Wayne Liu and William P. Livingston and Aditya Locharla and Erik Lucero and Daniel Lundahl and Aaron Lunt and Sid Madhuk and Fionn D. Malone and Ashley Maloney and Salvatore Mandrá and Leigh S. Martin and Steven Martin and Orion Martin and Cameron Maxfield and Jarrod R. McClean and Matt McEwen and Seneca Meeks and Anthony Megrant and Xiao Mi and Kevin C. Miao and Amanda Mieszala and Reza Molavi and Sebastian Molina and Shirin Montazeri and Alexis Morvan and Ramis Movassagh and Wojciech Mruczkiewicz and Ofer Naaman and Matthew Neeley and Charles Neill and Ani Nersisyan and Hartmut Neven and Michael Newman and Jiun How Ng and Anthony Nguyen and Murray Nguyen and Chia-Hung Ni and Thomas E. O'Brien and William D. Oliver and Alex Opremcak and Kristoffer Ottosson and Andre Petukhov and Alex Pizzuto and John Platt and Rebecca Potter and Orion Pritchard and Leonid P. Pryadko and Chris Quintana and Ganesh Ramachandran and Matthew J. Reagor and David M. Rhodes and Gabrielle Roberts and Eliott Rosenberg and Emma Rosenfeld and Pedram Roushan and Nicholas C. Rubin and Negar Saei and Daniel Sank and Kannan Sankaragomathi and Kevin J. Satzinger and Henry F. Schurkus and Christopher Schuster and Andrew W. Senior and Michael J. Shearn and Aaron Shorter and Noah Shutty and Vladimir Shvarts and Shraddha Singh and Volodymyr Sivak and Jindra Skruzny and Spencer Small and Vadim Smelyanskiy and W. Clarke Smith and Rolando D. Somma and Sofia Springer and George Sterling and Doug Strain and Jordan Suchard and Aaron Szasz and Alex Sztein and Douglas Thor and Alfredo Torres and M. Mert Torunbalci and Abeer Vaishnav and Justin Vargas and Sergey Vdovichev and Guifre Vidal and Benjamin Villalonga and Catherine Vollgraff Heidweiller and Steven Waltman and Shannon X. Wang and Brayden Ware and Kate Weber and Theodore White and Kristi Wong and Bryan W. K. Woo and Cheng Xing and Z. Jamie Yao and Ping Yeh and Bicheng Ying and Juhwan Yoo and Noureldin Yosri and Grayson Young and Adam Zalcman and Yaxing Zhang and Ningfeng Zhu and Nicholas Zobrist},
      year={2024},
      eprint={2408.13687},
      archivePrefix={arXiv},
      primaryClass={quant-ph},
      url={https://arxiv.org/abs/2408.13687}, 
}

@article{Litinski_2019,
   title={A Game of Surface Codes: Large-Scale Quantum Computing with Lattice Surgery},
   volume={3},
   ISSN={2521-327X},
   url={http://dx.doi.org/10.22331/q-2019-03-05-128},
   DOI={10.22331/q-2019-03-05-128},
   journal={Quantum},
   publisher={Verein zur Forderung des Open Access Publizierens in den Quantenwissenschaften},
   author={Litinski, Daniel},
   year={2019},
   month=mar, pages={128} }

@misc{pogorelov2024experimentalfaulttolerantcodeswitching,
      title={Experimental fault-tolerant code switching}, 
      author={Ivan Pogorelov and Friederike Butt and Lukas Postler and Christian D. Marciniak and Philipp Schindler and Markus Müller and Thomas Monz},
      year={2024},
      eprint={2403.13732},
      archivePrefix={arXiv},
      primaryClass={quant-ph},
      url={https://arxiv.org/abs/2403.13732}, 
}

@misc{vandewetering2020zxcalculusworkingquantumcomputer,
      title={ZX-calculus for the working quantum computer scientist}, 
      author={John van de Wetering},
      year={2020},
      eprint={2012.13966},
      archivePrefix={arXiv},
      primaryClass={quant-ph},
      url={https://arxiv.org/abs/2012.13966}, 
}

@misc{michnicki20123dquantumstabilizercodes,
      title={3-d quantum stabilizer codes with a power law energy barrier}, 
      author={Kamil Michnicki},
      year={2012},
      eprint={1208.3496},
      archivePrefix={arXiv},
      primaryClass={quant-ph},
      url={https://arxiv.org/abs/1208.3496}, 
}

@misc{cross2024improvedqldpcsurgerylogical,
      title={Improved QLDPC Surgery: Logical Measurements and Bridging Codes}, 
      author={Andrew Cross and Zhiyang He and Patrick Rall and Theodore Yoder},
      year={2024},
      eprint={2407.18393},
      archivePrefix={arXiv},
      primaryClass={quant-ph},
      url={https://arxiv.org/abs/2407.18393}, 
}

@misc{zhang2024timeefficientlogicaloperationsquantum,
      title={Time-efficient logical operations on quantum LDPC codes}, 
      author={Guo Zhang and Ying Li},
      year={2024},
      eprint={2408.01339},
      archivePrefix={arXiv},
      primaryClass={quant-ph},
      url={https://arxiv.org/abs/2408.01339}, 
}

@article{Anderson_2014,
   title={Fault-Tolerant Conversion between the Steane and Reed-Muller Quantum Codes},
   volume={113},
   ISSN={1079-7114},
   url={http://dx.doi.org/10.1103/PhysRevLett.113.080501},
   DOI={10.1103/physrevlett.113.080501},
   number={8},
   journal={Physical Review Letters},
   publisher={American Physical Society (APS)},
   author={Anderson, Jonas T. and Duclos-Cianci, Guillaume and Poulin, David},
   year={2014},
   month=aug 
}

@article{Davydova_2024,
   title={Quantum computation from dynamic automorphism codes},
   volume={8},
   ISSN={2521-327X},
   url={http://dx.doi.org/10.22331/q-2024-08-27-1448},
   DOI={10.22331/q-2024-08-27-1448},
   journal={Quantum},
   publisher={Verein zur Forderung des Open Access Publizierens in den Quantenwissenschaften},
   author={Davydova, Margarita and Tantivasadakarn, Nathanan and Balasubramanian, Shankar and Aasen, David},
   year={2024},
   month=aug, pages={1448} 
}

@misc{rains1997nonbinaryquantumcodes,
      title={Nonbinary quantum codes}, 
      author={Eric M. Rains},
      year={1997},
      eprint={quant-ph/9703048},
      archivePrefix={arXiv},
      primaryClass={quant-ph},
      url={https://arxiv.org/abs/quant-ph/9703048}, 
}

@misc{ide2024faulttolerantlogicalmeasurementshomological,
      title={Fault-tolerant logical measurements via homological measurement}, 
      author={Benjamin Ide and Manoj G. Gowda and Priya J. Nadkarni and Guillaume Dauphinais},
      year={2024},
      eprint={2410.02753},
      archivePrefix={arXiv},
      primaryClass={quant-ph},
      url={https://arxiv.org/abs/2410.02753}, 
}

@misc{stacks-project,
  author       = {The {Stacks project authors}},
  title        = {The Stacks project},
  howpublished = {\url{https://stacks.math.columbia.edu/tag/00ZX}},
  year         = {2025},
}

@misc{leverrier2022quantumtannercodes,
      title={Quantum Tanner codes}, 
      author={Anthony Leverrier and Gilles Zémor},
      year={2022},
      eprint={2202.13641},
      archivePrefix={arXiv},
      primaryClass={quant-ph},
      url={https://arxiv.org/abs/2202.13641}, 
}

@misc{heußen2024efficientfaulttolerantcodeswitching,
      title={Efficient fault-tolerant code switching via one-way transversal CNOT gates}, 
      author={Sascha Heußen and Janine Hilder},
      year={2024},
      eprint={2409.13465},
      archivePrefix={arXiv},
      primaryClass={quant-ph},
      url={https://arxiv.org/abs/2409.13465}, 
}

@inproceedings{de_beaudrap_pauli_2020,
	location = {Orange, {USA}},
	title = {Pauli Fusion: a computational model to realise quantum transformations from {ZX} terms},
	volume = {318},
	url = {http://arxiv.org/abs/1904.12817},
	doi = {10.4204/EPTCS.318},
	shorttitle = {Pauli Fusion},
	eventtitle = {{QPL} 2019},
	pages = {85--105},
	booktitle = {Proceedings of the 16th International Conference on Quantum Physics and Logic},
	publisher = {{EPTCS}},
	author = {de Beaudrap, Niel and Duncan, Ross and Horsman, Dominic and Perdrix, Simon},
	urldate = {2019-06-25},
	date = {2020-04-30},
	langid = {english},
	eprinttype = {arxiv},
	eprint = {1904.12817},
}

@misc{xu2024lettingtigercagebosonic,
      title={Letting the tiger out of its cage: bosonic coding without concatenation}, 
      author={Yijia Xu and Yixu Wang and Christophe Vuillot and Victor V. Albert},
      year={2024},
      eprint={2411.09668},
      archivePrefix={arXiv},
      primaryClass={quant-ph},
      url={https://arxiv.org/abs/2411.09668}, 
}

@inproceedings{cowtan_qudit_2022,
	location = {Oxford, United Kingdom},
	title = {Qudit lattice surgery},
	url = {http://arxiv.org/abs/2204.13228},
	eventtitle = {{QPL} 2022},
	booktitle = {{arXiv}:2204.13228 [quant-ph]},
	publisher = {{EPCTS}},
	author = {Cowtan, Alexander},
	urldate = {2022-04-29},
	date = {2022-04-27},
	eprinttype = {arxiv},
	eprint = {2204.13228},
}

@article{Paetznick_2013,
   title={Universal Fault-Tolerant Quantum Computation with Only Transversal Gates and Error Correction},
   volume={111},
   ISSN={1079-7114},
   url={http://dx.doi.org/10.1103/PhysRevLett.111.090505},
   DOI={10.1103/physrevlett.111.090505},
   number={9},
   journal={Physical Review Letters},
   publisher={American Physical Society (APS)},
   author={Paetznick, Adam and Reichardt, Ben W.},
   year={2013},
   month=aug }

@misc{williamson2024lowoverheadfaulttolerantquantumcomputation,
      title={Low-overhead fault-tolerant quantum computation by gauging logical operators}, 
      author={Dominic J. Williamson and Theodore J. Yoder},
      year={2024},
      eprint={2410.02213},
      archivePrefix={arXiv},
      primaryClass={quant-ph},
      url={https://arxiv.org/abs/2410.02213}, 
}

@misc{he2025extractorsqldpcarchitecturesefficient,
      title={Extractors: QLDPC Architectures for Efficient Pauli-Based Computation}, 
      author={Zhiyang He and Alexander Cowtan and Dominic J. Williamson and Theodore J. Yoder},
      year={2025},
      eprint={2503.10390},
      archivePrefix={arXiv},
      primaryClass={quant-ph},
      url={https://arxiv.org/abs/2503.10390}, 
}

@misc{cowtan2025parallellogicalmeasurementsquantum,
      title={Parallel Logical Measurements via Quantum Code Surgery}, 
      author={Alexander Cowtan and Zhiyang He and Dominic J. Williamson and Theodore J. Yoder},
      year={2025},
      eprint={2503.05003},
      archivePrefix={arXiv},
      primaryClass={quant-ph},
      url={https://arxiv.org/abs/2503.05003}, 
}
